\documentclass[peerreview]{IEEEtran}

\usepackage{amsmath, amsfonts, amstext, bm, ntheorem}
\usepackage{algorithm, algorithmic}
\usepackage{multirow,makecell}
\usepackage{tcolorbox}
\usepackage{tikz}
\usepackage{array}
\usepackage[caption=false,font=normalsize,labelfont=sf,textfont=sf]{subfig}
\usepackage{textcomp}
\usepackage{stfloats}
\usepackage{url}
\usepackage{verbatim}
\usepackage{color}
\usepackage{graphicx}
\usepackage{cite}
\usepackage{threeparttable}
\usepackage{hyperref}
\hypersetup{hypertex=true,
	colorlinks=true,
	linkcolor=blue,
	anchorcolor=blue,
	citecolor=blue}

\definecolor{red}{rgb}{1,0,0}

\newtheorem{theorem}{Theorem}
\newtheorem{lemma}{Lemma}

\newtheorem{definition}{Definition}
\newtheorem*{proof}{Proof}

\begin{document}

\title{Compact and Efficient KEMs over NTRU Lattices}

\author{Zhichuang Liang,
	\and Boyue Fang,
	\and Jieyu Zheng,
	\and Yunlei Zhao*
	
\thanks{*Corresponding author: ylzhao@fudan.edu.cn. Zhichuang Liang, Boyue Fang, Jieyu Zheng and Yunlei Zhao are with Department of Computer Science, Fudan University, Shanghai, China.}}

\markboth{Journal of \LaTeX\ Class Files,~Vol.~14, No.~8, August~2021}%
{Shell \MakeLowercase{\textit{et al.}}: A Sample Article Using IEEEtran.cls for IEEE Journals}


\maketitle
\setcounter{page}{1}

\begin{abstract}
	
The NTRU lattice is a promising candidate to construct practical cryptosystems, in particular key encapsulation mechanism (KEM),  resistant to quantum computing attacks. Nevertheless, there are still some inherent  obstacles to NTRU-based KEM schemes in having integrated performance, taking  security, bandwidth, error probability, and computational efficiency \emph{as a whole},  that is as good as and even better than their \{R,M\}LWE-based counterparts. In this work, we solve this problem by presenting a new family of NTRU-based KEM schemes, referred to as CTRU and CNTR. By bridging low-dimensional lattice codes  and  high-dimensional NTRU-lattice-based cryptography  with careful design and analysis, to the best of our knowledge CTRU and CNTR are the first NTRU-based KEM schemes with scalable ciphertext compression via only one \emph{single} ciphertext polynomial, and are the first that could outperform \{R,M\}LWE-based KEM schemes in integrated performance.  For instance, compared to Kyber that is currently the only standardized KEM by NIST, on the recommended parameter set CNTR-768 has about  $12\%$ smaller ciphertext size while encapsulating 384-bit keys compared to the fixed 256-bit key size of Kyber,   security strengthened   by $(8,7)$ bits for classical and quantum security respectively, and significantly lower error probability  ($2^{-230}$ of CNTR-768 vs. $2^{-164}$ of Kyber-768).  In particular,  CTRU and CNTR admit more flexible key sizes to be encapsulated, specifically $\frac{n}{2}$ where $n\in \{512,768,1024\}$ is the underlying polynomial dimension. In comparison with the state-of-the-art AVX2 implementation of Kyber-768,  CNTR-768 is faster by 1.9X in KeyGen, 2.6X in Encaps, and 1.2X in Decaps, respectively.  When compared to the NIST Round 3 finalist NTRU-HRSS, our CNTR-768 has about $15\%$ smaller ciphertext size, and the security is strengthened by $(55,49)$ bits for classical and quantum security respectively.  As for the AVX2 implementation, CNTR-768 is faster than NTRU-HRSS by 19X in KeyGen, 2.3X in Encaps, and 1.6X in Decaps, respectively. Along the way, we develop new techniques for more accurate error probability analysis, as well as unified implementations with respect to multiple dimensions with unified NTT methods, for NTRU-based KEM schemes over the polynomial ring $\mathbb{Z}_q[x]/(x^n-x^{n/2}+1)$, which might be of independent interest.

\end{abstract}


%
%
%
%

\section{Introduction}

Most current public-key cryptographic schemes in use,  which are based on the hardness assumptions of factoring large integers and solving (elliptic curve) discrete logarithms, will suffer from quantum attacks when  practical quantum computers are built. These cryptosystems play an important role in ensuring the confidentiality and authenticity of communications on the Internet. With the increasing cryptographic security risks of quantum computing, post-quantum cryptography (PQC) has become a research focus in recent years. There are five main types of post-quantum cryptographic schemes: hash-based, code-based, lattice-based, multivariable-based, and isogeny-based schemes, among which lattice-based cryptography is commonly viewed as amongst the most promising one due to its outstanding integrated  performance in security, communication bandwidth,  and computational  efficiency.


In the post-quantum cryptography standardization competition held by the U.S. National Institute of Standards and Technology (NIST), lattice-based schemes account for 26 out of 64  schemes in the first round~\cite{nist-round-1-submissions}, 12 out of 26 in the second round~\cite{nist-round-2-submissions}, and 7 out of 15 in the third round~\cite{nist-round-3-submissions}. Recently, NIST announced 4 candidates to be standardized~\cite{nist-to-be-standardized}, among which 3 schemes are based on lattices. Most of these lattice-based schemes are based on lattices of the following types: plain lattice and algebraically structured lattice (ideal lattice, NTRU lattice,  and module lattice). They are mainly instantiated from the following two categories of hardness assumptions. The first category consists of \emph{Learning With Errors} (LWE)~\cite{lwe-regev09} and its variants  with algebraic structures such as \emph{Ring-Learning With Errors} (RLWE)~\cite{rlwe-LPR10} and \emph{Module-Learning With Errors} (MLWE)~\cite{mlwe-LS15}, as well as the derandomized version of \{R,M\}LWE: \emph{Learning With Rounding} (LWR)~\cite{lwr-BPR12} and its variants such as \emph{Ring-Learning With Rounding} (RLWR)~\cite{lwr-BPR12} and \emph{Module-Learning With Rounding} (MLWR)~\cite{mlwr-AA16}. The second category is the \emph{NTRU} assumption~\cite{ntru-HPS98}.

NTRU was first proposed by Hoffstein, Pipher and Silverman at the rump session Crypto96~\cite{ntru-Hof96}, and it survived a lattice attack in 1997~\cite{ntru-attack-CS97}.
With some improvements on security, NTRU was published normally in 1998~\cite{ntru-HPS98}, which is named as NTRU-HPS  for presentation simplicity in this work. NTRU-HPS was the first practical public key cryptosystem based on the lattice hardness assumptions over polynomial rings, and there have been many variants of NTRU-HPS such as those proposed in \cite{ntru-variant-etru-JN15,ntru-variant-bqtru-BSP18,nttru-LS19,ntru-nist-round3,ntru-prime-nist-round3,ntru-variant-bat-FKPY22}.
Besides being survived attacks and cryptanalysis over 24 years since its introduction,  NTRU-based KEM schemes also enjoy many other desirable features. For example, 
they admit more flexible key sizes to be encapsulated (corresponding to the message space $\mathcal{M}$ in this work),  varying according to the degree of the underlying quotient polynomial. In comparison, the KEM schemes based on MLWE and MLWR like Kyber~\cite{kyber-nist-round3} and Saber~\cite{saber-nist-round3} in the NIST PQC standardization encapsulate  keys of fixed size that is restricted to  the underlying quotient polynomial that is of degree 256 for Kyber and Saber.

NTRU has played a basic role in many cryptographic protocols, e.g.,~\cite{ntrusign-HGP03,ntru-in-fhe-LATV12,ntru-in-multilinear-GGH13,ntru-in-multilinear-LSS14,ntru-in-ibe-DLP14,falcon-nist-round3}. In particular, NTRU-based schemes have achieved impressive success in the NIST PQC standardization. Specifically, Falcon signature scheme~\cite{falcon-nist-round3}, which is based on NTRU assumption, is one of the signature candidate to be standardized~\cite{nist-to-be-standardized}. NTRU KEM (including NTRU-HRSS and NTRUEncrypt)~\cite{ntru-nist-round3} is one of the seven finalists, and NTRU Prime KEM (including SNTRU Prime and NTRU LPRime)~\cite{ntru-prime-nist-round3} is one of the alternate candidates in the third round of NIST PQC standardization.
Although NTRU-based KEM schemes are not chosen to be standardized by NIST, we can not ignore their great potential in PQC research and standardization due to their attractive features.  The study and optimization of NTRU-based KEM schemes still deserve further research exploration. Actually,
some standardizations have already been including NTRU-based PKE/KEM schemes. The standard IEEE Std 1363.1~\cite{ieee1363}, which was issued in 2008, standardizes some lattice-based public-key schemes, including NTRUEncrypt. The standard X9.98~\cite{x9} standardizes NTRUEncrypt as a part of the X9 standards which are applied to the financial services industry. The European Union's PQCRYPTO project (i.e., Horizon 2020 ICT-645622)~\cite{horizon2020} is considering another NTRU variant~\cite{ntru-secure-as-ideal-lattice-SS11} as a potential European standard. In particular, the latest  updates of OpenSSH  since its 9.0 version released in April 2022 have adopted NTRU Prime, together with X25519 ECDH in a hybrid mode, to prevent ``capture now decrypt later" attacks~\cite{openssh2022}.

\subsection{Challenges and Motivations}


When considering \emph{integrated performance} in security, bandwidth, error probability, and computational  efficiency \emph{as a whole}, up to now NTRU-based KEM schemes are, \emph{in general},   inferior to  their counterparts of  \{R,M\}LWE-based KEM schemes. It might be the partial reason that NTRU-based KEM schemes were not finally standardized by NIST.   In the following, we will summarize some obstacles and challenges faced with current NTRU-based KEM schemes, which also presents  the motivations of this work.

\subsubsection{Small secret ranges}

The first inherent limitation of NTRU-based KEM schemes is that they usually  support very narrow secret ranges, typically $\{-1,0,1\}$, which inherently limits the security level achievable by NTRU-based KEM schemes~\cite{ntru-nist-round3,ntru-prime-nist-round3,nttru-LS19}. However, \{R,M\}LWE-based KEM schemes have the advantage in allowing larger secret ranges for stronger security when using the approximate moduli as in NTRU-based KEM schemes. 

\subsubsection{Large bandwidth}

\begin{figure}[!t]
	\centering
	\includegraphics[width=0.7\linewidth]{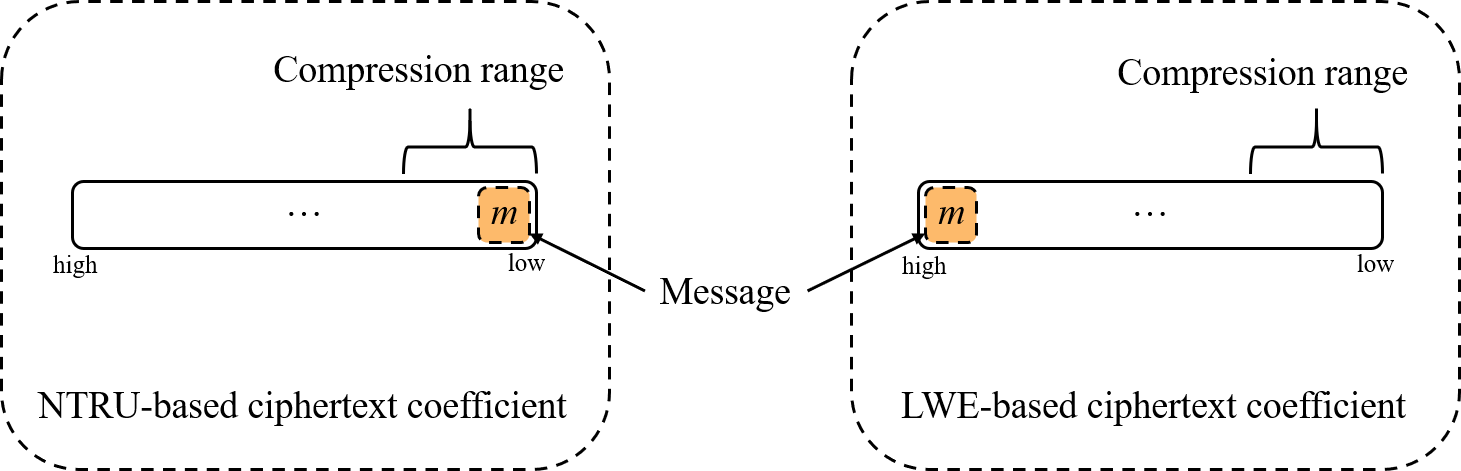}
	\caption{Differences of message positions between NTRU-based and LWE-based KEMs}
	\label{fig-message-position}
	\vspace{-0.5cm}
\end{figure}

The next limitation is that traditional NTRU-based KEM schemes commonly have larger bandwidth compared to their \{R,M\}LWE-based counterparts. The importance of reducing bandwidth is self-evident, since low communication bandwidth is  friendly to internet protocols (e.g., TLS) and to constrained internet-of-things (IoT) devices. On the one hand, traditional NTRU-based KEM schemes set relatively large moduli (together with relatively small secret ranges) in order to achieve perfect correctness. Although \{R,M\}LWE-based KEM schemes could  also choose larger moduli to have zero error probability, they prefer smaller moduli for smaller bandwidth,  tolerating negligible error probability instead of insisting on  zero error probability. On the other hand, the larger bandwidth  is  due to the inherent inability to compress ciphertexts of NTRU-based KEM schemes. Below, we briefly  explain  why ciphertext compression leads to a decryption failure at a high probability for traditional NTRU-based KEM schemes, but the impact of ciphertext compression for \{R,M\}LWE-based KEM schemes is within some  controllable range.

As shown in Figure~\ref{fig-message-position}, the initial plaintext  message $m$ is  encoded into the most significant bits of the second ciphertext term for most \{R,M\}LWE-based KEM schemes like Kyber, where their first ciphertext term (corresponding to an \{R,M\}LWE sample) is independent of the second ciphertext term. The randomness and security of the ciphertext are guaranteed by the  \{R,M\}LWE samples. On the contrary, the initial message is encoded into the least significant bits of the ciphertext for traditional NTRU-based KEM schemes. Actually, NTRU-based KEM schemes \cite{ntru-HPS98,ntru-nist-round3,ntru-secure-as-ideal-lattice-SS11,ntru-variant-eprint-1352-DHK21} have  ciphertexts of the form $c=phr+m \bmod q$, where $p$ is the message space modulus, $h$ is the public key, $r$ is the randomness, and $m$ is the message to be encrypted. In the decryption process, one could compute $cf \bmod q = pgr + m f$, and clean out the term $pgr$ via reduction modulo $p$. In order to obtain $m$, one can multiply the inverse of $f$ modulo $p$, or directly reduce modulo $p$ if $f=pf'+1$. This can be viewed as  a unidimensional error-correction mechanism.

Compressing the ciphertext of NTRU-based KEM schemes means dropping some least significant bits, which is equivalent to increasing the small error. For \{R,M\}LWE-based KEM schemes, compressing the first ciphertext term has no impact on the messages. The impact brought by reasonably compressing the second ciphertext term could be eliminated if the total error is within the capacity range of  the message-recovering mechanism. However, for traditional NTRU-based KEM schemes, ciphertext compression will destroy the useful information of the encoded messages in the least significant bits of the ciphertext. Consequently, the initial messages can not be recovered correctly.


\subsubsection{Weak starting point of security reduction}

For most NTRU-based KEM constructions, their chosen ciphertext attack (CCA) security is usually  reduced to the one-way (OW-CPA) secure encryption instead of the traditional IND-CPA secure encryption. Above all, IND-CPA security is a strictly stronger security notion than OW-CPA security. Though  OW-CPA security can be transformed into IND-CPA security, but at the price of further loosening the reduction bound particularly in the quantum random oracle model (QROM)~\cite{ntru-variant-eprint-1352-DHK21}. One can also have a tight reduction from CCA security to OW-CPA \emph{deterministic} public-key encryption (DPKE), but at the cost of a more complicated decapsulation process~\cite{ntru-nist-round3,ntru-prime-nist-round3}. More detailed discussions and clarifications on CCA security reduction of KEM in the ROM and the QROM are presented in Appendix~\ref{app-sec-reduction}. As a consequence, it is still desirable for NTRU-based KEM constructions  to have security reduction from CCA security to IND-CPA security, as is in \{R,M\}LWE-based KEM schemes.

\subsubsection{Complicated key generation}


Typically, there are only one or two polynomial multiplications in the encryption process and decryption process of NTRU-based KEM schemes, such that the  encryption process and decryption process are as efficient as (and could even be more efficient than)  those of \{R,M\}LWE-based ones. However, for most NTRU-based KEM schemes (with~\cite{nttru-LS19,ntru-variant-eprint-1352-DHK21} as exceptions), the main efficiency obstacle  is from their key generations, since there exits a complicated computation of polynomial inverse for which there  does not exit much efficient algorithms for  most of the  polynomial rings chosen  by NTRU-based KEM schemes.


Unfortunately, there are no literatures to propose such NTRU-based KEM schemes which can overcome all the obstacles mentioned above. This leads us to the following motivating question.

\begin{tcolorbox}[title=Motivating question]
	Is it possible to construct NTRU-based KEM schemes that have essentially the same or \emph{even better} integrated performance in security, bandwidth, error probability, and computational efficiency \emph{as a whole}, than \{R,M\}LWE-based KEM schemes?
\end{tcolorbox}

\subsection{Our Contributions}

Our main result shows that NTRU-based KEM schemes can practically have a remarkable integrated performance (in security, bandwidth, error probability, and computational efficiency as a whole), just as and even better than \{R,M\}LWE-based KEM schemes. In this work, we present such practical constructions, and  instantiate such  NTRU-based KEM schemes with detailed analysis.

Specifically, in this work,  we present new variants of NTRU-based cryptosystem, referred to as  CTRU and CNTR for presentation simplicity, which can allow larger secret ranges, achieve scalable ciphertext compression,
have CCA provable security reduced directly to IND-CPA security, and have fast implementations. The error probabilities of CTRU and CNTR are  low enough, which are usually   lower than those of Kyber.  They consist of IND-CPA secure public-key encryptions, named CTRU.PKE and CNTR.PKE, and IND-CCA secure key encapsulation mechanisms, named CTRU.KEM and CNTR.KEM constructed through $\text{FO}_{ID(pk),m}^{\not\bot}$ that is an enhanced variant of Fujisaki-Okamoto transformation~\cite{fo-transform-FO99,fo-transform-HHK17} with a short prefix of the public key into the hash function~\cite{fo-transform-prefix-hash-DHK+21}.

Our CTRU and CNTR    demonstrate novel approaches to constructing  NTRU-based schemes.
The descriptions of CTRU and CNTR are over  NTT-friendly rings of the form $\mathbb{Z}_q[x]/(x^n-x^{n/2}+1)$. We choose $n\in\{512,768,1024\}$ for NIST recommended security levels I, III and V, respectively,  with  the same modulus $q=3457$ (that is close to the modulus $q=3329$ of Kyber) set for all the three dimensions for ease of implementation simplicity and compatability. We may recommend the case of $n=768$ that  could have a moderate balance of post-quantum security and performance for most applications in practice. 

\subsubsection{Efficient constant-time scalable  lattice code}

Before introducing our proposed schemes, for better capability on recovering message and low enough error probability, we apply  and refine   the scalable $\text{E}_8$ lattice code  based on the works ~\cite{e8-lattice-decoding-CS82,e8-lattice-decoding-book-CS13,e8-lattice-Via18,akcn-e8-JSZ22}. As for its density, there is a remarkable mathematical breakthrough that sphere packing in the $\text{E}_8$ lattice is proved to be optimal in the sense of the best density when packing in $\mathbb{R}^{8}$~\cite{e8-lattice-Via18}. To avoid the potential timing attacks, we present constant-time encoding and decoding algorithms of the scalable $\text{E}_8$ lattice code.
 All the conditional statements are implemented by constant-time bitwise operations. Unlike most of other existing error correction codes whose constant-time implementations are  inherently difficult,  the  constant-time implementation of the scalable $\text{E}_8$ lattice code is practical and efficient. We present  the scalable $\text{E}_8$ lattice coding algorithms in section~\ref{sec-the-lattice-codes}, and give  the details about the constant-time implementations in section~\ref{sec-constant-time}.

\begin{table*}
	\centering
	\setlength{\tabcolsep}{1.5mm}
	\caption{Comparisons between CTRU, CNTR and other practical lattice-based KEM schemes. }
	\begin{tabular}{ccccccccccc}
		\hline
		                             Schemes                               &                Assumptions                 &                  Reduction                   &                                  Rings                                  & $n$  & $q$  & $|pk|$ & $|ct|$ & B.W. & (Sec.C, Sec.Q) &   $\delta$    \\ \hline
		                   \multirow{3}{*}{CTRU (Ours)}                    & \multirow{3}{*}{\makecell[c]{NTRU,\\RLWE}} & \multirow{3}{*}{\makecell[c]{IND-CPA\\RPKE}} &          \multirow{3}{*}{$\mathbb{Z}_q[x]/(x^{n}-x^{n/2}+1)$}           & 512  & 3457 &  768   &  640   & 1408 &   (118,107)    &  $2^{-143}$   \\
		                                                                   &                                            &                                              &                                                                         & 768  & 3457 &  1152  &  960   & 2112 &   (181,164)    &  $2^{-184}$   \\
		                                                                   &                                            &                                              &                                                                         & 1024 & 3457 &  1536  &  1408  & 2944 &   (255,231)    &  $2^{-195}$   \\ \hline
		                   \multirow{3}{*}{CNTR (Ours)}                    & \multirow{3}{*}{\makecell[c]{NTRU,\\RLWR}} & \multirow{3}{*}{\makecell[c]{IND-CPA\\RPKE}} &          \multirow{3}{*}{$\mathbb{Z}_q[x]/(x^{n}-x^{n/2}+1)$}           & 512  & 3457 &  768   &  640   & 1408 &   (127,115)    &  $2^{-170}$   \\
		                                                                   &                                            &                                              &                                                                         & 768  & 3457 &  1152  &  960   & 2112 &   (191,173)    &  $2^{-230}$   \\
		                                                                   &                                            &                                              &                                                                         & 1024 & 3457 &  1536  &  1280  & 2816 &   (253,230)    &  $2^{-291}$   \\ \hline
		                NTRU-HRSS~\cite{ntru-nist-round3}                  &                    NTRU                    &          \makecell[c]{OW-CPA\\DPKE}          &                       $\mathbb{Z}_q[x]/(x^{n}-1)$                       & 701  & 8192 &  1138  &  1138  & 2276 &   (136,124)    & $2^{-\infty}$ \\ \hline
		          SNTRU Prime-761~\cite{ntru-prime-nist-round3}            &                    NTRU                    &          \makecell[c]{OW-CPA\\DPKE}          &                      $\mathbb{Z}_q[x]/(x^{n}-x-1)$                      & 761  & 4591 &  1158  &  1039  & 2197 &   (153,137)    & $2^{-\infty}$ \\ \hline
		                     NTTRU~\cite{nttru-LS19}                       &                    NTRU                    &          \makecell[c]{OW-CPA\\RPKE}          &                   $\mathbb{Z}_q[x]/(x^{n}-x^{n/2}+1)$                   & 768  & 7681 &  1248  &  1248  & 2496 &   (153,140)    &  $2^{-1352}$  \\ \hline
		$\text{NTRU-C}_{3457}^{768}$~\cite{ntru-variant-eprint-1352-DHK21} &         \makecell[c]{NTRU,\\RLWE}          &         \makecell[c]{IND-CPA\\RPKE}          &                   $\mathbb{Z}_q[x]/(x^{n}-x^{n/2}+1)$                   & 768  & 3457 &  1152  &  1184  & 2336 &   (171,155)    &  $2^{-281}$   \\ \hline
		         \multirow{3}{*}{Kyber~\cite{kyber-nist-round3}}           &           \multirow{3}{*}{MLWE}            & \multirow{3}{*}{\makecell[c]{IND-CPA\\RPKE}} & \multirow{3}{*}{\makecell[c]{$\mathbb{Z}_q[x]/(x^{n/k}+1)$\\$k=2,3,4$}} & 512  & 3329 &  800   &  768   & 1568 &   (118,107)    &  $2^{-139}$   \\
		                                                                   &                                            &                                              &                                                                         & 768  & 3329 &  1184  &  1088  & 2272 &   (183,166)    &  $2^{-164}$   \\
		                                                                   &                                            &                                              &                                                                         & 1024 & 3329 &  1568  &  1568  & 3136 &   (256,232)    &  $2^{-174}$   \\ \hline
		         \multirow{3}{*}{Saber~\cite{saber-nist-round3}}           &           \multirow{3}{*}{MLWR}            & \multirow{3}{*}{\makecell[c]{IND-CPA\\RPKE}} & \multirow{3}{*}{\makecell[c]{$\mathbb{Z}_q[x]/(x^{n/k}+1)$\\$k=2,3,4$}} & 512  & 8192 &  672   &  736   & 1408 &   (118,107)    &  $2^{-120}$   \\
		                                                                   &                                            &                                              &                                                                         & 768  & 8192 &  992   &  1088  & 2080 &   (189,172)    &  $2^{-136}$   \\
		                                                                   &                                            &                                              &                                                                         & 1024 & 8192 &  1312  &  1472  & 2784 &   (260,236)    &  $2^{-165}$   \\ \hline
	\end{tabular}
	\label{tab-comparisons-of-schemes}
\end{table*}

\subsubsection{New constructions}

The key generation algorithm in CTRU is similar to the  exiting NTRU-based KEM schemes such as~\cite{ntru-HPS98,ntru-nist-round3,ntru-prime-nist-round3,ntru-variant-eprint-1352-DHK21}.  CTRU uses $h=g/f$  as its public key and $f$ as its secret key. We develop a novel  encryption algorithm which breaks through the limitation of  ciphertext compression for  NTRU-based KEM, which allows us  to compress the ciphertexts in the case of one \emph{single} polynomial. To be specific, we encode every 4-bit message into a scalable $\text{E}_8$ lattice point,  and hide its information by adding an RLWE instance, which forms  the ciphertext. This way,  the message is encoded into the most significant bits of the ciphertext, such that compressing the ciphertext does not destroy the useful information of the message. As for the decryption algorithm, we multiply the ciphertext polynomial by the secret polynomial, and finally recover the messages correctly with the aid of the decoding algorithm in the scalable $\text{E}_8$ lattice whenever the $\ell_2$ norm of the error term is less than the sphere radius of the scalable $\text{E}_8$ lattice. An important point to note here is, different from the most existing NTRU-based KEM schemes such as~\cite{ntru-HPS98,ntru-nist-round3,ntru-variant-eprint-1352-DHK21}, in CTRU  the message space modulus $p$ is removed in the public key $h$ and in the ciphertext $c$, as it is not needed there to recover the message $m$ with our CTRU construction. The only reserved position for  $p$ is the secret key $f$, which has the form of $f=pf'+1$. We show that the above steps constitute an IND-CPA secure PKE scheme: CTRU.PKE, based on the NTRU assumption and the RLWE assumption. Finally, we  apply the $\text{FO}_{ID(pk),m}^{\not\bot}$ transformation~\cite{fo-transform-prefix-hash-DHK+21} to get the IND-CCA secure CTRU.KEM. The CNTR scheme is a simplified and more efficient variant of CTRU: the noise polynomial is eliminated, and the rounding of the output of the scalable $\text{E}_8$ lattice encoding algorithm is moved. The security of CNTR is based on the NTRU assumption and the RLWR assumption.  The detailed construction of CTRU and CNTR are given in section~\ref{sec-ntru-e8-lwe-lwr-proposal}. To our knowledge, CTRU and CNTR are the first  NTRU-based KEM constructions which bridges high-dimensional NTRU-lattice-based cryptography and low-dimensional lattice codes, and are the first NTRU-based KEM schemes with scalable ciphertext compression via only one single ciphertext polynomial.


\subsubsection{Provable security}
As for security reduction, our CTRU.PKE (resp., CNTR.PKE)  can achieve the  IND-CPA security under the NTRU assumption and the RLWE (resp., RLWR) assumption, while most of the existing practical  NTRU-based PKEs only  achieve OW-CPA security. Note that, the RLWE and RLWR assumptions are only required to achieve IND-CPA security for our schemes, since CTRU.PKE and CNTR.PKE are still OW-CPA secure only based on the NTRU assumption (i.e., without further relying on the RLWE or RLWR assumptions). The reduction advantages of CCA security of our CTRU.KEM and CNTR.KEM  are tighter than  those of NTTRU~\cite{nttru-LS19} and $\text{NTRU-C}_{3457}^{768}$~\cite{ntru-variant-eprint-1352-DHK21}. For example, in the quantum setting, the CCA reduction bound of  CTRU.KEM is dominated by  $O(\sqrt{q' \epsilon_{CPA}})$,  while those of NTTRU and $\text{NTRU-C}_{3457}^{768}$ are $O(q'\sqrt{\epsilon_{OW}})$ and $O(q'^{1.5} \sqrt[4]{\epsilon_{OW}})$ respectively, where $\epsilon_{CPA}(\epsilon_{OW})$ is the advantage against the underlying IND-CPA (resp., OW-CPA) secure  PKE and $q'$ is the total query number.
However, NTRU-HRSS~\cite{ntru-variant-HRSS17,ntru-nist-round3} has a tight CCA reduction bound starting from OW-CPA \emph{deterministic} PKE (DPKE), at the cost of more complicated and time-consuming decryption process~\cite{ntru-nist-round3}. In any case, IND-CPA security is a strictly stronger security notion than OW-CPA security.

\subsubsection{More accurate analysis of error probability}

Previously, the work~\cite{nttru-LS19} gave a conservative estimation of the error probability, based on the worst case consisting of $\frac{3}{2}n$ terms for each polynomial product coefficient in $\mathbb{Z}_q[x]/(x^n-x^{n/2}+1)$. In this work, we derive the exact number of the terms of the polynomial product coefficient, and improve the error probability analysis developed in~\cite{nttru-LS19}, which might be of independent interest. The concrete analysis is provided in section~\ref{sec-exact-form-polynomial-product}.

\subsubsection{Performance and comparisons}

By careful evaluation and selection, we provide some parameter sets for CTRU and CNTR, and present the recommended parameter sets in section \ref{sec-parameters}. We also make a comprehensive analysis of CTRU and CNTR on the provable security, core-SVP hardness, refined gate-count estimate, dual attack, S-unit attack, BKW attack and side channel attack, etc, in section~\ref{sec-attack-analysis}. Here, we present brief comparisons between our schemes on the recommended parameters and other prominent practical NTRU-based KEM schemes: NTRU-HRSS~\cite{ntru-variant-HRSS17,ntru-nist-round3}, SNTRU Prime~\cite{ntru-prime-nist-round3}, NTTRU~\cite{nttru-LS19} and $\text{NTRU-C}_{3457}^{768}$~\cite{ntru-variant-eprint-1352-DHK21}, as well as the NIST standardized candidate Kyber~\cite{kyber-nist-round3} and the NIST Round 3 finalist Saber~\cite{saber-nist-round3}. The comparisons are summarized in Table~\ref{tab-comparisons-of-schemes}.  There, the column ``Assumptions'' refers to the underlying hardness assumptions. The column ``Reduction'' means that IND-CCA security is reduced to what kinds of CPA security, where ``IND'' (``OW'') refers to indistinguishability (resp., one-wayness) and ``RPKE'' (``DPKE'') refers to randomized (resp., deterministic) public-key encryptions. ``Rings'' refers to the underlying polynomial rings. The column ``$n$'' means the total dimension of algebraically structured lattices. ``$q$'' is the modulus. The public key sizes $|pk|$, ciphertext sizes $|ct|$, and B.W. (bandwidth, $|pk|+|ct|$)  are measured in bytes.  ``Sec.C'' and ``Sec.Q'' mean the estimated security expressed in bits in the classical and quantum setting respectively, which are gotten by the same methodology and scripts provided by Kyber, Saber, and NTRU KEM in NIST PQC Round 3, where we  minimize the target values if the two hardness problems, say NTRU and RLWE/RLWR, have different security values. The column ``$\delta$'' indicates the error probabilities, where the error probabilities of NTTRU and $\text{NTRU-C}_{3457}^{768}$ are re-tested according to the accurate measurement methodology discussed in section~\ref{sec-exact-form-polynomial-product}.


From the comparisons,  CNTR has the smallest bandwidth  and the strongest security guarantees among all the practical NTRU-based KEM schemes. For example,  when compared to the NIST Round 3 finalist NTRU-HRSS~\cite{ntru-nist-round3}, our CNTR-768 has about $15\%$ smaller ciphertext size,  and its security is strengthened   by $(55,49)$ bits for classical and quantum security, respectively. The error probabilities of CNTR are set according to the security level targeted by each set of parameters, which can be viewed as negligible in accordance with the security level.
When compared to  Kyber-768~\cite{kyber-nist-round3} that is standardized by NIST,  CNTR-768 has about  $12\%$ smaller ciphertext size,  and its security is strengthened   by $(8,7)$ bits for classical and quantum security, respectively. For all the three recommended parameter sets, the error probabilities of CNTR are significantly lower than those of Kyber (e.g., $2^{-230}$ of CNTR-768 vs. $2^{-164}$ of Kyber-768). To the best of our knowledge, CNTR is the first NTRU-based KEM that could outperform Kyber in the  integrated performance by considering security, bandwidth, error probability, and computational efficiency as a whole. We also would like to stress that we do not know how to have the well balance achieved by CTRU/CNTR by simply adjusting parameters for the existing NTRU-based KEM schemes. Another significant point is that CTRU and CNTR admit more flexible key sizes to be encapsulated, i.e., $n/2$-bit shared keys according to the polynomial rings we used, but Kyber and Saber can only encapsulate fixed 256-bit shared keys.

\subsubsection{Unified NTT}

The NTT-based polynomial operations over $\mathbb{Z}_q[x]/(x^n -x^{n/2} + 1)$ are very efficient. However, as the dimension $n$ varies with CTRU and CNTR, we have to equip  with  multiple  NTT algorithms  with different input/output lengths in accordance with   $n\in \{512, 768, 1024\}$.
This brings inconvenient issues  for  software implementations and especially for  hardware implementations. In this work, we overcome this problem by presenting  the  methodology of using a unified NTT technique to compute NTTs over $\mathbb{Z}_q[x]/(x^n-x^{n/2}+1)$ for all $n\in \{512,768,1024\}$ with $q=3457$, which might be of independent interest. Technically speaking, we split $f \in \mathbb{Z}_q[x]/(x^n-x^{n/2}+1)$ into $\alpha \in \{2,3,4\}$ sub-polynomials of lower degrees, each of which is in $\mathbb{Z}_q[x]/(x^{256}-x^{128}+1)$. We then design a 256-point unified NTT based on the ideas from~\cite{nttru-LS19,incomplete-ntt-moenck76},  and apply it to each sub-polynomial. Finally,  their intermediate NTT results are combined to generate the final results. In this case, in order to obtain the public key (the quotient of two $n$-dimension polynomials), we need to compute the inversions in the rings of the form $\mathbb{Z}_q[x]/(x^{2\alpha } - \zeta)$, where $\zeta$ is some primitive root of unity in $\mathbb{Z}_q$. We use Cramer's Rule~\cite{cramer-rule-linear-algebra-book} to compute the inverse of polynomials of low degree. More details are presented in section~\ref{sec-poly-operations-in-ntru-e8}.

\subsubsection{Implementation and benchmark}

We provide portable C implementation and optimized AVX2 implementation for CTRU-768 and CNTR-768.
 More details and discussions about the implementation can be seen in section~\ref{sec-implementation}.
We perform benchmark comparisons with the related lattice-based KEM schemes and some prominent  non-lattice-based KEM schemes. The benchmark comparisons show that the encapsulation and decapsulation algorithms of our schemes are among the most efficient. As for the optimized AVX2 implementations, CTRU-768 is faster than NTRU-HRSS by 23X in KeyGen, 2.1X in Encaps, and 1.6X in Decaps, respectively; CNTR-768 is faster than NTRU-HRSS by 19X in KeyGen, 2.3X in Encaps, and 1.6X in Decaps, respectively. When compared to the state-of-the-art AVX2 implementation of Kyber-768, CTRU-768 is faster by 2.3X in KeyGen, 2.3X in Encaps, and 1.2X in Decaps, respectively; CNTR-768 is faster by 1.9X in KeyGen, 2.6X in Encaps, and 1.2X in Decaps, respectively. The benchmark comparisons are referred to section \ref{sec-benchmark-comparison}.


\subsection{Related Work}

In recent years, many NTRU variants have been proposed. Jarvis and Nevins~\cite{ntru-variant-etru-JN15} presented a new variant of NTRU-HPS~\cite{ntru-HPS98} over
the ring of Eisenstein integers $\mathbb{Z}[\omega]/(x^n-1)$ where $\omega=e^{2 \pi i /3}$, which has smaller key sizes and faster performance than NTRU-HPS.
Bagheri et al.~\cite{ntru-variant-bqtru-BSP18} generalized NTRU-HPS over bivariate
polynomial rings of the form $(-1,-1)/(\mathbb{Z}[x,y]/(x^n-1,y^n-1))$ for stronger security and smaller public key sizes. H{\"{u}}lsing et al.~\cite{ntru-variant-HRSS17} improved NTRU-HPS in terms of speed, key size, and ciphertext size, and presented NTRU-HRSS, which was one of the finalists in NIST PQC Round 3~\cite{ntru-nist-round3}. Bernstein et al.~\cite{ntru-prime-BCLV17} proposed NTRU Prime, which aims for ``an efficient implementation of high security prime-degree large-Galois-group inert-modulus ideal-lattice-based cryptography''. It tweaks the textbook NTRU scheme to use some rings with less special structures, i.e., $\mathbb{Z}_q[x]/(x^n - x -1)$, where both $n$ and $q$ are primes.

In order to obtain better performance of NTRU encryption, Lyubashevsky and Seiler~\cite{nttru-LS19} instantiated it over $\mathbb{Z}_{7681}[x]/(x^{768} - x^{384}+1)$. Then Duman et al.~\cite{ntru-variant-eprint-1352-DHK21} generalized the rings $\mathbb{Z}_{q}[x]/(x^{n} - x^{n/2}+1)$ with various $n$ for flexible parameter selection. But all of them follow the similar structure of NTRU-HPS and do not support  ciphertext compression.

Recently, Fouque et al.~\cite{ntru-variant-bat-FKPY22} proposed a new NTRU variant named BAT. It shares many similarities with Falcon signature~\cite{falcon-nist-round3} where a trapdoor basis is required in the  secret key, which makes its key generation complicated. BAT uses two linear equations in two unknowns to recover the secret and error, without introducing the modulus $p$ to extract message. It reduces the ciphertext sizes by constructing its intermediate value as an RLWR instance (with binary secrets),  and encrypts the message via $\text{ACWC}_0$ transformation~\cite{ntru-variant-eprint-1352-DHK21}. However, $\text{ACWC}_0$ transformation consists of two terms, causing that there are some dozens of bytes in the second ciphertext term.
Another disadvantage is about the inflexibility of selecting
parameters. Since BAT applies power-of-two cyclotomics $\mathbb{Z}_q[x]/(x^n+1)$,  it is inconvenient to find an underlying cyclotomic polynomial of some particular degree up to the next power of two. For example, BAT chooses $\mathbb{Z}_q[x]/(x^{512}+1)$ and $\mathbb{Z}_q[x]/(x^{1024}+1)$ for NIST recommended security levels I and V, but lacks of parameter set for level III, which, however, is the aimed and recommended security level for  most lattice-based KEM schemes like Kyber~\cite{kyber-nist-round3} and our schemes. Although BAT has an advantage of bandwidth, its key generation is 1,000 times slower than other NTRU-based KEM schemes,  and there are some worries about its provable security based on the RLWR assumption with binary secrets which is quite a new assumption tailored for BAT. For the above  reasons, we do not make a direct comparison between our schemes and BAT.

%
%

%
%
%
%

~{}
\section{Preliminaries}\label{sec-preliminaries}

\subsection{Notations and Definitions}\label{sec-notation-definition}

Let $ \mathbb{Z} $  and $ \mathbb{R}$ be the set of rational integers and real numbers, respectively. Let $n$  and $q$ be some positive integers. Denote $ \mathbb{Z}_q = \mathbb{Z}/q\mathbb{Z} \cong \{0,1,\ldots,q-1\}$ and $ \mathbb{R}_q = \mathbb{R}/q\mathbb{R} $. Let $\mathbb{Z}_q^{\times}$ be the group of invertible elements of $ \mathbb{Z}_q$. For any $x \in \mathbb{R}$, $\lfloor x \rceil $ denotes the closest integer to $x$. We denote $\mathbb{Z}[x]/(x^{n}-x^{n/2}+1)$ and $\mathbb{Z}_q[x]/(x^{n}-x^{n/2}+1)$ by $\mathcal{R}$ and $\mathcal{R}_q$ respectively in this work. The elements in $\mathcal{R}$ or $\mathcal{R}_q$ are polynomials, which are denoted by regular font letters such as $f,g$. The polynomial, e.g., $f$, in $\mathcal{R}$ (or $\mathcal{R}_q$) can be represented in the form of power series: $f=\sum_{i=0}^{n-1}{f_i x^i}$, or in the form of vector: $f=(f_0, f_1,\ldots,f_{n-1})$, where $f_i \in \mathbb{Z}$ (or $f_i \in \mathbb{Z}_q$), $i=0,1,\ldots,n-1$. A function $\epsilon: \mathbb{N} \to [0,1] $ is negligible, if $\epsilon(\lambda) < 1/\lambda^c$ holds for any positive $c$ and sufficiently large $\lambda$. Denote a negligible function by $negl$.


\textbf{Cyclotomics.} More details about cyclotomics can be found in \cite{cyclotomic-fields}. Let $m$ be a positive integer, $\xi_m=\exp(\frac{2 \pi i}{m})$ be a $m$-th root of unity. The $m$-th cyclotomic polynomial $\Phi_m(x)$ is defined as $\Phi_m(x)=\prod_{j=1,\gcd(j,m)=1}^{m} { ( x - \xi_m^j)}$. It is a monic irreducible
polynomial of degree $\phi(m)$ in $\mathbb{Z}[x]$, where $\phi$ is the Euler function. The $m$-th cyclotomic field is $\mathbb{Q}(\xi_m)\cong \mathbb{Q}[x]/(\Phi_m(x))$ and its corresponding ring of integers is exactly $\mathbb{Z}[\xi_m] \cong \mathbb{Z}[x]/(\Phi_m(x))$. Most of cryptographic schemes based on algebraically structured lattices are defined over power-of-two cyclotomic rings,  $\mathbb{Z}[x]/(x^{n}+1)$ and $\mathbb{Z}_q[x]/(x^{n}+1)$, where $n=2^e$ is a power of two such that $x^n+1$ is the $2^{e+1}$-th cyclotomic polynomial. We use non-power-of-two cyclotomic rings $\mathbb{Z}[x]/(x^{n}-x^{n/2}+1)$ and $\mathbb{Z}_q[x]/(x^{n}-x^{n/2}+1)$, where $n=3^{l}\cdot 2^e, l \ge 0, e \ge 1$ throughout this paper and in this case $x^n - x^{n/2} + 1$ is the $3^{l+1}\cdot 2^e$-th cyclotomic polynomial.

\textbf{Modular reductions.} In this work, we expand the definition of modular reduction from $\mathbb{Z}$ to $\mathbb{R}$. For a positive number $q$, $r'=r \bmod^{\pm} q$ means that $r'$ is the representative element of $r$ in $[- \frac{q}{2} , \frac{q}{2} )$. Let $r'=r \bmod q$ denote as the representative element of $r$ in $[0, q)$.

\textbf{Sizes of elements.} Let $q$ be a positive number. For any number $w \in \mathbb{R}$, denote by $\|w\|_{q,\infty}=|w \bmod^{\pm} q|$ its $\ell_\infty$ norm. If ${w}$ is an $n$-dimension vector, then its $\ell_2$ norm is defined as $ \|{w}\|_{q,2} = \sqrt{\|w_0\|_{q,\infty}^2 + \cdots+\|w_{n-1}\|_{q,\infty}^2 } $. Notice that $ \|w\|_{q,2} = \|w\|_{q,\infty}$ holds for any number $w \in \mathbb{R}$.

\textbf{Sets and Distributions.} For a set $D$, we denote by $x \xleftarrow{\$} D$ sampling $x$ from $D$ uniformly at random. If $D$ is a probability distribution, $x \gets D$ means that $x$ is chosen according to the distribution $D$. The centered binomial distribution $B_\eta$ with respect to a positive integer $\eta$ is defined as follows: $ \text{Sample } (a_1,\ldots,a_\eta,b_1,\ldots,b_\eta) \xleftarrow{\$} \{0,1\}^{2\eta}$, and output $\sum_{i=1}^{\eta}{( a_i - b_i)} $. Sampling a polynomial $f \gets B_\eta$ means sampling each coefficient according to $B_\eta$ individually.


\subsection{Cryptographic Primitives}\label{sec-cryptographic-primitives}

A public-key encryption scheme contains PKE = (KeyGen, Enc, Dec), with a message space $\mathcal{M}$. The key generation algorithm KeyGen returns a pair of public key and secret key $(pk,sk)$. The encryption algorithm Enc takes a public key $pk$ and a message $m \in \mathcal{M}$ to produce a ciphertext $c$. Denote by Enc$(pk,m;coin)$ the encryption algorithm with an explicit randomness $coin$ if necessary. The deterministic decryption algorithm Dec takes a secret key $sk$ and a ciphertext $c$, and outputs either a message $m \in \mathcal{M}$ or a special symbol $\perp$ to indicate a rejection. The decryption error $\delta$ of PKE is defined as E[$\max_{m \in \mathcal{M}}$Pr[Dec($sk$,Enc($pk,m$)) $\neq m$]]$ < \delta$ where the expectation is taken over $(pk,sk) \gets \text{KeyGen}$ and the probability is taken over the random coins of $\text{Enc}$. The advantage of an adversary $\mathsf{A}$ against \emph{indistinguishability under chosen-plaintext attacks} (IND-CPA) for public-key encryption is defined as $\textbf{Adv}_{\text{PKE}}^{\text{IND-CPA}}(\mathsf{A})=$
\begin{align*}
	\begin{split}
		\left| \text{Pr} \left[
		b'=b:
		\begin{array}{c}
			(pk,sk) \gets \text{KeyGen}(); \\
			(m_0,m_1,s) \gets \mathsf{A}(pk);\\
			b \xleftarrow{\$}  \{0,1\}; c^* \gets \text{Enc}(pk,m_b);\\
			b' \gets \mathsf{A}(s,c^*)
		\end{array}
		\right]
		- \frac{1}{2} \right|.
	\end{split}
\end{align*}

A key encapsulation mechanism contains KEM = (KeyGen, Encaps, Decaps) with a key space $\mathcal{K}$. The key generation algorithm KeyGen returns a pair of public key and secret key $(pk,sk)$. The encapsulation algorithm Encaps takes a public key $pk$ to produce a ciphertext $c$ and a key $ K \in \mathcal{K}$. The deterministic decapsulation algorithm Decaps inputs a secret key $sk$ and a ciphertext $c$, and outputs either a key $ K \in \mathcal{K}$ or a special symbol $\perp$ indicating a rejection. The error probability $\delta$ of KEM is defined as Pr[Decaps$(sk,c) \neq K:(c,K) \gets$ Encaps($pk$)] $ < \delta$ where the probability is taken over $(pk,sk) \gets \text{KeyGen}$ and the random coins of $\text{Encaps}$. The advantage of an adversary $\mathsf{A}$ against \emph{indistinguishability under chosen-ciphertext attacks} (IND-CCA) for KEM is defined as $\textbf{Adv}_{\text{KEM}}^{\text{IND-CCA}}(\mathsf{A})= $
\begin{align*}
	\begin{split}
		\left| \text{Pr} \left[
		b'=b:
		\begin{array}{c}
			(pk,sk) \gets \text{KeyGen}(); \\
			b \xleftarrow{\$}  \{0,1\}; \\
			(c^*,K_0^*) \gets \text{Encaps}(pk);\\
			K_1^* \xleftarrow{\$}  \mathcal{K}; \\
			b' \gets \mathsf{A}^{ \text{ Decaps}( \cdot )  }(pk,c^*,K_b^*)
		\end{array}
		\right]
		- \frac{1}{2} \right|.
	\end{split}
\end{align*}

\subsection{Hardness Assumptions}

As the lattice cryptography evolved over the decades, the security of NTRU and its variants can be naturally viewed as two assumptions. One is the \emph{NTRU} assumption~\cite{ntru-HPS98}, and the other is the \emph{Ring-Learning with error} (RLWE) assumption~\cite{rlwe-LPR10}, which are listed as follows. In some sense, the NTRU assumption can be viewed as a  special case of the RLWE assumption. More details about NTRU cryptosystem and its applications can be seen in the excellent survey~\cite{ntru-survey-Ste14}.

\begin{definition}[NTRU assumption~\cite{ntru-HPS98}]
	Let $\Psi$ be a distribution over a polynomial ring R. Sample $f$ and $g $ according to $\Psi$, and $f$ is invertible in R. Let $h=g/f$. The decisional NTRU assumption states that $h$ is indistinguishable from a uniformly-random element in R. More precisely, the decisional NTRU assumption  is hard if the advantage  $\textbf{Adv}_{R, \Psi}^{\text{NTRU}}(\mathsf{A})$ of any probabilistic polynomial time (PPT) adversary $\mathsf{A}$ is negligible, where $\textbf{Adv}_{R, \Psi}^{\text{NTRU}}(\mathsf{A}) =$
	\begin{align*}
		\begin{split}
			\bigg| \text{Pr} \left[
			b'=1:
			\begin{array}{c}
				f,g \leftarrow \Psi \land  f^{-1} \in R\\
				h=g/f \in R; b' \leftarrow \mathsf{A}(h)
			\end{array}
			\right]
			-
			\text{Pr} \left[ b'=1: h \xleftarrow{\$} R; b' \leftarrow \mathsf{A}(h) \right]
			\bigg|.
		\end{split}
	\end{align*}	
\end{definition}

\begin{definition}[RLWE assumption~\cite{rlwe-LPR10}]
	Let $\Psi$ be a distribution over a polynomial ring R. The (decisional) Ring-Learning with error (RLWE) assumption over R is to distinguish uniform samples $({h},c) \xleftarrow{\$} R \times R$ from samples $({h},c) \in R \times R$ where ${h} \xleftarrow{\$} R$ and $c = {h} {r} + e $ with $r,e \leftarrow \Psi$. It  is hard if the advantage  $\textbf{Adv}_{R,\Psi}^{\text{RLWE}}(\mathsf{A})$ of any probabilistic polynomial time adversary $\mathsf{A}$ is negligible, where $\textbf{Adv}_{R,\Psi}^{\text{RLWE}}(\mathsf{A}) =$
	\begin{align*}
		\begin{split}
			\Bigg| \text{Pr} \left[
			b'=1:
			\begin{array}{c}
				{h} \xleftarrow{\$} R; r,e \leftarrow \Psi; \\
				c = {h} {r} + e \in R ; b' \gets \mathsf{A}({h}, {c})
			\end{array}
			\right]
			-
			\text{Pr} \left[ b'=1: {h} \xleftarrow{\$} R; {c} \xleftarrow{\$} R; b' \gets \mathsf{A}({h}, {c}) \right]
			\Bigg|.
		\end{split}
	\end{align*}	
\end{definition}

\begin{definition}[RLWR assumption~\cite{lwr-BPR12}]
	Let $q > p \ge 2$ be integers. Let $\Psi$ be a distribution over a polynomial ring R. Let $R_q=R/qR$ and $R_p=R/pR$ be the quotient rings. The (decisional) Ring-Learning with rounding (RLWR) assumption is to distinguish uniform samples $({h},c) \xleftarrow{\$} R_q \times R_p$ from samples $({h},c) \in R_q \times R_p$ where ${h} \xleftarrow{\$} R_q$ and $c = \lfloor \frac{p}{q} {h} {r} \rceil \bmod p$ with $r \leftarrow \Psi$. It  is hard if the advantage  $\textbf{Adv}_{R,\Psi}^{\text{RLWR}}(\mathsf{A})$ of any probabilistic polynomial time adversary $\mathsf{A}$ is negligible, where $\textbf{Adv}_{R,\Psi}^{\text{RLWR}}(\mathsf{A}) =$
	\begin{align*}
		\begin{split}
			\Bigg| \text{Pr} \left[
			b'=1:
			\begin{array}{c}
				{h} \xleftarrow{\$} R_q; r \leftarrow \Psi; \\
				c = \lfloor \frac{p}{q} {h} {r} \rceil \bmod p \in R_p ; b' \gets \mathsf{A}({h}, {c})
			\end{array}
			\right]
			-
			\text{Pr} \left[ b'=1: {h} \xleftarrow{\$} R_q; {c} \xleftarrow{\$} R_p; b' \gets \mathsf{A}({h}, {c}) \right]
			\Bigg|.
		\end{split}
	\end{align*}	
\end{definition}

\subsection{Number Theoretic Transform}

From a computational point of view, the fundamental and also time-consuming operations in NTRU-based schemes are the multiplications and divisions of the elements in the rings $\mathbb{Z}_q[x]/(\Phi(x))$. Number theoretic transform (NTT) is a special case of fast Fourier transform (FFT) over a finite field~\cite{ntt-Pol71}. NTT is the most efficient method for computing polynomial multiplication of high degrees, due to its quasilinear complexity $O(n \log n)$. The complete NTT-based multiplication with respect to $f$ and $g$ is $INTT(NTT(f) \circ NTT(g))$, where $NTT$ is the forward transform, $INTT$ is the inverse transform and ``$\circ$'' is the point-wise multiplication.

The FFT trick~\cite{ber01} is a fast algorithm to compute NTT, via  the Chinese Remainder Theorem (CRT) in the ring form. Briefly speaking, given pairwise co-prime polynomials $g_1,g_2,\ldots,g_k$, the CRT isomorphism is that $\varphi :$
$$ \mathbb{Z}_q[x]/(g_1 g_2 \cdots g_k) \cong \mathbb{Z}_q[x]/(g_1) \times \mathbb{Z}_q[x]/(g_2) \times \cdots \times \mathbb{Z}_q[x]/(g_k)$$
along with $\varphi(f) = ( f \bmod g_1,f \bmod g_2, \ldots, f \bmod g_k)$. In the case of the classical radix-2 FFT trick step, given the isomorphism $ \mathbb{Z}_q[x]/(x^{2m} -\zeta^2) \cong \mathbb{Z}_q[x]/(x^m - \zeta) \times \mathbb{Z}_q[x]/(x^m+\zeta)$ where $\zeta$ is invertible in $\mathbb{Z}_q$, the computation of the forward FFT tirck and inverse FFT tirck can be conducted via Cooley-Tukey butterfly~\cite{ct-butterfly} and Gentleman-Sande butterfly~\cite{gs-butterfly}, respectively. The former indicates the computation from $(f_i,f_j)$ to $(f_i + \zeta \cdot f_j, f_i - \zeta \cdot f_j)$, while the later indicates the computation from $(f'_i,f'_j)$ to $(f'_i+f'_j, (f'_i-f'_j) \cdot \zeta^{-1})$. As for the classical radix-3 FFT trick step, it is more complicated, given the isomorphism $ \mathbb{Z}_q[x]/(x^{3m} -\zeta^3) \cong \mathbb{Z}_q[x]/(x^m - \zeta) \times \mathbb{Z}_q[x]/(x^m -\rho \zeta) \times \mathbb{Z}_q[x]/(x^m - \rho^2 \zeta)$  where $\zeta$ is invertible in $\mathbb{Z}_q$ and $\rho$ is the third root of unity. The mixed-radix NTT means that there are more than one type of FFT trick step.

%
%
%
%

\section{The Lattice Code}\label{sec-the-lattice-codes}

Before introducing our proposed NTRU-based KEM schemes, we present a simple and efficient lattice code. The motivation is that a dense lattice with efficient decoding algorithm is needed in our construction for better efficiency on recovering message and low enough error probability. The coding algorithms should satisfy the following  conditions.

\begin{itemize}
	\item The operations should be  simple enough, and can be implemented by efficient arithmetic (better for integer-only operations).	
	\item The implementations of the coding algorithms are constant-time to avoid timing attacks.
	\item The decoding bound is large enough such that it leads to a high fault-tolerant mechanism.
\end{itemize}


We note that an 8-dimension lattice, named $\text{E}_8$ lattice (see~\cite[Chapter 4]{e8-lattice-decoding-book-CS13}) could satisfy the above requirements to some extent.
As for its density, there is a remarkable mathematical breakthrough that sphere packing in the $\text{E}_8$ lattice is proved to be optimal in the sense of the best density when packing in $\mathbb{R}^{8}$~\cite{e8-lattice-Via18}. As for the efficiency on coding, there have been simple executable encoding and decoding algorithms of the $\text{E}_8$ lattice in~\cite{e8-lattice-decoding-CS82,e8-lattice-decoding-book-CS13,akcn-e8-JSZ22}.
However, the known coding algorithms in~\cite{e8-lattice-decoding-CS82,e8-lattice-decoding-book-CS13} cannot be directly applied here. To work in our setting, we need to specify a one-to-one mapping from binary strings to the $\text{E}_8$ lattice points to encode messages. The work~\cite{akcn-e8-JSZ22} specifies such a mapping by choosing a basis, but its implementation of decoding algorithm is not constant-time, since in multiple places there is data flow from secret polynomials into the variables which are then used as lookup indexes.
In this work, we provide a scalable version of the $\text{E}_8$ lattice in the spirit of~\cite{akcn-e8-JSZ22} as well as a corresponding constant-time implementation of its decoding algorithm, which can also transform the lattice points to the binary strings without involving Gaussian Elimination.


\subsection{Scalable $\text{E}_8$ Lattice Code}


The scalable $\text{E}_8$ lattice is constructed from the Extended Hamming Code with respect to dimension 8, which is defined as $H_8 = \{ \mathbf{c} \in \{0,1\}^8 \mid \mathbf{c} = \mathbf{z} \mathbf{H} \bmod 2, \mathbf{z} \in \{0,1\}^4\}$
where the binary matrix $\mathbf{H}$ is
\begin{equation*}
	\mathbf{H} = \begin{bmatrix}
		1 & 1 & 1 & 1 & 0 & 0 & 0 & 0 \\
		0 & 0 & 1 & 1 & 1 & 1 & 0 & 0 \\
		0 & 0 & 0 & 0 & 1 & 1 & 1 & 1 \\
		0 & 1 & 0 & 1 & 0 & 1 & 0 & 1
	\end{bmatrix}.
\end{equation*}

Let $C = \{(x_1, x_1, x_2, x_2, x_3, x_3, x_4, x_4) \in \{0,1\}^8 \mid \sum{x_i} \equiv 0 \bmod 2 \}$, where $C$ is spanned by the up most three rows of $\mathbf{H}$. Then the scalable $\text{E}_8 $ lattice (named $\text{E}'_8$ lattice) is constructed as
$$\text{E}'_8 = \lambda \cdot [ C \cup (C + \mathbf{c}) ] \subset [0,\lambda]^8,$$
where $\mathbf{c} = (0, 1, 0, 1, 0, 1, 0, 1)$ is the last row of $\mathbf{H}$, $\lambda \in \mathbb{R}^+$ is the scale factor and $\lambda \cdot C$ means that all the elements in $C$ multiply by $\lambda$.

\subsubsection{Encoding algorithm}

The encoding algorithm of the $\text{E}'_8$ lattice (see Algorithm~\ref{algo-e8-encoding}) is to calculate $\lambda \cdot (\mathbf{k} \mathbf{H} \bmod 2)$, given a 4-bit binary string $\mathbf{k}$.

\begin{algorithm}[H]
	\caption{$\text{Encode}_{E'_8}$($\mathbf{k} \in \{0,1\}^4$)}
	\label{algo-e8-encoding}
	\begin{algorithmic}[1]
			\STATE{$ \mathbf{v} := \lambda \cdot (\mathbf{k} \mathbf{H} \bmod 2) \in [0,\lambda]^8$}
			\RETURN{$\mathbf{v}$}					
	\end{algorithmic}
\end{algorithm}

\vspace{-0.5cm}
\begin{algorithm}[H]
	\caption{$\mathsf{Decode}_{E'_8}$($\mathbf{x}=(x_0,\ldots,x_7) \in \mathbb{R}^8$)}
	\label{algo-e8-decoding-e8}
	\begin{algorithmic}[1]
			\STATE{Recall that $ \mathbf{c} := (0, 1, 0, 1, 0, 1, 0, 1)$}		
			\STATE{$(\mathbf{k}_{0}, \mathsf{TotalCost}_{0}):= \mathsf{Decode}_{C'}(\mathbf{x})$}
			\STATE{$(\mathbf{k}_{1}, \mathsf{TotalCost}_{1}):= \mathsf{Decode}_{C'}(\mathbf{x} - \lambda \cdot \mathbf{c})$}
			\STATE{$b:= \arg\min\{ \mathsf{TotalCost}_{0}, \mathsf{TotalCost}_{1}\} $}
			\STATE{$(k_0, k_1, k_2, k_3) := \mathbf{k}_{b}$}\label{line-e8-decoding-e8-output}
			\STATE{$\mathbf{k} := (k_0, k_1 \oplus k_0, k_3, b) \in \{0,1\}^4$}\label{line-e8-decoding-e8-output-tweak}
			\RETURN{$\mathbf{k}$}		
	\end{algorithmic}
\end{algorithm}

\vspace{-0.5cm}
\begin{algorithm}[H]
	\caption{$\mathsf{Decode}_{C'}$($\mathbf{x} \in \mathbb{R}^8$)}
	\label{algo-e8-decoding-c}
	\begin{algorithmic}[1]
			\STATE{$mind := +\infty$}
			\STATE{$mini := 0$}
			\STATE{$\mathsf{TotalCost} := 0$}
			\FOR{$i = 0 \dots 3$}
			\STATE{$c_0 := \|x_{2i}\|_{2\lambda,2}^2 +  \|x_{2i+1}\|_{2\lambda,2}^2$}
			\STATE{$c_1 := \|x_{2i} - \lambda\|_{2\lambda,2}^2 +  \|x_{2i+1} - \lambda \|_{2\lambda,2}^2  $}
			\STATE{$k_i := \arg\min\{ c_{0}, c_{1}\} $}
			\STATE{$\mathsf{TotalCost} := \mathsf{TotalCost} + c_{k_i}$}
			\IF {$c_{1-k_i} - c_{k_i} < mind$}
			\STATE{$mind := c_{1-k_i} - c_{k_i}$}
			\STATE{$mini := i$}
			\ENDIF
			\ENDFOR
			\IF {$ k_0+k_1+k_2+k_3 \bmod 2 = 1$}
			\STATE{$k_{mini} := 1 - k_{mini}$}
			\STATE{$\mathsf{TotalCost} := \mathsf{TotalCost} + mind$}
			\ENDIF
			\STATE{$\mathbf{k} := (k_0, k_1, k_2, k_3) \in \{0,1\}^4$}
			\RETURN {$(\mathbf{k}, \mathsf{TotalCost})$}		
	\end{algorithmic}
\end{algorithm}

\subsubsection{Decoding algorithm}

Given any $\mathbf{x} \in \mathbb{R}^8$, the decoding algorithm is to find the solution of the closest vector problem (CVP) of $\mathbf{x}$ in the $\text{E}'_8$ lattice, which is denoted by $ \lambda \cdot \mathbf{k}' \mathbf{H} \bmod 2$, and it outputs the 4-bit string $\mathbf{k}'$.  To solve the CVP of $\mathbf{x} \in \mathbb{R}^8$ in the $\text{E}'_8$ lattice, we turn to solve the CVP of $\mathbf{x}$ and $\mathbf{x}-\lambda \mathbf{c}$ in the lattice $C' = \lambda \cdot C$. The one that has smaller distance is the final answer.

We briefly introduce the idea of solving the CVP in the lattice $C'$ here. Given $\mathbf{x} \in \mathbb{R}^8$, for every two components in $\mathbf{x}$, determine whether they are close to $(0,0)$ or $(\lambda,\lambda)$. Assign the corresponding component of $\mathbf{k}$ to 0 if the former is true, and 1 otherwise. If $\sum{k_i}  \bmod 2 = 0 $ holds, it indicates that $\lambda \cdot (k_0,k_0,k_1,k_1,k_2,k_2,k_3,k_3)$ is the solution. However, $\sum{k_i} \bmod 2$  might be equal to 1. Then we choose the secondly closest vector, $\lambda \cdot (k'_0,k'_0,k'_1,k'_1,k'_2,k'_2,k'_3,k'_3)$, where there will be at most one-bit difference between $(k_0,k_1,k_2,k_3)$ and $(k'_0,k'_1,k'_2,k'_3)$. The detailed algorithm is given in Algorithm~\ref{algo-e8-decoding-e8}, along with Algorithm~\ref{algo-e8-decoding-c} as its subroutines. Note that in Algorithm~\ref{algo-e8-decoding-c}, $mind$ and $mini$ are set to store the minimal difference of the components and the corresponding index, respectively.

Finally, $\mathsf{Decode}_{C'}$ in Algorithm~\ref{algo-e8-decoding-e8} will output the 4-bit string $(k_0,k_1,k_2,k_3)$ such that the lattice point $\lambda \cdot(k_0 , k_0 \oplus b, k_1 , k_1 \oplus b,k_2 , k_2 \oplus b, k_3 , k_3 \oplus b)$ is closest to $\mathbf{x}$ in the $\text{E}'_8$ lattice. Since the lattice point has the form of $\lambda \cdot (\mathbf{k} \mathbf{H} \bmod 2)$, the decoding result $\mathbf{k}$ can be obtained by tweaking the solution of the CVP in the $\text{E}'_8$ lattice, as in line~\ref{line-e8-decoding-e8-output} and line~\ref{line-e8-decoding-e8-output-tweak} in Algorithm~\ref{algo-e8-decoding-e8}. The details about constant-time implementation of decoding algorithms are presented in section~\ref{sec-constant-time}.

\subsection{Bound of Correct Decoding}

Theorem~\ref{thm-correctness-of-e8-lattice} gives a bound of correct decoding w.r.t. Algorithm~\ref{algo-e8-decoding-e8}. Briefly speaking, for any 8-dimension vector which is close enough to the given $\text{E}'_8$ lattice point under the metric of $\ell_2$ norm, it can be decoded into the same 4-bit string that generates the lattice point. This theorem is helpful when we try to recover the targeted message from the given lattice point with error terms in our schemes.

\begin{theorem}[Correctness bound of the scalable $\text{E}_8$ lattice decoding]\label{thm-correctness-of-e8-lattice}
	For any given $\mathbf{k}_1 \in \{0,1\}^4$, denote $\mathbf{v}_1 :=\text{Encode}_{E'_8}(\mathbf{k}_1)$. For any $\mathbf{v}_2 \in \mathbb{R}^8$, denote $\mathbf{k}_2:=\text{Decode}_{E'_8}(\mathbf{v}_2)$. If $\| \mathbf{v}_2 - \mathbf{v}_1 \|_{2\lambda,2} <  \lambda$, then $\mathbf{k}_1 = \mathbf{k}_2 $.
\end{theorem}

\begin{proof}
	According to the construction of the Extended Hamming Code $H_8$, we know that its minimal Hamming distance is 4. Thus, the radius of sphere packing in the $\text{E}'_8$ lattice we used is $\frac{1}{2}\sqrt{4 \cdot \lambda^2 } = \lambda$. As shown in Algorithm~\ref{algo-e8-encoding}, $\mathbf{v}_1$ is the lattice point generated from $\mathbf{k}_1$. As for $\mathbf{v}_2 \in \mathbb{R}^8$, if $\| \mathbf{v}_2 - \mathbf{v}_1 \|_{2\lambda,2} <  \lambda$, the solution of the CVP about  $\mathbf{v}_2$ in the $\text{E}'_8$ lattice is $\mathbf{v}_1$. Since $\text{Decode}_{E'_8}$ in Algorithm~\ref{algo-e8-decoding-e8} will output the 4-bit string finally, instead of the intermediate solution of the CVP,  $\mathbf{v}_1$ is also generated from $\mathbf{k}_2$, i.e., $\mathbf{v}_1= \lambda \cdot (\mathbf{k}_2 \mathbf{H} \bmod 2) $, which indicates that $\mathbf{k}_1 = \mathbf{k}_2 $.
	$\hfill\square$
\end{proof}

%
%
%
%
%

\section{Construction and Analysis}\label{sec-ntru-e8-lwe-lwr-proposal}

In this section, we propose our two new cryptosystems based on NTRU lattice, named CTRU and CNTR, both of which contain an IND-CPA secure public-key encryption and an IND-CCA secure key encapsulation mechanism. CTRU and CNTR have similar forms of public key and secret key to those of the traditional NTRU-based KEM schemes, but the method to recover message in CTRU and CNTR is significantly different from them.
With our construction, CTRU and CNTR will achieve integrated performance in security, bandwidth, error probability and computational efficiency as a whole.

\subsection{CTRU: Proposal Description}

Our CTRU.PKE scheme is specified  in Algorithm~\ref{algo-ntru-e8-pke-keygen}-\ref{algo-ntru-e8-pke-dec}. Restate that $\mathcal{R}_{q}= \mathbb{Z}_{q}[x]/(x^n - x^{n/2} +1)$, where $n$ and $q$ are the ring parameters. Let $q_2$ be the  modulus, which is usually set to be a power of two that is smaller than $q$. Let $p$ be the message space modulus, satisfying $\gcd(q,p)=1$. We mainly focus on $p=2$ for the odd modulus $q$ in this paper. Let $\Psi_1$ and $\Psi_2$  be the distributions over $\mathcal{R}$. For presentation simplicity, the secret terms, $f^\prime,g$ are taken from $\Psi_1$, and $r,e$ are taken from $\Psi_2$. Actually, $\Psi_1$ and $\Psi_2$ can be different distributions. Let $\mathcal{M}=\{0,1\}^{n/2}$ denote the message space, where each $m \in \mathcal{M}$ can be seen as a $\frac{n}{2}$-dimension polynomial with coefficients in $\{0,1\}$.

\begin{algorithm}[H]
	\caption{CTRU.PKE.KeyGen()}
	\label{algo-ntru-e8-pke-keygen}
	\begin{algorithmic}[1]
			\STATE{$ f' , g\leftarrow \Psi_1 $}
			\STATE{$ f := p f'+1$}
			\STATE{If $f$ is not invertible in $\mathcal{R}_{q}$, restart.}
			\STATE{$h :=g/f$}
			\RETURN{$ (pk:=h,sk:=f )$}					
	\end{algorithmic}
\end{algorithm}

\vspace{-0.5cm}

\begin{algorithm}[H]
	\caption{CTRU.PKE.Enc($pk=h$, $m \in \mathcal{M}$)}
	\label{algo-ntru-e8-pke-enc}
	\begin{algorithmic}[1]
			\STATE{$ r, e \leftarrow \Psi_2 $}
			\STATE{$ \sigma := h r + e $} \label{line-ntru-e8-enc-sigma}
			\STATE{$ c := \Big\lfloor \cfrac{q_2}{q} ( \sigma +  
				\big\lfloor\text{PolyEncode}(m) \big\rceil ) \Big\rceil \bmod q_2 $}\label{line-ntru-e8-enc-ciphertext-c}
			\RETURN{$c$}					
	\end{algorithmic}
\end{algorithm}

\vspace{-0.5cm}

\begin{algorithm}[H]
	\caption{CTRU.PKE.Dec($sk=f$, $c$)}
	\label{algo-ntru-e8-pke-dec}
	\begin{algorithmic}[1]	
			\STATE{$ m := \text{PolyDecode}\left(   c f \bmod^{\pm} q_2   \right) $}
			\RETURN{$ m $}					
	\end{algorithmic}
\end{algorithm}

\vspace{-0.5cm}

\begin{algorithm}[H]
	\caption{$\text{PolyEncode}$($m =\sum\limits_{i=0}^{n/2-1}{m_i x^i} \in \mathcal{M} $) }
	\label{algo-e8-poly-encoding}
	\begin{algorithmic}[1]
			\STATE{$\text{E}'_8 :=\frac{q}{2} \cdot [ C \cup (C + \mathbf{c}) ] \subset [0,\frac{q}{2}]^8$}
			\FOR{$i = 0 \dots n/8-1$}
			\STATE{$ \mathbf{k}_i := (m_{4i},m_{4i+1},m_{4i+2},m_{4i+3}) \in \{0,1\}^4 $}
			\STATE{$ (v_{8i},v_{8i+1},\ldots,v_{8i+7}) := \text{Encode}_{E'_8}(\mathbf{k}_i) \in [0,\frac{q}{2}]^8 $}
			\ENDFOR	
			\STATE{$ {v}:= \sum\limits_{i=0}^{n-1}{v_i x^i}$}					
			\RETURN{${v}$}					
	\end{algorithmic}
\end{algorithm}

\vspace{-0.5cm}

\begin{algorithm}[H]
	\caption{$\text{PolyDecode}$($ v = \sum\limits_{i=0}^{n-1}{v_i x^i} \in \mathcal{R}_{q_2}$)}
	\label{algo-e8-poly-decoding}
	\begin{algorithmic}[1]
			\STATE{$\text{E}''_8 :=\frac{q_2}{2} \cdot [ C \cup (C + \mathbf{c}) ] \subset [0,\frac{q_2}{2}]^8$}
			\FOR{$i = 0 \dots n/8-1$}
			\STATE{$ \mathbf{x}_i := (v_{8i},v_{8i+1},\ldots,v_{8i+7})
				 \in \mathbb{R}^8 $}
			\STATE{$ (m_{4i},m_{4i+1},m_{4i+2},m_{4i+3}) := \text{Decode}_{E''_8}(\mathbf{x}_i) \in \{0,1\}^4 $}
			\ENDFOR
			\STATE{$ {m}:= \sum\limits_{i=0}^{n/2-1}{m_i x^i} \in \mathcal{M}$}											
			\RETURN{$m$}		
	\end{algorithmic}
\end{algorithm}

The $\text{PolyEncode}$ algorithm and $\text{PolyDecode}$ algorithm are described in Algorithm~\ref{algo-e8-poly-encoding} and~\ref{algo-e8-poly-decoding}, respectively. Specifically, we construct the $\text{E}'_8$ lattice with the scale factor $\frac{q}{2}$ in Algorithm~\ref{algo-e8-poly-encoding}. That is, the encoding algorithm works over $\text{E}'_8 :=\frac{q}{2} \cdot [ C \cup (C + \mathbf{c}) ]$. The $\text{PolyEncode}$ algorithm splits each $m \in \mathcal{M}$ into some quadruples, each of which will be encoded via $\text{Encode}_{E'_8}$. As for $\text{PolyDecode}$ algorithm, the decoding algorithm works over the lattice $\text{E}''_8 :=\frac{q_2}{2} \cdot [ C \cup (C + \mathbf{c}) ]$. It splits $v \in \mathcal{R}_{q_2}$ into some octets, each of which will be decoded via $\text{Decode}_{E''_8}$. The final message $m$ can be recovered by combining all the 4-bit binary strings output by $\text{Decode}_{E''_8}$.

We construct our CTRU.KEM=(Keygen, Encaps, Decaps) by applying $\text{FO}_{ID(pk),m}^{\not\bot}$, a variant of Fujisaki-Okamoto (FO) transformation~\cite{fo-transform-HHK17,fo-transform-FO99} aimed for the strengthened IND-CCA security in  multi-user setting ~\cite{fo-transform-prefix-hash-DHK+21}. Let $\iota,\gamma$ be positive integers. We prefer to choose $\iota,\gamma \ge 256$ for strong security. Let $\mathcal{H} :\{0,1\}^* \rightarrow \mathcal{K} \times \mathcal{COINS} $ be a hash function, where  $\mathcal{K}$ is the shared key space of CTRU.KEM and $\mathcal{COINS}$ is the randomness space of CTRU.PKE.Enc. Note that we make explicit the randomness in CTRU.PKE.Enc here. Define $\mathcal{H}_1(\cdot)$ as $\mathcal{H}(\cdot)$'s partial output that is mapped into $\mathcal{K}$.  Let $\mathcal{PK}$ be the public key space of CTRU.PKE. Let $ID: \mathcal{PK} \rightarrow \{0,1\}^{\gamma}$ be a fixed-output length function. The algorithms of CTRU.KEM are described in Algorithm~\ref{algo-ntru-e8-kem-keygen}-\ref{algo-ntru-e8-kem-decaps}.

\begin{algorithm}[H]
	\caption{CTRU.KEM.KeyGen()}
	\label{algo-ntru-e8-kem-keygen}
	\begin{algorithmic}[1]			
			\STATE{$(pk,sk)\leftarrow \text{CTRU.PKE.KeyGen()}$}
			\STATE{$z\xleftarrow{\$}\{0,1\}^{\iota}$}
			\RETURN{$ (pk':= pk,sk':=(sk,z))$}					
	\end{algorithmic}
\end{algorithm}

\vspace{-0.5cm}
\begin{algorithm}[H]
	\caption{CTRU.KEM.Encaps($pk$)}
	\label{algo-ntru-e8-kem-encaps}
	\begin{algorithmic}[1]
		\STATE{$m \xleftarrow{\$} \mathcal{M}$}
		\STATE{$(K,coin) := \mathcal{H}(ID(pk),m)$}		
		\STATE{$c := \text{CTRU.PKE.Enc}(pk,m;coin)$}
		\RETURN{$(c,K)$}
	\end{algorithmic}
\end{algorithm}

\vspace{-0.5cm}
\begin{algorithm}[H]
	\caption{CTRU.KEM.Decaps($(sk,z),c$)}
	\label{algo-ntru-e8-kem-decaps}
	\begin{algorithmic}[1]
		\STATE{$m':=\text{CTRU.PKE.Dec}(sk,c)$}
		\STATE{$(K',coin') := \mathcal{H}(ID(pk),m')$}
		\STATE{$\tilde{K}:= \mathcal{H}_1(ID(pk),z,c)$}						
		\IF{$m' \neq \perp$ and $ c= \text{CTRU.PKE.Enc}(pk,m';coin')$}
			 \RETURN{$K'$}
		\ELSE
			 \RETURN{$\tilde{K}$}
		\ENDIF
	\end{algorithmic}
\end{algorithm}

\subsection{CNTR: Proposal Description}

CNTR is a simple variant of CTRU, which is based on the NTRU assumption~\cite{ntru-HPS98} and the RLWR assumption~\cite{lwr-BPR12}. Here, CNTR stands for ``Compact NTRu based on RLWR". CNTR is also usually the abbreviation of container, which has the meaning CNTR is an economically concise yet powerful key encapsulation mechanism.

Our CNTR.PKE scheme is specified  in Algorithm~\ref{algo-rlwr-variant-ntru-e8-pke-keygen}-\ref{algo-rlwr-variant-ntru-e8-pke-dec}. The $\text{PolyEncode}$ algorithm and $\text{PolyDecode}$ algorithm are the same as Algorithm~\ref{algo-e8-poly-encoding} and~\ref{algo-e8-poly-decoding}, respectively. The symbols and definitions used here are the same as those of CTRU.

\begin{algorithm}[H]
	\caption{CNTR.PKE.KeyGen()}
	\label{algo-rlwr-variant-ntru-e8-pke-keygen}
	\begin{algorithmic}[1]
			\STATE{$ f' , g\leftarrow \Psi_1 $}
			\STATE{$ f := p f'+1$}
			\STATE{If $f$ is not invertible in $\mathcal{R}_{q}$, restart.}
			\STATE{$h :=g/f$}
			\RETURN{$ (pk:=h,sk:=f )$}					
	\end{algorithmic}
\end{algorithm}

\vspace{-0.5cm}

\begin{algorithm}[H]
	\caption{CNTR.PKE.Enc($pk=h$, $m \in \mathcal{M}$)}
	\label{algo-rlwr-variant-ntru-e8-pke-enc}
	\begin{algorithmic}[1]
			\STATE{$ r \leftarrow \Psi_2 $}			
			\STATE{$ \sigma := h r $} \label{line-rtru-e8-enc-sigma}
			\STATE{$ c := \Big\lfloor \cfrac{q_2}{q} ( \sigma +  
				\text{PolyEncode}(m) ) \Big\rceil \bmod q_2 $}\label{line-rlwr-variant-ntru-e8-enc-ciphertext-c}
			\RETURN{$c$}					
	\end{algorithmic}
\end{algorithm}

\vspace{-0.5cm}

\begin{algorithm}[H]
	\caption{CNTR.PKE.Dec($sk=f$, $c$)}
	\label{algo-rlwr-variant-ntru-e8-pke-dec}
	\begin{algorithmic}[1]
			\STATE{$ m := \text{PolyDecode}\left(   c f \bmod^{\pm} q_2   \right) $}
			\RETURN{$ m $}					
	\end{algorithmic}
\end{algorithm}


Unlike the encryption algorithm of CTRU (see Algorithm~\ref{algo-ntru-e8-pke-enc}), that of CNTR has the following distinctions: (1) the noise polynomial is eliminated; (2) the rounding of the PolyEncode algorithm is moved.

Our CNTR.KEM scheme is constructed in the same way as CTRU.KEM, via the FO transformation $\text{FO}_{ID(pk),m}^{\not\bot}$~\cite{fo-transform-prefix-hash-DHK+21}. The algorithms of CNTR.KEM can be referred to Algorithm~\ref{algo-ntru-e8-kem-keygen}-\ref{algo-ntru-e8-kem-decaps}.

\subsection{Correctness Analysis}\label{sec-error}

\begin{lemma}\label{lemma-ntru-e8-lwe-correctness-analysis}
	It holds that $ c f \bmod^{\pm} q_2 = \frac{q_2}{q} ( (\frac{q}{q_2} c) f \bmod^{\pm} q ) $.
\end{lemma}

\begin{proof}
	Since polynomial multiplication can be described as matrix-vector multiplication, which keeps the linearity, it holds that $(\frac{q}{q_2} c) f = \frac{q}{q_2} (c f )$.
	There exits an integral vector $\theta \in \mathbb{Z}^n$ such that $\frac{q}{q_2} c f \bmod^{\pm} q = \frac{q}{q_2} c f + q \theta $ where each component of $  \frac{q}{q_2} c f + q \theta $ is in $[-\frac{q}{2},\frac{q}{2})$. Thus, each component of $  c f + q_2 \theta$ is in $[-\frac{q_2}{2},\frac{q_2}{2})$. Hence, we obtain
	$$ c f \bmod^{\pm} q_2 = c f + q_2 \theta = \frac{q_2}{q}(\frac{q}{q_2} c f + q \theta ) = \frac{q_2}{q} ( (\frac{q}{q_2} c) f \bmod^{\pm} q ).$$
	$\hfill\square$
\end{proof}

\begin{theorem}[Correctness of CTRU]\label{thm-ntru-e8-lwe-correctness-analysis}
	Let $\Psi_1$ and $\Psi_2$ be the distributions over the ring $\mathcal{R}$, and $q,q_2$ be positive integers. Let $f',g \leftarrow \Psi_1$ and $r,e \leftarrow \Psi_2$. Let $\varepsilon \leftarrow \chi$, where $\chi$ is the distribution over $\mathcal{R}$ defined as follows: Sample $u \xleftarrow{\$} \mathcal{R}_{q}$ and output $\left[ \big\lfloor \frac{q_2}{q} u \big\rceil  -  \frac{q_2}{q} u  \right] \bmod^{\pm} q_2$. Let $\text{Err}_i$ be the $i$-th octet of $gr + ef + \vec{1} \cdot f'+\frac{q}{q_2} \varepsilon f $, where $\vec{1}$ is the polynomial with each coefficient being 1. Denote 	$1-\delta=\text{Pr}\left[ \| \text{Err}_i \|_{q,2} < \frac{q}{2} -\sqrt{2} \right] $. Then, the error probability of CTRU is $\delta$.
\end{theorem}

\begin{proof}
	Scale the $\text{E}''_8$ lattice and $c f \bmod^{\pm} q_2$ by the factor $q / q_2$. According to Lemma~\ref{lemma-ntru-e8-lwe-correctness-analysis}, we have
	\begin{align}\label{equ-ntru-e8-lwe-correctness-analysis-m}
		\begin{split}
			m &= \text{PolyDecode}_{E''_8}\left(   c f \bmod^{\pm} q_2   \right) \\
			&=  \text{PolyDecode}_{E'_8}\left(   (\frac{q}{q_2} c) f \bmod^{\pm} q  \right),
		\end{split}
	\end{align}
	in Algorithm~\ref{algo-ntru-e8-pke-dec}. For any $m \in \mathcal{M}$, the result of $\text{PolyEncode}(m)$ in Algorithm~\ref{algo-ntru-e8-pke-enc} can be denoted by $\frac{q}{2}s $ where $s \in \mathcal{R}_{2}$. Based on the hardness of the NTRU assumption and the RLWE assumption, $\sigma$ in line~\ref{line-ntru-e8-enc-sigma} in Algorithm~\ref{algo-ntru-e8-pke-enc} is pseudo-random in $\mathcal{R}_{q}$ so is $ \sigma +  \lfloor\frac{q}{2}s\rceil$ for any given $s \in \mathcal{R}_{2}$. We  mainly consider the case of odd $q$, since an even $q$ leads to a simpler proof due to $\lfloor\frac{q}{2}s\rceil=\frac{q}{2}s$.
	
	Therefore, the value of $c$ in line~\ref{line-ntru-e8-enc-ciphertext-c} in Algorithm~\ref{algo-ntru-e8-pke-enc} is
	$$ c= \Big\lfloor \cfrac{q_2}{q} ( \sigma +  \big\lfloor\cfrac{q}{2}s\big\rceil ) \Big\rceil \bmod q_2  =  \cfrac{q_2}{q} ( \sigma +  \cfrac{q+1}{2}s) + \varepsilon \bmod q_2 .$$
	
	With $\sigma = hr + e$, $h=g/f$ and $f=2 f' +1$, for the formula (\ref{equ-ntru-e8-lwe-correctness-analysis-m}) we get
	\begin{align}\label{equ-ntru-e8-lwe-correctness-analysis-whole}
		\begin{split}
			(\frac{q}{q_2} c) f \bmod^{\pm} q   &=  \frac{q}{q_2} [\cfrac{q_2}{q} ( \sigma +  \cfrac{q+1}{2}s) + \varepsilon ] \cdot f \bmod^{\pm} q  \\
			&= \cfrac{q+1}{2}s (2 f' +1) + \sigma  f + \cfrac{q}{q_2} \varepsilon f \bmod^{\pm} q\\
			&=	\cfrac{q}{2}s + gr + ef + sf'+ \frac{s}{2}+\cfrac{q}{q_2} \varepsilon f \bmod^{\pm} q
		\end{split}
	\end{align}
	
	Each octet of $\frac{q}{2}s$ in (\ref{equ-ntru-e8-lwe-correctness-analysis-whole}) is essentially a lattice point in the  $\text{E}'_8$ lattice, which we denoted by $\frac{q}{2} (\mathbf{k}_i \mathbf{H} \bmod 2)$. Denote the $i$-th octet of the polynomial $X$ by $(X)_i$. From Theorem~\ref{thm-correctness-of-e8-lattice} we know that to recover $\mathbf{k}_i$, one could hold the probability condition $\| (gr + ef + sf'+ \frac{s}{2}+\frac{q}{q_2} \varepsilon f)_i \|_{q,2} < \frac{q}{2}$ which can be indicated by the condition $ \| (gr + ef + \vec{1} \cdot f'+\frac{q}{q_2} \varepsilon f)_i \|_{q,2} + \sqrt{2} < \frac{q}{2}$, since $\vec{1} \cdot f'$ has a ``wider'' distribution than $s \cdot f'$ and $\| (\frac{s}{2})_i \|_{q,2} \le \sqrt{2}$ for any $s \in \mathcal{R}_2$. Similarly, for an even $q$, it can be simplified to the condition $\| (gr + ef +\frac{q}{q_2} \varepsilon f)_i \|_{q,2}  < \frac{q}{2}$ directly which can be implied by the inequality of the case of odd $q$. Therefore, we consider the bound of the case of odd $q$ as a general bound.
	$\hfill\square$
\end{proof}

\begin{theorem}[Correctness of CNTR]\label{lemma-ntru-e8-lwr-correctness-analysis}
	Let $\Psi_1$ and $\Psi_2$ be the distributions over the ring $\mathcal{R}$, and $q,q_2$ be positive integers. $q_2$ is an even number that is smaller than $q$. Let $f',g \leftarrow \Psi_1$ and $r \leftarrow \Psi_2$.
	Let $\varepsilon \leftarrow \chi$, where $\chi$ is the distribution over $\mathcal{R}$ defined as follows: Sample $h \xleftarrow{\$} \mathcal{R}_q$ and $r \leftarrow \Psi_2$, and output $ \left( \big\lfloor \frac{q_2}{q}  hr \big\rceil - \frac{q_2}{q}  hr \right)  \bmod^{\pm} q_2 $.	
	Let $\text{Err}_i$ be the $i$-th octet of $gr+\frac{q}{q_2}\varepsilon f $.
	Denote 	$1-\delta=\text{Pr}\left[ \| \text{Err}_i \|_{q,2} < \frac{q}{2} \right] $. Then, the error probability of CNTR is $\delta$.
\end{theorem}

\begin{proof}	
	The main observation is that the computation of the ciphertext $c$ is equivalent to
	\begin{align}
		\begin{split}\label{equ-ntru-e8-rlwr-ciphertext-computation}
			c &= \Big\lfloor \cfrac{q_2}{q} ( \sigma +
			\text{PolyEncode}(m)  ) \Big\rceil \bmod q_2 \\
			&= \Big\lfloor \cfrac{q_2}{q}  hr + \cfrac{q_2}{q}\cdot \cfrac{q}{2}s   \Big\rceil \bmod q_2 \\
			&= \Big\lfloor \cfrac{q_2}{q}  hr \Big\rceil +   \cfrac{q_2}{2}s   \bmod q_2 \\
			&= \cfrac{q_2}{q}  hr + \varepsilon +   \cfrac{q_2}{2}s  \bmod q_2
		\end{split}
	\end{align}
	for even $q_2 < q $, where $s \in \mathcal{R}_{2}$. Based on the hardness of the NTRU assumption, $h$ is pseudo-random in $\mathcal{R}_{q}$.
	The term $\big\lfloor \frac{q_2}{q}  hr \big\rceil$ indicates an RLWR sample, and the term  $\frac{q_2}{2}s $ implies the encoding output of $m$ via the scalable $\text{E}_8$ lattice w.r.t. the scale factor $\frac{q_2}{2}$.
	
	Similarly, we have
	\begin{align*}
		\begin{split}
			m &= \text{PolyDecode}_{E''_8}\left(   c f \bmod^{\pm} q_2   \right) \\
			&=  \text{PolyDecode}_{E'_8}\left(   (\frac{q}{q_2} c) f \bmod^{\pm} q  \right)
		\end{split}
	\end{align*}
	thereby $(\frac{q}{q_2} c) f \bmod^{\pm} q = \frac{q}{2}s + gr + \frac{q}{q_2}\varepsilon f \bmod^{\pm} q$. Each octet of $\frac{q}{2}s$ is essentially a lattice point in the scalable $\text{E}_8$ lattice w.r.t. the scale  factor $\frac{q}{2}$, which we denote by $\frac{q}{2} (\mathbf{k}_i \mathbf{H} \bmod 2)$. From Theorem~\ref{thm-correctness-of-e8-lattice} we know that to recover $\mathbf{k}_i$, it should hold $\| \text{Err}_i \|_{q,2} < \frac{q}{2}$, where $\text{Err}_i$ is the $i$-th octet of $gr+\frac{q}{q_2} \varepsilon f $.
	$\hfill\square$
\end{proof}

\subsection{More Accurate Form of Polynomial Product and Error Probability Analysis over $\mathbb{Z}_q[x]/(x^n - x^{n/2} + 1)$}\label{sec-exact-form-polynomial-product}

The more accurate form of polynomial product and more accurate corresponding analysis of error probability over the ring $\mathbb{Z}_q[x]/(x^n - x^{n/2} + 1 )$ are presented in detail here. As previously described in~\cite{nttru-LS19}, the general form of the polynomial product of $f=\sum\limits_{i=0}^{n-1}{f_i x^i}$ and $g=\sum\limits_{i=0}^{n-1}{g_i x^i}$ in the ring $\mathbb{Z}_q[x]/(x^n - x^{n/2} + 1)$ is presented via a matrix-vector multiplication

\begin{equation}\label{equ-compute-general-h-in-non-power-of-two-ring}
	h= \left[
	\begin{array}{c}
		h_{0}\\
		h_{1}\\
		\vdots\\
		h_{n-1}
	\end{array}
	\right]
	=
	\left[
	\begin{array}{cc}
		\textbf{L}-\textbf{U} & \ \ -\textbf{F}-\textbf{U} \\
		\textbf{F}+\textbf{U} & \ \ \textbf{F}+\textbf{L}
	\end{array}
	\right]
	\cdot
	\left[
	\begin{array}{c}
		g_{0}\\
		g_{1}\\
		\vdots\\
		g_{n-1}
	\end{array}
	\right] ,
\end{equation}
where $\textbf{F},\textbf{L},\textbf{U}$ are the $n/2$-dimension Toeplitz matrices as follows:
{
	\begin{align*}
		\begin{split}
			\textbf{F}
			=
			\left[
			\begin{array}{cccc}
				f_{n/2}  & f_{n/2-1} & \cdots &  f_{1}  \\
				f_{n/2+1} &  f_{n/2}  & \cdots &  f_{2}  \\
				\vdots   &  \vdots   & \ddots & \vdots  \\
				f_{n-1}  &  f_{n-2}  & \cdots & f_{n/2}
			\end{array}
			\right] , \
			\textbf{L}
			=
			\left[
			\begin{array}{cccc}
				f_{0}   &     0     & \cdots &   0    \\
				f_{1}   &   f_{0}   & \cdots &   0    \\
				\vdots   &  \vdots   & \ddots & \vdots \\
				f_{n/2-1} & f_{n/2-2} & \cdots & f_{0}
			\end{array}
			\right] , \
			\textbf{U}
			=
			\left[
			\begin{array}{cccc}
				0    & f_{n-1} & \cdots & f_{n/2+1} \\
				\vdots & \vdots  & \ddots &  \vdots   \\
				0    &    0    & \cdots &  f_{n-1}  \\
				0    &    0    & \cdots &     0
			\end{array}
			\right] .
		\end{split}
	\end{align*}
}

However, in \cite{nttru-LS19}, to bound the error probability, they consider the worst case consisting of the sums of $\frac{3}{2}n$ terms of the form $f_i g_j$ for each coefficient of $h$, which will give the most conservative estimation result. In the following, we will derive the exact number of the terms of the polynomial product coefficient, and improve the original error probability analysis developed in~\cite{nttru-LS19}, instead of roughly considering the worst case of using $\frac{3}{2}n$ terms. Firstly, focusing on the arithmetic operations in $\mathbb{Z}_q$, the product of $f$ and $g$ in $\mathbb{Z}_q[x]$ is written as

$$ \sum_{i+j=k, \atop 0 \leq k\leq n-1}f_ig_jx^k + \sum_{i+j=k, \atop n \leq k \leq 2n-2}f_ig_jx^k .$$

To obtain the result in $\mathbb{Z}_q[x]/(x^n - x^{n/2} + 1)$, we consider the second summation $\sum\limits_{i+j=k, \atop n \leq k \leq 2n-2}f_ig_jx^k $, and have
\begin{equation}
	\begin{aligned}
		\sum_{i+j=k, \atop n \leq k \leq 2n-2}f_ig_jx^k &= \sum_{i+j=k+n, \atop 0 \leq k \leq n-2}f_i g_j x^k(x^{\frac{n}{2}}-1)
		&=\sum_{i+j=k+\frac{n}{2}, \atop \frac{n}{2} \leq k \leq \frac{3n}{2}-2}f_ig_jx^k-\sum_{i+j=k+n, \atop 0 \leq k \leq n-2} f_i g_j x^k .\\
	\end{aligned}
\end{equation}

Then consider the form of $\sum\limits_{i+j=k+\frac{n}{2}, \atop \frac{n}{2} \leq k \leq \frac{3n}{2}-2} f_i g_j x^k$. Actually, we have
\begin{equation}
	\begin{aligned}
		\sum_{i+j=k+\frac{n}{2}, \atop \frac{n}{2} \leq k \leq \frac{3n}{2}-2}f_ig_jx^k &=\sum_{i+j=k+\frac{n}{2}, \atop \frac{n}{2} \leq k \leq n-1}f_ig_jx^k+\sum_{i+j=k+\frac{n}{2}, \atop n \leq k \leq \frac{3n}{2}-2}f_ig_jx^k \\
		&=\sum_{i+j=k+\frac{n}{2}, \atop \frac{n}{2} \leq k \leq n-1}f_ig_jx^k+\sum_{i+j=k+\frac{3n}{2}, \atop 0 \leq k \leq \frac{n}{2}-2}f_ig_jx^k(x^\frac{n}{2}-1) \\
		&=\sum_{i+j=k+\frac{n}{2}, \atop \frac{n}{2} \leq k \leq n-1}f_ig_jx^k+\sum_{i+j=k+n, \atop \frac{n}{2} \leq k \leq n-2}f_ig_jx^k
		-\sum_{i+j=k+\frac{3n}{2}, \atop 0 \leq k \leq \frac{n}{2}-2}f_ig_jx^k.
	\end{aligned}
\end{equation}

Therefore, for each coefficient $h_k$, there will be
\begin{equation}
	\begin{aligned}
		h_k =
		\sum_{i+j=k, \atop 0 \leq k\leq n-1}f_ig_j
		-\sum_{i+j=k+n, \atop 0 \leq k \leq n-2}f_ig_j
		+\sum_{i+j=k+\frac{n}{2}, \atop \frac{n}{2} \leq k \leq n-1}f_ig_j
		+\sum_{i+j=k+n, \atop \frac{n}{2} \leq k \leq n-2}f_ig_j
		-\sum_{i+j=k+\frac{3n}{2}, \atop 0 \leq k \leq \frac{n}{2}-2}f_ig_j.
	\end{aligned}
\end{equation}

Concretely, for $0 \leq k \leq \frac{n}{2}-2$, we have
$$h_k=\sum_{i+j=k}f_ig_j
-\sum_{i+j=k+n}f_ig_j-\sum_{i+j=k+\frac{3n}{2}}f_ig_j.$$

For $\frac{n}{2}-2 < k \leq \frac{n}{2}$, we have
$$h_k=
\sum_{i+j=k}f_ig_j-\sum_{i+j=k+n}f_ig_j.$$

For $\frac{n}{2} < k \leq n-1$, we have
$$h_k=
\sum_{i+j=k}f_ig_j
+\sum_{i+j=k+\frac{n}{2}}f_ig_j.$$

Note the symmetry of the distribution, here we only focus on the number of terms of the form $f_i g_j$, which are listed as follows: For $0 \leq k \leq \frac{n}{2}-2$, $h_k$ has $\frac{3n}{2}-k-1$ terms of the form $f_i g_j$; For $\frac{n}{2}-2 < k \leq \frac{n}{2}$, $h_k$ has $n$ terms of the form $f_i g_j$; For $\frac{n}{2} < k \leq n-1$, $h_k$ has $\frac{3n}{2}$ terms of the form $f_i g_j$. According to the results mentioned above, the exact number of terms of the polynomial product coefficient is provided, based on which we can compute more accurate error probabilities of our schemes, as well as NTTRU and $\text{NTRU-C}_{3457}^{768}$. Thus, the error probabilities of CTRU, CNTR, NTTRU and $\text{NTRU-C}_{3457}^{768}$ in this work are estimated by using a Python script according to our methodology. The whole results of CTRU and CNTR for the selected parameters are given in Table~\ref{tab-ntru-e8-recommended-parameter-set} and Table~\ref{tab-ntru-e8-rlwr-variant-recommended-parameter-set}.

\subsection{Provable Security}

We prove that CTRU.PKE is IND-CPA secure under the NTRU assumption and the RLWE assumption, and CNTR.PKE is IND-CPA secure under the NTRU assumption and the RLWR assumption.

\begin{theorem}[IND-CPA security of CTRU.PKE]\label{thm-ntru-e8-lwe-pke-ind-cpa-security}
	For any adversary $\mathsf{A}$, there exit adversaries $\mathsf{B}$ and $\mathsf{C}$ such that $\textbf{Adv}_{\text{CTRU.PKE}}^{\text{IND-CPA}}(\mathsf{A}) \le \textbf{Adv}_{\mathcal{R}_q, \Psi_1}^{\text{NTRU}}(\mathsf{B}) + \textbf{Adv}_{\mathcal{R}_q,\Psi_2}^{\text{RLWE}}(\mathsf{C}) $.
\end{theorem}

\begin{proof}
	We complete our proof through a sequence of games $\textbf{G}_0$, $\textbf{G}_1$ and $\textbf{G}_2$. Let $\mathsf{A}$ be the adversary against the IND-CPA security experiment. Denote by $\textbf{Succ}_i$ the event that $\mathsf{A}$ wins in the game $\textbf{G}_i$, that is, $\mathsf{A}$ outputs $b'$ such that $b'=b$ in $\textbf{G}_i$.
	
	Game $\textbf{G}_0$. This game is the original IND-CPA security experiment. Thus, $\textbf{Adv}_{\text{CTRU.PKE}}^{\text{IND-CPA}}(\mathsf{A}) = | \text{Pr}[ \textbf{Succ}_0] -1/2    |$.
	
	Game $\textbf{G}_1$. This game is the same as $\textbf{G}_0$, except that replacing the public key $h=g/f$ in the KeyGen by $ h \xleftarrow{\$} \mathcal{R}_q$. To distinguish $\textbf{G}_1$ from $\textbf{G}_0$ is equivalent to solve an NTRU problem. More precisely, there exits an adversary $\mathsf{B}$ with the same running time as that of $\mathsf{A}$ such that $| \text{Pr}[ \textbf{Succ}_0] - \text{Pr}[ \textbf{Succ}_1] | \le \textbf{Adv}_{\mathcal{R}_q, \Psi_1}^{\text{NTRU}}(\mathsf{B}) $.
	
	Game $\textbf{G}_2$. This game is the same as $\textbf{G}_1$, except that using uniformly random elements from $ \mathcal{R}_q$ to replace $\sigma$ in the encryption. Similarly, there exits an adversary $\mathsf{C}$ with the same running time as that of $\mathsf{A}$ such that $| \text{Pr}[ \textbf{Succ}_1] - \text{Pr}[ \textbf{Succ}_2] | \le \textbf{Adv}_{\mathcal{R}_q,\Psi_2}^{\text{RLWE}}(\mathsf{C}) $.
	
	In Game $\textbf{G}_2$, for any given $m_b$, according to Algorithm~\ref{algo-ntru-e8-pke-enc} and~\ref{algo-e8-poly-encoding}, $m_b$ is split into $n/8$ quadruples. Denote the $i$-th quadruple of $m_b$ as $m_b^{(i)}$, which will later be operated to output the $i$-th octet of the ciphertext $c$ that is denoted as $c^{(i)}$, $i=0,1,\ldots,n/8-1$. Since $c^{(i)}$ is only dependent on $m_b^{(i)}$ and other parts of $m_b$ do not interfere with $c^{(i)}$, our aim is to prove that $c^{(i)}$ is independent of $m_b^{(i)}$, $i=0,1,\ldots,n/8-1$. For any $i$ and any given $m_b^{(i)}$, $\lfloor \text{Encode}_{E'_8}(m_b^{(i)}) \rceil$ is fixed. Based on the uniform randomness of $\sigma$ in $ \mathcal{R}_q$ , its $i$-th octet (denoted as $\sigma^{(i)}$) is uniformly random in $\mathbb{Z}_q^{8}$, so is $\sigma^{(i)} + \lfloor \text{Encode}_{E'_8}(m_b^{(i)}) \rceil$. Therefore, the resulting $c^{(i)}$ is subject to the distribution $ \lfloor \frac{q_2}{q}u \rceil \bmod q_2$ where $u$ is uniformly random in $\mathbb{Z}_q^{8}$, which implies that $c^{(i)}$ is independent of $m_b^{(i)}$. Hence, each $c^{(i)}$ leaks no information of the corresponding $m_b^{(i)}$, $i=0,1,\ldots,n/8-1$. We have $ \text{Pr}[ \textbf{Succ}_2]  = 1/2$.
	
	Combining all the probabilities finishes the proof.
	$\hfill\square$
\end{proof}

\begin{theorem}[IND-CPA security of CNTR.PKE]
	For any adversary $\mathsf{A}$, there exit adversaries $\mathsf{B}$ and $\mathsf{C}$ such that $\textbf{Adv}_{\text{CNTR.PKE}}^{\text{IND-CPA}}(\mathsf{A}) \le \textbf{Adv}_{\mathcal{R}_q, \Psi_1}^{\text{NTRU}}(\mathsf{B}) + \textbf{Adv}_{\mathcal{R},\Psi_2}^{\text{RLWR}}(\mathsf{C}) $.
\end{theorem}

\begin{proof}
	We complete our proof through a sequence of games $\textbf{G}_0$, $\textbf{G}_1$ and $\textbf{G}_2$. Let $\mathsf{A}$ be the adversary against the IND-CPA security experiment. Denote by $\textbf{Succ}_i$ the event that $\mathsf{A}$ wins in the game $\textbf{G}_i$, that is, $\mathsf{A}$ outputs $b'$ such that $b'=b$ in $\textbf{G}_i$.
	
	Game $\textbf{G}_0$. This game is the original IND-CPA security experiment. Thus, $\textbf{Adv}_{\text{CNTR.PKE}}^{\text{IND-CPA}}(\mathsf{A}) = | \text{Pr}[ \textbf{Succ}_0] -1/2    |$.
	
	Game $\textbf{G}_1$. This game is the same as $\textbf{G}_0$, except that replacing the public key $h=g/f$ in the KeyGen by $ h \xleftarrow{\$} \mathcal{R}_q$. To distinguish $\textbf{G}_1$ from $\textbf{G}_0$ is equivalent to solve an NTRU problem. More precisely, there exits an adversary $\mathsf{B}$ with the same running time as that of $\mathsf{A}$ such that $| \text{Pr}[ \textbf{Succ}_0] - \text{Pr}[ \textbf{Succ}_1] | \le \textbf{Adv}_{\mathcal{R}_q, \Psi_1}^{\text{NTRU}}(\mathsf{B}) $.
	
	Game $\textbf{G}_2$. This game is the same as $\textbf{G}_1$, except that using random elements from $\mathcal{R}_{q_2}$ to replace $\lfloor \frac{q_2}{q}  hr \rceil$ of $c = \lfloor \frac{q_2}{q}  hr \rceil +   \frac{q_2}{2}s   \bmod q_2 $ (see the formula (\ref{equ-ntru-e8-rlwr-ciphertext-computation})) in the encryption where the term  $\frac{q_2}{2}s $ implies the encoding output of the given challenge plaintext $m_b$ via the scalable $\text{E}_8$ lattice w.r.t. the scale factor $\frac{q_2}{2}$. Similarly, there exits an adversary $\mathsf{C}$ with the same running time as that of $\mathsf{A}$ such that $| \text{Pr}[ \textbf{Succ}_1] - \text{Pr}[ \textbf{Succ}_2] | \le \textbf{Adv}_{\mathcal{R},\Psi_2}^{\text{RLWR}}(\mathsf{C}) $.
	
	In Game $\textbf{G}_2$, the information of the challenge plaintext $m_b$ is perfectly hidden by the uniformly random element from $\mathcal{R}_{q_2}$. Hence, the advantage of the adversary is zero in $\textbf{G}_2$. We have $ \text{Pr}[ \textbf{Succ}_2]  = 1/2$.
	
	Combining all the probabilities finishes the proof.
	$\hfill\square$
\end{proof}


By applying the $\text{FO}_{ID(pk),m}^{\not\bot}$ transformation and adapting the results given in \cite{fo-transform-prefix-hash-DHK+21}, we have the following results on CCA security of CTRU.KEM and CNTR.KEM in the random oracle model (ROM)~\cite{rom-BR93} and the quantum random oracle model (QROM)~\cite{qrom-BDF+11}.

\begin{theorem}[IND-CCA security in the ROM and QROM~\cite{fo-transform-prefix-hash-DHK+21}]~\label{thm-ntru-e8-kem-ind-cca-security}
	Let $\ell$ be the min-entropy~\cite{fo-transform-FO99} of $ID(pk)$, i.e., $\ell = H_{\infty}(ID(pk))$, where $(pk,sk)$$\leftarrow$CTRU/CNTR.PKE.KeyGen. For any (quantum)  adversary $\mathsf{A}$, making at most $q_D$ decapsulation queries, $q_H$ (Q)RO queries, against the IND-CCA security of CTRU/CNTR.KEM, there exits a (quantum) adversary $\mathsf{B}$ with roughly the same running time of $\mathsf{A}$, such that:
	\begin{itemize}
		\item In the ROM, it holds that $\textbf{Adv}_{\text{CTRU/CNTR.KEM}}^{\text{IND-CCA}}(\mathsf{A}) \le$
		$$ 2\left( \textbf{Adv}_{\text{CTRU/CNTR.PKE}}^{\text{IND-CPA}}(\mathsf{B})   + \frac{q_H+1}{|\mathcal{M}|}  \right) + \frac{q_H}{2^{\iota}}  + (q_H+q_D) \delta + \frac{1}{2^{\ell}};$$
		\item In the QROM, it holds that $\textbf{Adv}_{\text{CTRU/CNTR.KEM}}^{\text{IND-CCA}}(\mathsf{A}) \le$
		
		$$ 2 \sqrt{q_{HD} \textbf{Adv}_{\text{CTRU/CNTR.PKE}}^{\text{IND-CPA}}(\mathsf{B}) }+ \frac{4 q_{HD}}{\sqrt{|\mathcal{M}|}}+ \frac{4(q_H+1)}{\sqrt{2^{\iota}}}+ 16 q_{HD}^{2}\delta + \frac{1}{|\mathcal{M}|} +\frac{1}{2^{\ell}},$$
		where $q_{HD}:=q_H + q_D+1$.
	\end{itemize}		
\end{theorem}

The detailed discussions and clarifications on CCA security reduction of KEM in the ROM and the QROM are given in Appendix~\ref{app-sec-reduction}.

\subsection{Discussions and Comparisons}

\textbf{The rings.}
As in~\cite{ntru-variant-eprint-1352-DHK21,nttru-LS19},
we can choose non-power-of-two cyclotomics $\mathbb{Z}_q[x]/(x^n-x^{n/2}+1)$ with respect to $n=3^{l}\cdot 2^e$ and prime $q$,  This type of ring  allows very fast NTT-based polynomial multiplication if the ring moduli are set to be NTT-friendly. Moreover, it also allows very  flexible parameter selection, since there are many integral $n$ of the form $3^{l}\cdot 2^e, l \ge 0, e \ge 1$.

\textbf{The message modulus.}
Note that the modulus $p$ is removed in the public key $h$  (i.e., $h = g/f$ ) and in the ciphertext $c$ of our CTRU and CNTR, for the reason that $p$ is not needed in $h$ and $c$ to recover the message $m$ in our construction.
The only reserved position of $p$ is the secret key $f$, which has the form of $f=pf'+1$. Since $\gcd(q,p)=1$ is required for NTRU-based KEM schemes, we can use $p=2$ instead of $p=3$. A smaller $p$ can lead to a lower error probability. Note that for other NTRU-based KEM schemes with power-of-two modulus $q$  as in  NTRU-HRSS~\cite{ntru-nist-round3}, $p$ is set to be $3$  because it is the smallest integer co-prime to the power-of-2 modulus.

\textbf{The decryption mechanism.}
Technically speaking, the ciphertext of NTRU-based PKE schemes~\cite{ntru-HPS98,ntru-nist-round3,ntru-secure-as-ideal-lattice-SS11,ntru-variant-eprint-1352-DHK21} has the form of $c=phr+m \bmod q$. One can recover the message $m$ through a unidimensional error-correction mechanism, after computing $cf \bmod q$. Instead, we use a multi-dimension coding mechanism. We encode each 4-bit message into a lattice point in the scalable  $\text{E}_8$ lattice. They can be recovered correctly with the aid of the scalable $\text{E}_8$ decoding algorithm if the $\ell_2$ norm of the error term is less than the sphere radius of the scalable $\text{E}_8$ lattice.

\textbf{The ciphertext compression.}
The step of compressing ciphertext in CTRU is described in line~\ref{line-ntru-e8-enc-ciphertext-c} in Algorithm~\ref{algo-ntru-e8-pke-enc}, while the step of compressing ciphertext in CNTR is described in line~\ref{line-rlwr-variant-ntru-e8-enc-ciphertext-c} in Algorithm~\ref{algo-rlwr-variant-ntru-e8-pke-enc}. Both of them can be mathematically written as $y :=\lfloor \cfrac{q_2}{q} x \rceil \bmod q_2, x\in \mathbb{Z}_q $ for each component. The sufficient condition of ciphertext compression is that $ \lceil \log(q_2) \rceil < \lceil \log(q) \rceil$, such that $y$ occupies less bits than $x$. But, the capacity of ciphertext compression is invalid under the condition of $ \lceil \log(q_2)\rceil \ge \lceil \log(q) \rceil$, especially $q_2 = q$. To the best of our knowledge, CTRU and CNTR are the first NTRU-based KEM schemes with scalable ciphertext compression via a single polynomial.
The ciphertext modulus $q_2$ is adjustable, depending on the bits to be dropped. Recall that most NTRU-based KEM schemes~\cite{ntru-HPS98,ntru-nist-round3,ntru-variant-eprint-1352-DHK21} fail to compress ciphertext due to the fact that the message information would be destroyed once the ciphertext is compressed.

\textbf{The role of noise $e$ and rounding.} We remark that the noise $e$ in line~\ref{line-ntru-e8-enc-sigma} in Algorithm~\ref{algo-ntru-e8-pke-enc} is only necessary for basing the IND-CPA security on the RLWE assumption for CTRU.PKE. As we shall see, even without $e$, CTRU.PKE is still IND-CPA secure under the RLWR assumption if the ciphertext compression exits. In some sense, CTRU.PKE degenerates into tweaking CNTR.PKE with the existing of the rounding of PolyEncode. Moreover, without the noise $e$ and the ciphertext compression, CTRU.PKE is still OW-CPA secure only based on the NTRU assumption (without further relying on the RLWE or RLWR assumption). As for CNTR.PKE, once eliminating the ciphertext compression, CNTR.PKE requires the rounding of PolyEncode, such that CNTR.PKE can be OW-CPA secure similarly based on the NTRU assumption.

\textbf{The trade-off of using lattice code.}
The scalable $\text{E}_8$ lattice code allows an efficient and constant-time implementation, as described in section~\ref{sec-the-lattice-codes} and section~\ref{sec-constant-time}. The error-correction capability of the scalable $\text{E}_8$ lattice code yields a low enough error probability, and allows wider noise distribution and ciphertext compression, such that the security of our schemes can be strengthened by about 40 bits and the ciphertext size is reduced by 15\% at least, with the mostly the same (even faster) overall running time when compared to other NTRU-based KEM schemes. Although the scalable $\text{E}_8$ lattice code sightly increases extra implementation complexity, this is likely to be a reasonable and acceptable trade-off of the security, ciphertext size and implementation complexity in many contexts.

%
%
%
%
%

\section{Concrete Hardness and Parameter Selection}\label{sec-attack-analysis}

In this section, we first estimate and select parameters for CTRU and CNTR, by applying the methodology of core-SVP hardness estimation~\cite{newhope-usenix-ADPS16}. Then, we present the  refined gate-count estimate, by using the scripts provided by Kyber and NTRU Prime in NIST PQC Round 3.  Finally, we overview and discuss some recent attacks beyond the core-SVP hardness.

\subsection{Parameter Selection with Core-SVP}

\subsubsection{Primal attack and dual attack}
Currently, for the parameters selected for  most practical lattice-based cryptosystems, the dominant attacks considered are the lattice-based primal and dual attacks.
The primal attack is to solve the \emph{unique-Short Vector Problem} (u-SVP) in the lattice by constructing an integer \emph{embedding lattice} (Kannan embedding~\cite{lattice-embedding-Kan87}, Bai-Galbraith embedding~\cite{lattice-embedding-BG14}, etc). The most common lattice reduction algorithm is the BKZ algorithm~\cite{bkz-SE94,bkz-CN11}. Given a lattice basis, the \emph{blocksize}, which we denote by $b$, is necessarily chosen to recover the short vector while running the BKZ algorithm. NTRU problem can be treated as a u-SVP instance in the NTRU lattice~\cite{ntru-attack-CS97}, while a u-SVP instance can also be constructed from the LWE problem. The dual attack~\cite{MDO08} is to solve the \emph{decisional} LWE problem, consisting of using the BKZ algorithm in the dual lattice, so as to recover part of the secret and infer the final secret vector.

\subsubsection{Core-SVP hardness}

Following the simple and conservative methodology of the core-SVP hardness developed from~\cite{newhope-usenix-ADPS16}, the best known cost of running SVP solver on $b$-dimension sublattice is $2^{0.292 b}$ for the classical case and $2^{0.265 b}$ for the quantum case. These cost models can be used for conservative estimates of the security of our schemes. Note that the number of  samples is set to be $2n$ for NTRU problem (resp., $n$ for LWE problem),
since the adversary is given such samples. We estimate the classical and quantum core-SVP hardness security of CTRU and CNTR via the Python script  from~\cite{newhope-usenix-ADPS16,kyber-BDK+18,kyber-nist-round3}. The concrete results are given in Table~\ref{tab-ntru-e8-recommended-parameter-set} and Table~\ref{tab-ntru-e8-rlwr-variant-recommended-parameter-set}.

\begin{table*}
	\centering
	\caption{Parameter sets of CTRU.}
	\begin{tabular}{cccccccccccc}
		\hline
		         Schemes           &          $n$          &          $q$          &           $q_2$           &          $(\Psi_1,\Psi_2)$           &        $|pk|$         &        $|ct|$         &         B.W.          & \makecell[c]{NTRU\\(Sec.C, Sec.Q)} & \makecell[c]{LWE, primal\\(Sec.C, Sec.Q)} & \makecell[c]{LWE, dual\\(Sec.C, Sec.Q)} &          $\delta$           \\ \hline
		\multirow{2}{*}{CTRU-512}  &          512          &         3457          &          $2^{9}$          &         $({B}_{2},{B}_{2})$          &          768          &          576          &         1344          &             (111,100)              &                 (111,100)                 &                (110,100)                &         $2^{-122}$          \\
		                           & \textcolor{red}{512}  & \textcolor{red}{3457} & \textcolor{red}{$2^{10}$} & \textcolor{red}{$({B}_{3},{B}_{3})$} & \textcolor{red}{768}  & \textcolor{red}{640}  & \textcolor{red}{1408} &     \textcolor{red}{(118,107)}     &        \textcolor{red}{(118,107)}         &       \textcolor{red}{(117,106)}        & \textcolor{red}{$2^{-143}$} \\ \hline
		\multirow{3}{*}{CTRU-768}  & \textcolor{red}{768}  & \textcolor{red}{3457} & \textcolor{red}{$2^{10}$} & \textcolor{red}{$({B}_{2},{B}_{2})$} & \textcolor{red}{1152} & \textcolor{red}{960}  & \textcolor{red}{2112} &     \textcolor{red}{(181,164)}     &        \textcolor{red}{(181,164)}         &       \textcolor{red}{(180,163)}        & \textcolor{red}{$2^{-184}$} \\
		                           &          768          &         3457          &          $3457$           &         $({B}_{3},{B}_{3})$          &         1152          &         1152          &         2304          &             (192,174)              &                 (192,174)                 &                (190,173)                &         $2^{-136}$          \\
		                           &          768          &         3457          &         $2^{11}$          &         $({B}_{3},{B}_{3})$          &         1152          &         1056          &         2208          &             (192,174)              &                 (192,174)                 &                (190,173)                &         $2^{-121}$          \\ \hline
		\multirow{3}{*}{CTRU-1024} & \textcolor{red}{1024} & \textcolor{red}{3457} & \textcolor{red}{$2^{11}$} & \textcolor{red}{$({B}_{2},{B}_{2})$} & \textcolor{red}{1536} & \textcolor{red}{1408} & \textcolor{red}{2944} &     \textcolor{red}{(255,231)}     &        \textcolor{red}{(255,231)}         &       \textcolor{red}{ (252,229)}       & \textcolor{red}{$2^{-195}$} \\
		                           &         1024          &         3457          &         $2^{10}$          &         $({B}_{2},{B}_{2})$          &         1536          &         1280          &         2816          &             (255,231)              &                 (255,231)                 &                (252,229)                &         $2^{-132}$          \\
		                           &         1024          &         3457          &          $3457$           &         $({B}_{3},{B}_{3})$          &         1536          &         1536          &         3072          &             (269,244)              &                 (269,244)                 &                (266,241)                &          $2^{-96}$          \\ \hline
	\end{tabular}
	\label{tab-ntru-e8-recommended-parameter-set}
\end{table*}

\begin{table*}
	\centering
	\caption{Parameter sets of CNTR.}
	\begin{tabular}{ccccccccccc}
		\hline
		         Schemes           &          $n$          &          $q$          &           $q_2$           &         $(\Psi_1,\Psi_2)$          &        $|pk|$         &        $|ct|$         &         B.W.          & \makecell[c]{NTRU\\(Sec.C, Sec.Q)} & \makecell[c]{RLWR\\(Sec.C, Sec.Q)} &          $\delta$           \\ \hline
		\multirow{3}{*}{CNTR-512}  &          512          &         3457          &          $2^{9}$          &         $({B}_{3},B_{3})$          &          768          &          576          &         1344          &             (118,107)              &             (117,106)              &          $2^{-99}$          \\
		                           & \textcolor{red}{512}  & \textcolor{red}{3457} & \textcolor{red}{$2^{10}$} & \textcolor{red}{$({B}_{5},B_{5})$} & \textcolor{red}{768}  & \textcolor{red}{640}  & \textcolor{red}{1408} &     \textcolor{red}{(127,115)}     &     \textcolor{red}{(127,115)}     & \textcolor{red}{$2^{-170}$} \\
		                           &         {512}         &        {3457}         &        {$2^{10}$}         &        {$({B}_{6},B_{6})$}         &         {768}         &         {640}         &        {1408}         &            {(131,119)}             &            {(131,119)}             &        {$2^{-126}$}         \\ \hline
		\multirow{2}{*}{CNTR-768}  & \textcolor{red}{768}  & \textcolor{red}{3457} & \textcolor{red}{$2^{10}$} & \textcolor{red}{$({B}_{3},B_{3})$} & \textcolor{red}{1152} & \textcolor{red}{960}  & \textcolor{red}{2112} &     \textcolor{red}{(192,174)}     &     \textcolor{red}{(191,173)}     & \textcolor{red}{$2^{-230}$} \\
		                           &         {768}         &        {3457}         &        {$2^{10}$}         &        {$({B}_{4},B_{4})$}         &        {1152}         &         {960}         &        {2112}         &            {(200,181)}             &            {(199,181)}             &        {$2^{-151}$}         \\ \hline
		\multirow{2}{*}{CNTR-1024} & \textcolor{red}{1024} & \textcolor{red}{3457} & \textcolor{red}{$2^{10}$} & \textcolor{red}{$({B}_{2},B_{2})$} & \textcolor{red}{1536} & \textcolor{red}{1280} & \textcolor{red}{2816} &     \textcolor{red}{(255,231)}     &     \textcolor{red}{(253,230)}     & \textcolor{red}{$2^{-291}$} \\
		                           &        {1024}         &        {3457}         &        {$2^{10}$}         &        {$({B}_{3},B_{3})$}         &        {1536}         &        {1280}         &        {2816}         &            {(269,244)}             &            {(267,243)}             &        {$2^{-167}$}         \\ \hline
	\end{tabular}
	\label{tab-ntru-e8-rlwr-variant-recommended-parameter-set}
\end{table*}


\subsubsection{Parameter sets}\label{sec-parameters}

The parameter sets of CTRU and CNTR are given in Table~\ref{tab-ntru-e8-recommended-parameter-set} and Table~\ref{tab-ntru-e8-rlwr-variant-recommended-parameter-set} respectively, where those in red are the recommended  parameters  also given  in  Table~\ref{tab-comparisons-of-schemes}. Though the parameters in red are marked as  recommended, we believe the other parameter sets  are still very useful in certain application scenarios. Note that in Table~\ref{tab-comparisons-of-schemes} we did not list the security against the LWE dual attack. The reason is that the LWE dual attack was considered less realistic than the primal attack, and was not taken for concrete hardness estimates in many lattice-based cryptosystems including Kyber in NIST PQC Round 3~\cite{kyber-nist-round3}. For ease of a fair comparison, the  security estimate against the LWE dual attack was not listed in Table~\ref{tab-comparisons-of-schemes}.

The ring dimension $n$ is chosen from $\{512,768,1024\}$,  corresponding  to the  targeted security levels I, III and V recommended by  NIST.
The ring modulus $q$ is set to 3457, and $q_2$ is the ciphertext modulus (also the RLWR modulus for CNTR). Recall that we fix the message space modulus $p=2$ and the underlying cyclotomic polynomial $\Phi(x) = x^n - x^{n/2} + 1$, which are  omitted  in Table~\ref{tab-ntru-e8-recommended-parameter-set} and Table~\ref{tab-ntru-e8-rlwr-variant-recommended-parameter-set}. $\Psi_1$ and $\Psi_2$ are the probability distributions which are set to be $B_{\eta}$, where $B_{\eta}$ is the centered binomial distribution w.r.t. the integer $\eta$.  The public key sizes $|pk|$, ciphertext sizes $|ct|$ and B.W. (bandwidth, $|pk|+|ct|$) are measured in terms of bytes.  ``Sec.C'' and ``Sec.Q'' mean the estimated security level expressed in bits in the classical and quantum settings respectively, where the types of NTRU attack, LWE primal attack, LWE dual attack and RLWR attack are considered.
The last column ``$\delta$'' indicates  the error probability, which is evaluated by a script according to the analysis given in section \ref{sec-error}.

We stress that our schemes enjoy a flexibility of parameter selections, but selecting these $n$'s in Table~\ref{tab-ntru-e8-recommended-parameter-set} and Table~\ref{tab-ntru-e8-rlwr-variant-recommended-parameter-set} is only for simplicity. One can also choose $n$ from $\{576,648,864,972,1152,1296\}$  which are integers of the form $3^{l}\cdot 2^e, l \ge 0, e \ge 1$. Note also that the plaintext message space of CTRU and CNTR is $\{0,1\}^{{n}/{2}}$, compared to the fixed message space $\{0,1\}^{256}$ of Kyber and Saber.
As for CNTR-512, its first parameter set has the smallest ciphertext sizes, and the third parameter set has the strongest hardness of lattice problem (say, NTRU and RLWR). They can be practically applicable in certain application scenarios which are not much sensitive to error probability. In practice, each secret key will not be used for decryption more than $2^{80}$ times during its lifetime. In this case, the relatively higher error probabilities, e.g., $2^{-99}$, does not undermine the actual security of these parameter sets in reality. The third parameter set of CNTR-512 is highly recommended due to its robust security, since the recent improvements on the attacks~\cite{BT21,GT21,MAT22} might cause the worry that other lattice-based schemes like Kyber and Saber do not achieve the claimed security goals, especially on the dimension of 512. However, it is seen that the third parameter set of CNTR-512 is possible, which has a gate complexity of $2^{163.9}$ at the memory complexity of $2^{102.7}$.

\begin{table*}
	\centering
	\setlength{\tabcolsep}{2mm}
	\caption{Gate-count estimate of CTRU parameters.}
	\begin{tabular}{cccccccc}
		\hline
		         Schemes           & $(\Psi_1,\Psi_2)$ & $d$  & $b$ & $b'$ & $\log$(gates) & $\log$(memory) & \makecell[c]{$\log$(gates)\\ by NIST} \\ \hline
		\multirow{2}{*}{CTRU-512}  &    $(B_2,B_2)$    & 1007 & 386 & 350  &     144.1     &      88.4      &         \multirow{2}{*}{143}          \\
		                           &    $(B_3,B_3)$    & 1025 & 411 & 373  &     150.9     &      93.3      &                                       \\ \hline
		\multirow{2}{*}{CTRU-768}  &    $(B_2,B_2)$    & 1467 & 634 & 583  &     214.2     &     137.9      &         \multirow{2}{*}{207}          \\
		                           &    $(B_3,B_3)$    & 1498 & 671 & 618  &     224.6     &     145.3      &                                       \\ \hline
		\multirow{2}{*}{CTRU-1024} &    $(B_2,B_2)$    & 1919 & 890 & 825  &     286.1     &     188.9      &         \multirow{2}{*}{272}          \\
		                           &    $(B_3,B_3)$    & 1958 & 939 & 871  &     299.7     &     198.5      &                                       \\ \hline
	\end{tabular}
	\label{tab-ntru-e8-gate-count-estimate}
\end{table*}


\begin{table*}
	\centering
	\setlength{\tabcolsep}{2mm}
	\caption{Gate-count estimate of CNTR parameters.}
	\begin{tabular}{ccccccccc}
		\hline
		         Schemes           &  $q_2$   & $(\Psi_1,\Psi_2)$ & $d$  & $b$ & $b'$ & $\log$(gates) & $\log$(memory) & \makecell[c]{$\log$(gates)\\ by NIST} \\ \hline
		\multirow{3}{*}{CNTR-512}  & $2^{9}$  &    $(B_3,B_3)$    & 1025 & 411 & 373  &     150.9     &      93.3      &         \multirow{3}{*}{143}          \\
		                           & $2^{10}$ &    $(B_5,B_5)$    & 1025 & 444 & 404  &     160.1     &     100.0      &                                       \\
		                           & $2^{10}$ &    $(B_6,B_6)$    & 1025 & 457 & 417  &     163.9     &     102.7      &                                       \\ \hline
		\multirow{2}{*}{CNTR-768}  & $2^{10}$ &    $(B_3,B_3)$    & 1498 & 671 & 618  &     224.6     &     145.3      &         \multirow{2}{*}{207}          \\
		                           & $2^{10}$ &    $(B_4,B_4)$    & 1521 & 699 & 644  &     232.3     &     150.8      &                                       \\ \hline
		\multirow{2}{*}{CNTR-1024} & $2^{10}$ &    $(B_2,B_2)$    & 1919 & 890 & 825  &     286.1     &     188.9      &         \multirow{2}{*}{272}          \\
		                           & $2^{10}$ &    $(B_3,B_3)$    & 1958 & 939 & 871  &     299.7     &     198.5      &                                       \\ \hline
	\end{tabular}
	\label{tab-ntru-e8-rlwr-variant-gate-count-estimate}
\end{table*}


\subsection{Refined Gate-Count Estimate}

As for the quantum gates and space complexity related to the LWE and LWR problems, 
we use the same gate number estimation method  as Kyber, Saber, NTRU KEM, and SNTRU Prime in NIST PQC Round 3. Briefly speaking, it uses the probabilistic simulation of~\cite{DSD+20} rather than  the GSA-intersect model of~\cite{newhope-usenix-ADPS16,AGV+17} to determine the BKZ blocksize $b$ for a successful attack.
And it relies on the concrete estimation for the cost of sieving in gates from~\cite{AGP+19}. It also accounts for the ``few dimensions for free'' proposed in~\cite{Duc18}, which permits to solve SVP in dimension $b$ by sieving in a somewhat smaller dimension $b_0 = b - O(b)$.  Finally, it dismisses the dual attack as realistically more expensive than the primal attack. In particular, in the dual attack, exploiting the short vectors generated by the Nearest Neighbor Search used in lattice sieving is not compatible with the ``dimension for free'' trick~\cite{Duc18}.
The scripts for these refined estimates are provided in a git branch of the leaky-LWE estimator~\cite{DSD+20}\footnote{\url{https://github.com/lducas/leaky-LWE-Estimator/tree/NIST-round3}}.

The gate-count estimate results of the parameter sets of CTRU and CNTR are shown in Table~\ref{tab-ntru-e8-gate-count-estimate} and Table~\ref{tab-ntru-e8-rlwr-variant-gate-count-estimate}, respectively.
$\Psi_1$ and $\Psi_2$ are the probability distributions. $q_2$ is the RLWR modulus for CNTR. $d$ is the optimal lattice dimension for the attack. $b$ is the BKZ blocksize. $b'$ is the sieving dimension accounting for ``dimensions for free''. Gates and memory are expressed in bits. The last column means the required $\log$(gates) values by NIST.
It is estimated in~\cite{kyber-nist-round3}
that the actual cost may not be  more than 16 bits away from this estimate in either direction.

\subsection{Attacks Beyond Core-SVP Hardness}

\subsubsection{Hybrid attack}

The works~\cite{HHH+09,nist-round-2-submissions,nist-round-3-submissions} consider the hybrid attack as  the most powerful against  NTRU-based cryptosystems. However, even with many heuristic and theoretical analysis on hybrid attack ~\cite{Wun19,BGP+16,HHH+09,How07}, so far it still fails to make significant security impact on NTRU-based cryptosystems partially due to the  memory constraints. By improving the collision attack on NTRU problem, it is suggested in ~\cite{Ngu21}  that  the   mixed attack complexity estimate used for NTRU problem is unreliable, and there are both overestimation and underestimation. Judging from the current hybrid and meet-in-the-middle (MITM) attacks on NTRU problem, there is an estimation bias in the security estimates of NTRU-based KEMs, but this bias does not make a big difference to the  claimed  security. For example, under the MITM search, the security of NTRU KEM in NIST PQC Round 3 may  be $2^{-8}$ less than the acclaimed value in the worst situation~\cite{Ngu21}.

\subsubsection{Recent advances on dual attack}
There are some recent progress on the  dual attack, and we discuss their impacts on CTRU and CNTR. Duc et al.~\cite{DAF+15} propose that fast Fourier transform (FFT) can be useful to the dual attack. As for the small coefficients of the secrets, various improvements can also be achieved~\cite{AM17,BGP+16,CJM+19}. Albrecht and Martin~\cite{AM17} propose a re-randomization and smaller-dimensional lattice reduction method, and investigate the method  for  generating  coefficients of short vectors in the dual attack. Guo and Thomas~\cite{GT21} show that the current security estimates from the primal attacks are overestimated. Espitau et al.~\cite{ETA+20} achieve a dual attack that outperforms the primal attack. These attacks can be combined with the hybrid attack proposed in~\cite{HM07} to achieve a further optimized  attack under specific parameters~\cite{SC19,overstretched-ntru-attacks-KF17,BGP+16}. Very recently, MATZOV~\cite{MAT22}  further optimizes the dual attack, and claims  that the impact of its methods  is larger than  those of  Guo and Thomas's work~\cite{GT21}. It is also mentioned in~\cite{MAT22} that the newly developed  methods might also be applicable to  NTRU-based cryptosystems (e.g., by improving the hybrid attack).    The improvements of dual attacks mentioned above have potential threats to the security of CTRU and CNTR (as well as to other cryptosystems based on algebraically structured lattices). This line of research is still actively ongoing, and there is still no mature and convincing estimate method up to now.


\subsubsection{S-unit attack}
The  basis of the S-unit attack  is the unit attack: finding a short generator. On the basis of the constant-degree algorithm proposed in~\cite{Hal05,EHK+14}, Biasse et al.~\cite{BS16} present a quantum polynomial time algorithm, which is the basis for generating  the generator used in the  unit attack and S-unit attack. Then, the  unit attack is to shorten the generator by reducing the modulus of the unit, and the idea is based on the variant of the LLL algorithm~\cite{Coh12} to reduce the size of the generator in the S-unit group. That is, it replaces $y_i$ with $y_i/\epsilon$, thereby reducing the size of $y_i$, where $y_i$ refers to the size of the generator and $\epsilon$ is the reduction factor of the modulus of the unit. The S-unit attack is briefly recalled in Appendix~\ref{app-sec-s-unit-attack}.  Campbell et al.~\cite{CGS14} consider the application of the cycloid structure to the unit attack, which mainly depends on the simple generator of the cycloid unit. Under the cycloid structure, the determinant is easy to determine, and  is larger than the logarithmic length of the private key, which means that  the private key can  be recovered  through the LLL algorithm.


After establishing a set of short vectors, the simple reduction repeatedly uses $v-u$ to replace $v$, thereby reducing the modulus of vector $v$, where $u$ belongs to the set of short vectors. This idea is discovered in~\cite{AH17,Coh12}. The difference is that the algorithm proposed by Avanzi and Howard~\cite{AH17} can be applied to any lattice, but  is limited to the $\ell_2$ norm, while the algorithm proposed by Cohen~\cite{Coh12} is applicable to more norms. Pellet-Mary et al.~\cite{PGD19} analyze the algorithm of Avanzi and Howard~\cite{AH17}, and apply  it to S-unit. They point out that the S-unit attack could  achieve shorter vectors than existing methods, but still with exponential time for an  exponentially large approximation factor.
Very recently, Bernstein and Tanja~\cite{BT21} further improve the S-unit attack.



Up to now, it is still an open problem to predict the effectiveness of the reduction inside the unit attacks. The statistical experiments on various $m'$-th cyclotomics (with respect to power-of-two $m'$) show that the efficiency of the S-unit attack is much higher than a spherical model of the same lattice for $m'\in \{128, 256, 512\}$~\cite{Ber16}. The effect is about a factor of $2^{-3}$, $2^{-6}$ and $2^{-11}$, respectively. Therefore, even with a conservative estimate, the security impact on CTRU and CNTR may  not exceed a factor of $2^{-11}$.

\subsubsection{BKW attack}
For cryptographic schemes to which the BKW method can be applied, the combined methods proposed in~\cite{BAA+03,AMC+15,GTP15,KPP15}, which extend the BKW method, can be the most efficient method for specific parameters. These methods require a large number of samples, and their security estimates are based on the analysis of lattice basis reduction, either by solving the encoding problem in the lattice or by converting to a u-SVP problem~\cite{AMR+13,LRC11,LMP13}. These attacks do not affect the security of CTRU and CNTR, because the parameters chosen for CTRU and CNTR do not meet the conditions of BKW method.

\subsubsection{Side channel attack}

Ravi et al.~\cite{REB+22} construct some ciphertexts with specific structures where the key information exists in the intermediate variables, so as to recover the key through side channel attack (SCA).
They apply this attack to NTRU KEM and NTRU Prime in NIST PQC Round 3, which can  recover the full secret keys
through a few thousands of  chosen ciphertext queries. This type of SCA-aided chosen ciphertext attack  is not directly  applicable to  CTRU and CNTR, but might be  possible to be  improved against CTRU and CNTR.

Recently, Bernstein~\cite{ber22-fault-attack} proposes an efficient fault attack with a one-time single-bit fault in the random string stored inside the secret key, such that this attack can recover all the previous NTRU-HRSS session keys with the aid of about a thousand of modified ciphertexts in the standrad IND-CCA attack model. However, Bernstein's fault attack is valid for the specific ciphertext form of NTRU-HRSS, and is invalid for compressed ciphertext (as in CTRU and CNTR). Thus, this type of fault attack does not threaten CTRU and CNTR yet.

\subsubsection{Other attacks}

Algebraic attacks~\cite{CGS14,BS16,CDP+16,CDW17} and dense sublattice attacks~\cite{overstretched-ntru-attacks-KF17} also provide new ideas for LWE-based cryptographic analysis. However, these attacks do not currently affect the acclaimed security of the proposed parameters of CTRU and CNTR.

%
%
%
%

\section{Polynomial Arithmetic}\label{sec-poly-operations-in-ntru-e8}

In this section, some NTT algorithms are introduced to compute and accelerate the polynomial multiplication or division of CTRU and CNTR. In particular, to address the inconvenient issues that multiple NTT algorithms have to be equipped in accordance with each $n\in \{512, 768, 1024\}$, we provide  the  methodology of using a unified NTT technique to compute NTT algorithms over $\mathbb{Z}_q[x]/(x^n-x^{n/2}+1)$ for all $n\in \{512,768,1024\}$ with the same $q$.

\subsection{The Mixed-radix NTT}\label{sec-mixed-radix-ntt}

A type of mixed-radix NTT is utilized to compute the polynomial multiplication and division over $\mathbb{Z}_{q}[x]/(x^{n} - x^{n/2} + 1)$ with respect to $n=768$ and $q=3457$. We choose the $\frac{3}{2}n$-th primitive root of unity $\zeta=5$ in $\mathbb{Z}_{q}$ due to $\frac{3}{2}n|(q-1)$. As for the forward  NTT transform ($NTT$),  inspired by NTTRU~\cite{nttru-LS19}, there is a mapping such that $\mathbb{Z}_q[x]/(x^{n} - x^{n /2} + 1 ) \  \cong  \  \mathbb{Z}_q[x]/( x^{n/2} - \zeta_{1})  \times \mathbb{Z}_q[x]/( x^{n/2} - \zeta_{2})$ where $\zeta_{1} + \zeta_{2}=1$ and $ \zeta_{1} \cdot \zeta_{2}=1$. To apply the mixed-radix NTT, we choose $\zeta_{1} = \zeta^{n/4} \bmod q$ and $ \zeta_{2} = \zeta_{1}^5 = \zeta^{5n/4} \bmod q$. Thus, both $x^{n/2} - \zeta_{1}$ and $x^{n/2} - \zeta_{2}$ can be recursively split down into degree-$6$ terms like $x^{6} \pm \zeta^3$ through 6 steps of radix-2 FFT trick for $n=768$. Then the steps of radix-3 FFT trick can be utilized, for example, given the isomorphism $ \mathbb{Z}_q[x]/(x^{6} -\zeta^3) \cong \mathbb{Z}_q[x]/(x^2 - \zeta) \times \mathbb{Z}_q[x]/(x^2 -\rho \zeta) \times \mathbb{Z}_q[x]/(x^3 - \rho^2 \zeta)$ where $\rho=\zeta^{n/2} \bmod q$. Therefore, $\mathbb{Z}_q[x]/(  x^{n} - x^{n /2} + 1 )$ can be decomposed into $\prod\limits_{i=0}^{n/2-1}  { \mathbb{Z}_q[x]/(  x^{2} - \zeta^{\tau(i)}  )  } $, where $\tau(i)$ is the power of $\zeta$ of the $i$-th term and we start the index $i$ from zero. Upon receiving the polynomial $f$, its result of the forward NTT transform is $\hat{f}=(\hat{f}_0,\hat{f}_1,\ldots,\hat{f}_{\frac{n}{2}-1})$ where $\hat{f}_i \in \mathbb{Z}_q[x]/(  x^{2} - \zeta^{\tau(i)})$ is a linear polynomial, $i=0,1,\ldots,\frac{n}{2}-1$.

The inverse NTT transform ($INTT$) can be obtained by inverting these procedures. In this case, the point-wise multiplication (``$\circ$'') is the corresponding linear polynomial multiplication in $\mathbb{Z}_q[x]/(  x^{2} - \zeta^{\tau(i)}),i=0,1,\ldots,\frac{n}{2}-1$.

As for the mixed-radix NTT-based polynomial multiplication with respect to $h=f \cdot g$, it is computed by $h=INTT(NTT(f) \circ NTT(g))$. In addition, as for the mixed-radix NTT-based polynomial division with respect to $h=g/f$ (i.e., computing the public key in this paper), it is essentially to compute $h= INTT( \hat{g} \circ \hat{f}^{-1} )$. Here,  $\hat{g}=NTT(g)$, $\hat{f}=NTT(f)$, and $\hat{f}^{-1}= (\hat{f}_0^{-1},\hat{f}_1^{-1},\ldots,\hat{f}_{\frac{n}{2}-1}^{-1})$ where $\hat{f}_i^{-1}$ is the inverse of $\hat{f}_i$ in $\mathbb{Z}_q[x]/(  x^{2} - \zeta^{\tau(i)})$, if each $\hat{f}_i^{-1}$ exits, $i=0,1,\ldots,\frac{n}{2}-1$.

\subsubsection{The pure radix-2 NTT}

Similar techniques can be utilized to $\mathbb{Z}_{q}[x]/(x^{n} - x^{n/2} + 1)$ w.r.t. $(n=512,q=3457)$ and $(n=1024,q=3457)$.
As for both of them, only the steps of the radix-2 FFT trick are required, due to the power-of-two $n$. Note that for these two $n$, there only exits the $384$-th primitive root of unity $\zeta$ in $\mathbb{Z}_{q}$. Therefore, $\mathbb{Z}_q[x]/(  x^{n} - x^{n /2} + 1 )$ can be decomposed into $\prod\limits_{i=0}^{n/4-1}  { \mathbb{Z}_q[x]/(  x^{4} - \zeta^{\tau(i)}  )  } $ for $n=512$, but into $\prod\limits_{i=0}^{n/8-1}  { \mathbb{Z}_q[x]/(  x^{8} - \zeta^{\tau(i)}  )  } $ for $n=1024$. Thus, the point-wise multiplication and the base case inversion are aimed at the corresponding polynomials of degree 3 for $n=512$ (degree 7 for $n=1024$).

\subsection{The Unified NTT}\label{sec-unified-ntt}

As mentioned above, in order to achieve an efficient implementation, we conduct 6 steps of radix-2 FFT trick for $n \in \{512,768,1024\}$ and one more step of radix-3 FFT trick for $n=768$. Therefore, for different $n$'s, the various (mixed-radix and radix-2) NTT algorithms are required for three types of parameter sets. However, for many applications or platforms, all the three types parameter sets may need to be implemented. In order to deal with the inconvenient issue, we present a unified NTT methodology over $\mathbb{Z}_q[x]/(x^n-x^{n/2}+1)$, $n\in \{512,768,1024\}$, such that only one type of NTT computation is required for different $n$'s, which is useful for modular and unified implementations when all the three parameter sets are required.

In this work, we consider $n=\alpha \cdot N$, where $\alpha\in\{2,3,4\}$ is called the splitting-parameter and $N$ is a power of two. In fact, $\alpha$ can be chosen more freely as arbitrary values of the form $2^i 3^j,i\geq 0,j\geq 0$. With the traditional NTT technique, when the dimension $n$ changes we need to use different NTT algorithms of various input/output lengths to compute polynomial multiplications over $\mathbb{Z}_q[x]/(x^n -x^{n/2} + 1)$. This causes much inconvenience to software and particularly hardware implementations.  To address this issue, we unify the  various $n$-point NTTs through an $N$-point NTT, which is referred to as the unified NTT technique.  For $n\in \{512,768,1024\}$, we fix $N=256$ and choose $\alpha\in\{2,3,4\}$. With this technique, we only focus on the implementation of the  $N$-point NTT, which serves as the unified procedure to be invoked for different $n$'s. Specifically, the computation of NTT over $\mathbb{Z}_q[x]/(x^n -x^{n/2} + 1)$ is divided into three steps. For presentation simplicity, we only give the procedures of the forward transform as follows, since the inverse transform can be obtained by inverting these procedures. The map road is shown in Figure~\ref{fig-crt-split-map}.




\begin{figure}[!t]
	\centering
	\includegraphics[width=0.7\linewidth]{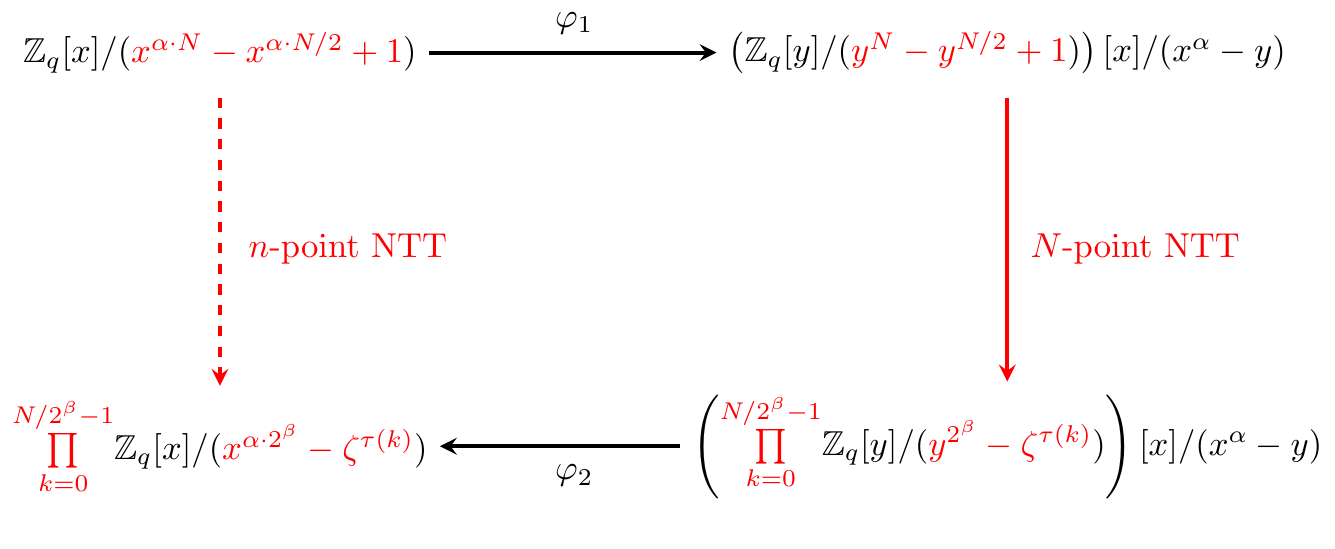}
	\vspace{-0.8cm}
	\caption{Map road for unified NTT}
	\label{fig-crt-split-map}
	\vspace{-0.5cm}
\end{figure}

\textbf{Step 1.} Construct a splitting-polynomial map $	\varphi_1 : $
{\small
	\begin{equation*}
	\begin{split}
	\mathbb{Z}_q[x]/( x^{\alpha \cdot N} - x^{\alpha \cdot N /2} + 1) & \rightarrow  \left(\mathbb{Z}_q[y]/(  y^{N} - y^{N /2} + 1 ) \right)[x]/ (x^{\alpha} - y )\\
	f = \sum\limits_{i=0}^{\alpha\cdot N -1}{ f_i x^i} & \mapsto \sum_{j=0}^{\alpha-1}{ F_j  x^j }
	\end{split}
	\end{equation*}
}where $F_j = \sum\limits_{i=0}^{N-1}{ f_{\alpha \cdot i + j } y^i} \in \mathbb{Z}_q[y]/(  y^{N} - y^{N /2} + 1 )$. Namely, the $n$-dimension polynomial is split into $\alpha$ $N$-dimension sub-polynomials.

\textbf{Step 2.} Apply the unified $N$-point NTT to $F_j$ over $\mathbb{Z}_q[y]/(  y^{N} - y^{N /2} + 1 )$, $j=0,1,\ldots,\alpha-1$. Specifically, inspired by NTTRU~\cite{nttru-LS19}, there is a mapping such that $\mathbb{Z}_q[y]/(y^{N} - y^{N /2} + 1 ) \  \cong  \  \mathbb{Z}_q[y]/( y^{N/2} - \zeta_{1})  \times \mathbb{Z}_q[y]/( y^{N/2} - \zeta_{2})$ where $\zeta_{1} + \zeta_{2}=1$ and $ \zeta_{1} \cdot \zeta_{2}=1$.  Let $q$ be the prime satisfying $\frac{3N}{2^\beta} | (q-1) $, where $\beta \in \mathbb{N}$ is called the truncating-parameter, such that it exits the primitive $\frac{3N}{2^\beta}$-th root of unity $\zeta$ in $\mathbb{Z}_q$. To apply the radix-2 FFT trick, we choose $\zeta_{1} = \zeta^{N/2^{\beta+1}} \bmod q$ and $ \zeta_{2} = \zeta_{1}^5 = \zeta^{5N/2^{\beta+1}} \bmod q$. Thus, both $y^{N/2} - \zeta_{1}$ and $y^{N/2} - \zeta_{2}$ can be recursively split down into degree-$2^\beta$ terms like $y^{2^\beta} \pm \zeta$. The idea of truncating FFT trick originates from~\cite{incomplete-ntt-moenck76}. Therefore, $\mathbb{Z}_q[y]/(  y^{N} - y^{N /2} + 1 )$ can be decomposed into $\prod\limits_{k=0}^{N/2^{\beta}-1}  { \mathbb{Z}_q[y]/(  y^{2^\beta} - \zeta^{\tau(k)}  )  } $, where $\tau(k)$ is the power of $\zeta$ of the $k$-th term and we start the index $k$ from zero. Let $\hat{F}_j$ be the NTT result of $F_j$ and $\hat{F}_{j,l}$ be its $l$-th coefficient, $l=0,1,\ldots,N-1$. Hence, we can write
{\small
	\begin{equation*}
	\begin{split}
	\hat{F}_j=( \sum\limits_{l=0}^{2^\beta -1 }{ \hat{F}_{j,l} y^{l} },  \sum\limits_{l=0}^{2^{\beta} -1 }{ \hat{F}_{j,l+2^\beta} y^{l} },\ldots,  \sum\limits_{l=0}^{2^{\beta} -1 }{ \hat{F}_{j,l+N-2^\beta} y^{l} } )
	\in \prod\limits_{k=0}^{N/2^{\beta}-1} { \mathbb{Z}_q[y]/(  y^{2^\beta} - \zeta^{\tau(k)}  ) }
	\end{split}
	\end{equation*}
}
\textbf{Step 3.} Combine the intermediate values and obtain the final result by the map $\varphi_2 :$
{\footnotesize
	\begin{equation*}
	\begin{split}
	\left(  \prod\limits_{k=0}^{N/2^{\beta}-1} { \mathbb{Z}_q[y]/( y^{2^\beta} - \zeta^{\tau(k)} )  } \right) [x]/(x^{\alpha} - y) & \rightarrow  \prod\limits_{k=0}^{N/2^{\beta}-1}{ \mathbb{Z}_q[x]/(x^{\alpha \cdot 2^\beta} - \zeta^{\tau(k)})  }\\
	\sum_{j=0}^{\alpha-1}{ \hat{F}_j  x^j } & \mapsto  \hat{f}
	\end{split}
	\end{equation*}
}
where $\hat{f} = \sum\limits_{i=0}^{\alpha\cdot N -1}{ \hat{f}_i x^i} $ is the NTT result of $f$. Its $i$-th coefficient is $\hat{f}_i = \hat{F}_{j,l}$, where $j =i\bmod\alpha$ and $l =\lfloor \frac{i}{\alpha}\rfloor $. It can be rewritten as:
{\small
	\begin{equation}\label{equ-unified-ntt-result-of-f-in-multi-ring}
	\begin{split}
	\hat{f}=( \sum\limits_{i=0}^{\alpha \cdot 2^\beta -1 }{ \hat{f}_i x^i}, \sum\limits_{i=0}^{\alpha \cdot 2^\beta -1 }{ \hat{f}_{i+\alpha \cdot 2^\beta} x^i} ,\ldots, \sum\limits_{i=0}^{\alpha \cdot 2^\beta -1 }{ \hat{f}_{i+n-\alpha \cdot 2^\beta} x^i}  )
	\in \prod\limits_{k=0}^{N/2^{\beta}-1}{ \mathbb{Z}_q[x]/(x^{\alpha \cdot 2^\beta} - \zeta^{\tau(k)})  }
	\end{split}
	\end{equation}
}

In this work, we choose $\beta=1$ and $q=3457$, where the primitive $384$-th root of unity $\zeta=55$ exits in $\mathbb{Z}_{3457}$. In this case, the point-wise multiplication is the corresponding $2\alpha$-dimension polynomial multiplication in $\mathbb{Z}_q[x]/(x^{2\alpha } - \zeta^{\tau(k)}),\alpha \in \{2,3,4\},k=0,1,\ldots,N/2-1$.

%
%
%

\subsection{Discussions}

As for the application scenarios w.r.t. $n=768$, the polynomial multiplication and division over $\mathbb{Z}_q[x]/(x^n-x^{n/2}+1)$ can be efficiently improved with the aid of the mixed-radix NTT (the benchmark results are shown in section~\ref{sec-benchmark-comparison}). But, this type of NTT can not be implemented universally and modularly for more general $n$'s, since it can be only applied in the case of $n=3 \cdot 2^{e}$ for some integer $e$, instead of power-of-two $n$ like $512$ and $1024$. Hence, for our CTRU and CNTR, three various NTT algorithms would be needed for the three recommended parameter sets. The unified NTT can overcome the inconvenient issue. Note that in the base case inversion of the unified NTT-based polynomial division, we need to compute the inverses of degree-3, degree-5 and degree-7 polynomials for $n=512,768,1024$ respectively, which are more complicated than that of linear polynomial in the mixed-radix NTT-based polynomial division. It causes that the unified NTT performs less efficiently in the KeyGen algorithm.

However, most application scenarios and cryptographic devices are usually equipped with three recommended parameter sets ($n=512,768,1024$) for the targeted security levels. In these cases, the unified NTT can lead to modular and simplified software and hardware implementation. Note that in practice the KeyGen algorithm is run once and for all, and its computational cost is less sensitive to most cryptographic applications. As KeyGen is less frequently run, we have taken priority on simple, unified and modular implementations, i.e., unified NTT. In addition, when compared to NTRU-HRSS and SNTRU Prime, the KeyGen algorithm of CTRU and CNTR can be still faster in the context of the unified NTT.

\subsection{Base Case Inversion}

The inverse of the polynomials in $\mathbb{Z}_q[x]/(x^{i} - \zeta^{j}),i,j \ge 0$ can be computed by Cramer's Rule~\cite{cramer-rule-linear-algebra-book}. Take $\mathbb{Z}_q[x]/(x^{4} - \zeta) $ as an example. Let $f$ be a degree-3 polynomial in $\mathbb{Z}_q[x]/(x^{4} - \zeta)$, and denote its inverse by $f'$, which implies $f \cdot f' =1  \bmod x^{4 } - \zeta$. It can be written in the form of matrix-vector multiplication:

\begin{equation}\label{equ-compute-inverse-of-f}
\left[
\begin{array}{cccc}
f_{0} & \ \zeta f_{3} & \ \zeta f_{2} & \ \zeta f_{1} \\
f_{1} & \    f_{0}    & \ \zeta f_{3} & \ \zeta f_{2} \\
f_{2} & \    f_{1}    & \    f_{0}    & \ \zeta f_{3} \\
f_{3} & \    f_{2}    & \    f_{1}    & \    f_{0}
\end{array}
\right]
\cdot
\left[
\begin{array}{c}
f'_{0}\\
f'_{1}\\
f'_{2}\\
f'_{3}
\end{array}
\right]
=
\left[
\begin{array}{c}
1\\
0\\
0\\
0
\end{array}
\right] .
\end{equation}

Let $\Delta$ be the determinant of the coefficient matrix. Hence, the inverse of $f$ exits if and only if $\Delta \ne 0 $. In this case, according to Cramer's Rule, there is a unique $f'$, whose individual components are given by
\begin{equation}
f'_i = \cfrac{\Delta_i}{\Delta}, i=0,1,2,3
\end{equation}
where $\Delta_i$ is the determinant of the matrix generated by replacing the $(i+1)$-th column of the coefficient matrix with $(1,0,0,0)^T$. And  $\Delta^{-1}$ can be computed by using Fermat's Little Theorem, i.e., $\Delta^{-1} \equiv \Delta^{q-2} \bmod q$.

\subsection{Multi-moduli NTT}\label{sec-ntt-unfriendly-ring}

Although directly-using NTT is invalid over $\mathcal{R}_{q_2}$ w.r.t. power-of-two $q_2$ in the decryption process of CTRU and CNTR, the work~\cite{ntt-unfriendly-ring-CHK21,multi-moduli-ACC22} show that it is still possible to conduct an efficient multi-moduli NTT over $\mathcal{R}_{q_2}$. Briefly speaking, according to ~\cite{ntt-unfriendly-ring-CHK21}, the polynomial multiplication over $\mathcal{R}_{q_2}$ can be lifted to that over $\mathcal{R}_{Q}=\mathbb{Z}_{Q}[x]/(x^n-x^{n/2}+1)$, where $Q$ is a positive integer and larger than the maximum absolute value of the coefficients during the computation over $\mathbb{Z}$. Then, one can recover the targeted product polynomial through reduction modulo $q_2$. We choose $Q=qq'$ where $q=3457$ and $q'=7681$ in this paper. We have the CRT isomorphism: $\mathcal{R}_{Q} \cong \mathcal{R}_{q} \times \mathcal{R}_{q'}$. The first NTT algorithm w.r.t ($n,q$) over $\mathcal{R}_{q}$ can be instanced as the mixed-radix NTT in section~\ref{sec-mixed-radix-ntt} or the unified NTT in section~\ref{sec-unified-ntt}. The second NTT algorithm w.r.t. ($n,q'$) over $\mathcal{R}_{q'}$ can be followed from that in NTTRU~\cite{nttru-LS19}. After using NTTs to compute two intermediate products in $\mathcal{R}_{q}$ and $\mathcal{R}_{q'}$ respectively, the targeted product in $\mathcal{R}_{Q}$ can be recovered via CRT.

%
%
%
%

\section{Implementation}\label{sec-implementation}

In this section, the remaining details of the implementations of our schemes are provided, including our portable C implementation, as well as optimized implementations with AVX2 instruction sets. All the implementations are carefully protected against timing attacks.

\subsection{Symmetric Primitives}

All the hash functions are instantiated with functions from \textbf{SHA-3} family. To generate the secret polynomials, i.e., $f',g,r,e$, the secret seeds are needed to be expanded to the sampling randomness by using \textbf{SHAKE-128}. The hash function $\mathcal{H}$ is instantiated with \textbf{SHA3-512}, aiming to hash the short prefix of public key $ID(pk)$ and message $m$ into the 64 bytes where the first 32 bytes are used to generate the shared keys and the later 32 bytes are used as the secret seed for the encryption algorithm.

\subsection{Generation of Secret Polynomials}

All the secret polynomials in our schemes are sampled according to the centered binomial distribution $B_{\eta}$. Each of them totally requires $2n\eta$ bits, or saying $2n\eta/8$ bytes, as sampling randomness, which are produced from the output of \textbf{SHAKE-128} with a secret seed as an input. To generate each coefficient of an secret polynomial, we arrange the adjacent independent $2\eta$ random bits and subtract the Hamming weight of the most significant $\eta$ bits from the Hamming weight of the least significant $\eta$ bits.

\subsection{The Keys and Ciphertexts}

\textbf{The format of the public key.}
The public key is transmitted in the NTT representation, as in the works like~\cite{kyber-BDK+18,kyber-nist-round3,newhope-usenix-ADPS16,nttru-LS19}. Specifically, treat the public key as $\hat{h}= \hat{g} \circ \hat{f}^{-1} $, which saves an inverse transform in the key generation and a forward transform in the encryption (re-run in the decapsulation). The coefficients of $\hat{h}$ are reduced modulo $q$ into $\mathbb{Z}_q$, causing that each coefficient occupies 12 bits. Therefore, the public key is packed into an array of $12n$ bits, i.e., $3n/2$ bytes in total.

\textbf{The format of the secret key.}
Note that the polynomial $f$ has coefficients in normal representation of $[-2\eta,2\eta+1]$ where $\eta$ is the parameter of the centered binomial distribution $B_{\eta}$. Instead of directly packing the polynomial $f$ into bytes array, we subtract each coefficient from  $2\eta+1$, making sure that all the coefficients are in $[0,4\eta+1]$. We pack the resulting polynomial into $n \lceil \log(4 \eta+1) \rceil/8$ bytes. The initial coefficient can be recovered by being subtracted from $2\eta+1$ in the unpack step in decryption. Since the public key is needed in the re-encryption during the decapsulation, we simply concatenate and store the packed public key as part of the secret key. An extra 32-byte $z$ is also concatenated, since $z$ is used to derive a pseudo-random key as output of implicit rejection if re-encryption does not succeed. The total size of a decapsulation secret key contains $n \lceil\lfloor \log(4 \eta+1) \rceil/8+3n/2+32$ bytes.

\textbf{The format of the ciphertext.}
The ciphertexts of our schemes consist of only one (compressed) polynomial $c$. The polynomial $c$ is in normal representation instead of NTT representation, since the compression through rounding has to work in normal representation. Each coefficient of $c$ occupied $\lceil \log(q_2) \rceil$ bits. Thus, to pack and store such a ciphertext only costs $n \lceil \log(q_2) \rceil/8$ bytes.

\textbf{The prefix of the public key.}
As for the prefix $ID(pk)$ of the public key $h$ in CTRU and CNTR, we use the first 33 bytes of the bit-packed NTT representation of $h$. It is reasonable, since $h$ is computationally indistinguishable from a uniformly random polynomial in $\mathcal{R}_q$ and the forward NTT transform keeps the randomness property (i.e., $h$ is uniformly random, so is $\hat{h}=NTT(h)$). Thus, the first 22 coefficients of the public key have the min-entropy of more than 256 bits and occupy 33 bytes in the bit-packed NTT representation since each coefficient has 12 bits.

\subsection{Portable C Implementation}

Our portable C implementations rely on 16-bit and 32-bit integer arithmetic mainly, excluding any floating-point arithmetic. The polynomials are represented as arrays of 16-bit signed integers. It is reasonable since we use a 12-bit prime. Our implementation of decoding algorithm of the scalable $\text{E}_8$ lattice follows the methodology in~\cite{akcn-e8-JSZ22} which is based on 32-bit integer arithmetic, but with high developments on constant-time skills for security and simplicity.

\textbf{NTT implementation.}
Our C implementations of NTTs do not make use of variable-time operator ``\%'' for the modular reductions, but we turn to use Barrett reduction~\cite{barrett-reduce,sei18} and Montgomery reduction~\cite{montgomery-reduce,sei18}, where the former is applied after additions and the later is applied for multiplication between coefficients with primitive roots. But we only use the signed variants of these two reductions as described in~\cite{sei18}. Lazy reduction strategy~\cite{sei18} is suitable for the forward NTT transform. Note that the output range of Montgomery reduction is in $[-q,q]$. For 12-bit coefficients of the input polynomial, after 7-level FFT tricks, the forward NTT transform outputs the polynomials with coefficients in $[-8q,8q]$, which does not overflow the valid representation of a 16-bit signed integer in the context of 12-bit $q$. However, NTT is invalid in $\mathcal{R}_{q_2}$ w.r.t. power-of-two $q_2$ in the decryption process, so we turn to the schoolbook algorithm to compute $c f \bmod^{\pm} q_2$ for efficiency and simplicity in the portable C implementation, where the modular reduction w.r.t. power-of-two $q_2$ can be implemented by logical AND operations efficiently.

\subsection{Optimized AVX2 Implementation}

The optimized implementations of our schemes for CPUs which support the AVX2 instruction sets are provided. The main optimized targets are polynomial arithmetic, sampling secrets and modular reduction algorithms in NTT, all of which are the time-consuming operations.
However, as for \textbf{SHA-3} hash functions, we do not have any AVX2-based optimization. Consistently, we use the same source codes as in portable C implementation. This is because the vectorized implementations of \textbf{SHA-3} hash functions are not very helpful for accelerating, and the fastest implementation is based on C language~\cite{supercop,newhope-usenix-ADPS16}. As for the computation of $c f \bmod^{\pm} q_2$ w.r.t. power-of-two $q_2$ in the decryption process of CTRU and CNTR, we choose the multi-moduli NTT (see section~\ref{sec-ntt-unfriendly-ring}), instead of the schoolbook algorithm, which is different from that in the portable C implementation. Besides, according to our experiments on a full polynomial multiplication over $\mathcal{R}_{q_2}$, compared to the schoolbook algorithm, the multi-moduli NTT is slower in the context of C implementation, but is faster in the context of AVX2 implementation, for which NTT is suitable for vectorized implementation, especially AVX2.

\textbf{NTT optimizations.}
Our AVX2-based NTT implementation handles 16-bit signed integer coefficients, every 16 values of which are loaded into one vector register. Load and store instructions are time-consuming in AVX2 instruction set. To accelerate AVX2 implementation, we need to reduce the memory access operations. We present some implementation strategies to fully utilize vector registers and minimize total CPU cycles. For the radix-2 FFT trick, we merge the first three levels and the following three levels. During the merging levels there is no extra load or store operations. We achieve this by using different pair of vector registers and permutating coefficients order.  The instructions we use for permutation task are \textsf{vperm2i128}, \textsf{vpunpcklqdq}, \textsf{vpblendd} and \textsf{vpblendw}. The coefficients are permutated in levels 3-5. After the radix-2 FFT trick, we do not choose to store the coefficients immediately, instead we use the vector register to complete the radix-3 FFT trick if necessary. It's because after the last level of the radix-2 FFT trick, the order of coefficients in the vector register is naturally the order we need in the radix-3 FFT trick. To reduce total permutation time, we propose to store coefficients in a shuffled order. For polynomial point-wise multiplication and polynomial inversion, pairwise modular multiplications are involved, since we store coefficients in the shuffled order. The two continual coefficients are stored in two different vector registers, therefore, we can easily implement polynomial point-wise multiplication and polynomial inversion. We encapsulate some functional codes which are used multiple times during the whole process into \textsf{macro}, making the codes more concise and readable, as well as avoiding repetitive codes.

\subsection{Constant-time Implementation}\label{sec-constant-time}

We report on our constant-time implementation to avoid the potential timing attacks. Specifically, our implementations do not use any variable-time instructions to operate the secret data, do not use any branch depending on the secret data and do not access any memory at addresses depending on the secret data.

As for the modular reductions used in the NTTs, as described in~\cite{sei18,nttru-LS19}, both Barrett reduction~\cite{barrett-reduce,sei18} and Montgomery reduction~\cite{montgomery-reduce,nttru-LS19} used in our implementations are constant-time algorithms. Furthermore, the reduction algorithms are not specific to the modulus $q$.

As for the scalable $\text{E}_8$ lattice code, in our encoding algorithm, $\mathbf{k} \mathbf{H} \bmod 2$ can be computed efficiently by simple bitwise operations, which we have implemented with constant-time steps. For the implementation of the scalable $\text{E}_8$ decoding algorithms in Algorithm~\ref{algo-e8-decoding-e8} and Algorithm~\ref{algo-e8-decoding-c}, we present branching-free implementations. All the ``arg min'' statements and ``if'' conditional statements are implemented by constant-time bitwise operations.
In essence, these can be summarized as choosing the minimal value of two secrets in signed integer representation, which are defined as $\texttt{a},\texttt{b}$. This can be implemented without timing leakage of the secret data flow as: \texttt{c} = ((-(\texttt{r} \textsf{XOR} 1)) \textsf{AND} \texttt{a}) \textsf{XOR} ((-(\texttt{r} \textsf{AND} 1)) \textsf{AND} \texttt{b}), where $\texttt{r}\in \{0,1\}$ is the sign bit of the value $\texttt{b}-\texttt{a}$, \textsf{XOR} is the logical Exclusive OR operator, and \textsf{AND} is the logical AND operator. We emphasize that, although there exit a variety of error correction codes, including lattice codes, there are indeed inherent difficulties on constant-time implementations for most existing error correction codes. Take BCH code and LDPC code~\cite{bch-ldpc-newhope-simple} as examples, which are widely used in reality. BCH code does not enable a constant-time implementation for its decoding process, since it needs to locate the error bits by computing the syndrome and correct the error bits within its range of error correction capability, but these proceeds are not constant-time~\cite{bch-ldpc-non-constant-time}. LDPC code also does not have a constant-time decoding process, since its decoding process works under iterative steps which stop unless the errors are corrected or the iterations reach the maximum number~\cite{bch-ldpc-non-constant-time}. A similar situation happens to some lattice codes. For example, none has found constant-time implementations of decoding algorithms with respect to $BW_{32}$ lattice and $BW_{64}$ lattice~\cite{e8-lattice-decoding-book-CS13}. However, unlike those error correction codes, our scalable $\text{E}_8$ lattice code features constant-time encoding and decoding algorithms, enabling safe implementations against timing attacks.

%
%
%
%

\section{Benchmark and Comparison}\label{sec-benchmark-comparison}

In this section, we provide the benchmark results of our CTRU and CNTR where we focus on the recommended parameter set of dimension 768, i.e., $(n=768, q=3457, q_2=2^{10}, \Psi_1=\Psi_2 = B_2)$ for CTRU-768 and $(n=768, q=3457, q_2=2^{10}, \Psi_1=\Psi_2 = B_3)$ for CNTR-768 under the applications of the mixed-radix NTT.
All the benchmark tests are run on an Intel(R) Core(TM) i7-10510U CPU at 2.3GHz (16 GB memory) with Turbo Boost and Hyperthreading disabled. The operating system is Ubuntu 20.04 LTS with Linux Kernel 4.4.0 and the gcc version is 9.4.0. The compiler flag of our schemes is listed as follows: \emph{-Wall -march=native -mtune=native -O3 -fomit-frame-pointer -Wno-unknown-pragmas}. We run the corresponding KEM algorithms for 10,000 times and calculate the average CPU cycles. The benchmark results are shown in Table~\ref{tab-benchmark-of-schemes}, along with comparisons with other schemes. Concretely, we re-run the C source codes of parts of other schemes on the exact same system as CTRU and obtain the corresponding benchmark results for providing reasonable reference comparisons, but their state-of-the-art AVX2 benchmark results are directly taken from the literatures or SUPERCUP (the supercop-20220506 benchmarking run on a 3.0GHz Intel Xeon E3-1220 v6)~\cite{supercop}. Regarding the benchmark results in this section, we stress that this may be not an exhaustive benchmark ranking but serves as optional illustration that our schemes might perform reasonably well when compared to other schemes.

\begin{table*}
	\centering
	\caption{CPU cycles of KEMs (in kilo cycles).}
	\begin{threeparttable}
		\begin{tabular}{c|ccc|ccc}
			\hline
			       \multirow{2}{*}{Schemes}         &                       \multicolumn{3}{c|}{C}                       &                      \multicolumn{3}{c}{AVX2}                       \\ \cline{2-7}
			                                        &        KeyGen         &       Encaps        &        Decaps        &        KeyGen         &        Encaps        &        Decaps        \\ \hline
			   CTRU-768, Mixed-radix NTT (Ours)     &         117.5         &        63.4         &        134.6         &         10.6          &         11.7         &         35.7         \\
			   CNTR-768, Mixed-radix NTT (Ours)     &         118.4         &        64.9         &        133.1         &         12.8          &         10.6         &         35.4         \\ \hline
			     CTRU-768, Unified NTT (Ours)       &  $6.2 \times 10^{3}$  &        75.3         &        135.9         &          --           &          --          &          --          \\
			     CNTR-768, Unified NTT (Ours)       &  $6.3 \times 10^{3}$  &        78.0         &        137.0         &          --           &          --          &          --          \\ \hline
			               NTRU-HRSS                & $127.6 \times 10^{3}$ & $3.2 \times 10^{3}$ & $9.4 \times 10^{3}$  &         254.0         &         24.9         &         59.2         \\
			            SNTRU Prime-761             & $17.1 \times 10^{3}$  & $9.0 \times 10^{3}$ & $23.7 \times 10^{3}$ &         156.3         &         46.9         &         56.2         \\
			               Kyber-768                &         140.3         &        159.0        &        205.9         &         25.3          &         27.6         &         43.4         \\
			               Saber-768                &         94.7          &        109.7        &        138.9         &         64.2          &         69.3         &         95.3         \\ \hline
			CTRU-768, \textbf{SHA-2} variant (Ours) &         101.4         &        58.6         &        116.9         &          7.0          &         6.0          &         21.8         \\
			CNTR-768, \textbf{SHA-2} variant (Ours) &         104.6         &        60.0         &        114.2         &          8.4          &         6.0          &         23.6         \\
			                 NTTRU                  &         157.4         &        98.9         &        142.4         &          6.4          &         6.1          &         7.9          \\ \hline
			            BIKE (Level 3)              &          --           &         --          &          --          &  $1.7 \times 10^{3}$  &        267.1         & $5.3 \times 10^{3}$  \\
			        Classic McEliece460896          &          --           &         --          &          --          & $150.0 \times 10^{3}$ &         77.3         &        253.9         \\
			                HQC-192                 &          --           &         --          &          --          &         417.8         &        719.7         & $1.2 \times 10^{3}$  \\
			               SIKEp610                 &          --           &         --          &          --          & $15.0 \times 10^{3}$  & $27.2 \times 10^{3}$ & $27.7 \times 10^{3}$ \\ \hline
		\end{tabular}

	\end{threeparttable}
	\label{tab-benchmark-of-schemes}
\end{table*}

\subsection{Comparison with Other NTRU-based KEM Schemes}

The C source codes of NTRU-HRSS and SNTRU-Prime are taken from their Round 3 supporting documentations, while those of NTTRU are taken from~\cite{nttru-LS19}. One regret is that the source codes of $\text{NTRU-C}_{3457}^{768}$ are not online available in~\cite{ntru-variant-eprint-1352-DHK21}, and the AVX2 benchmark results of $\text{NTRU-C}_{3457}^{768}$ are absent in~\cite{ntru-variant-eprint-1352-DHK21}, so all of its benchmark results are omitted here.

Note that the work~\cite{ntt-unfriendly-ring-CHK21} shows how to apply multi-moduli NTT to accelerate the polynomial multiplications in NTRU-HRSS, but the polynomial divisions remain unchanged. However, the resulting speed-up for NTRU-HRSS in~\cite{ntt-unfriendly-ring-CHK21} is not obvious (in fact, the speed-up is $\pm 0\%$). Hence, we omit the benchmark results about NTRU-HRSS provided in~\cite{ntt-unfriendly-ring-CHK21}. Consequently, the benchmark results of the state-of-the-art AVX2 implementation of NTRU-HRSS are reported in~\cite{supercop}. As for SNTRU Prime-761, its state-of-the-art AVX2 implementation is presented in~\cite{opensslntru-BBCT22}. Thus, we take the AVX2 benchmark results of NTRU-HRSS and SNTRU Prime-761 from~\cite{supercop} and~\cite{opensslntru-BBCT22}, respectively. The AVX2 benchmark results of NTTRU are taken from~\cite{nttru-LS19}.

When compared to NTRU-HRSS and SNTRU Prime-761, the efficiency improvements of CTRU-768 and CNTR-768 are benefited from the applications of NTT in polynomial operations. For example, as for portable C implementation, CTRU-768 is faster than NTRU-HRSS by 1,000X in KeyGen, 50X in Encaps, and 69X in Decaps, respectively. As for optimized AVX2 implementation, CTRU-768 is faster than NTRU-HRSS by 23X in KeyGen, 2.1X in Encaps, and 1.6X in Decaps, respectively; CNTR-768 is faster than NTRU-HRSS by 19X in KeyGen, 2.3X in Encaps, and 1.6X in Decaps, respectively.

When compared to NTTRU fairly, we do the following modifications to present vaianrts of CTRU and CNTR, which have the same hash functions, symmetric primitives and FO transformations as NTTRU: (1) use \textbf{SHA-2} family to instantiate hash functions; (2) use \textbf{AES} to expand seeds; (3) change the FO transformation into $\text{FO}_{m}^{\bot}$.
The AVX2 benchmark results of the KeyGen of CTRU/CNTR variants are close to those of NTTRU, whereas the Encaps results of CTRU/CNTR variants are faster than that of NTTRU. However, the Decaps results of CTRU/CNTR variants are slightly slower than that of NTTRU, on the following grounds: (1) the decoding algorithm of the scalable $\text{E}_8$ lattice code costs extra time; (2) multi-moduli NTT is more time-consuming than the NTT algorithm of NTTRU, since multi-moduli NTT essentially consists of two routines of NTT algorithms.

\subsection{Comparison with Other Lattice-based KEM Schemes}

The C source codes of Kyber-768 and Saber-768 are taken from their Round 3 supporting documentations. We modify their FO transformation into $\text{FO}_{ID(pk),m}^{\not\bot}$, and re-run their C source codes. The work~\cite{fo-transform-prefix-hash-DHK+21} has reported the state-of-the-art AVX2 implementation of Kyber and Saber with $\text{FO}_{ID(pk),m}^{\not\bot}$. Thus, we take their AVX2 benchmark results from~\cite{fo-transform-prefix-hash-DHK+21} directly.

As shown in Table~\ref{tab-benchmark-of-schemes}, our CTRU-768 and CNTR-768 outperform both Kyber-768 and Saber-768. For example, when compared to the state-of-the-art AVX2 implementation of Kyber-768, CTRU-768 is faster by 2.3X in KeyGen, 2.3X in Encaps, and 1.2X in Decaps, respectively; CNTR-768 is faster by 1.9X in KeyGen, 2.6X in Encaps, and 1.2X in Decaps, respectively. It is due to the following reasons: (1) in Kyber, there are a rejection sampling to generate the matrix $\mathbf{A}$ and a complicated polynomial matrix-vector multiplication in the  key generation process and the encryption process of Kyber-768 (which are also re-run with the Decaps); (2) but there is only one polynomial multiplication in the encryption process of CTRU-768 and CNTR-768.

\subsection{Comparison with Other Non-lattice-based KEM Schemes}

We present a rough comparison with other non-lattice-based KEM schemes, i.e., BIKE~\cite{bike-nist-round3}, Classic McEliece~\cite{classic-mcEliece-nist-round3}, HQC~\cite{hqc-nist-round3} and SIKE~\cite{sike-nist-round3}, which are candidates advancing to the fourth round of NIST PQC~\cite{nist-round-4-submissions}. The first three KEM schemes are code-based, and the last one is isogeny-based. However, the SIKE team acknowledges that SIKE is insecure and should not be used~\cite{nist-round-4-submissions}. Nevertheless, the benchmark results of SIKE are still presented and only used for intuitive comparisons. We only present their state-of-the-art AVX2 benchmark results, which can be also found in SUPERCUP~\cite{supercop}. As shown in Table~\ref{tab-benchmark-of-schemes}, our schemes are much faster than these non-lattice-based KEM schemes. For example, CTRU-768 is faster by 160X in KeyGen, 6.6X in Encaps and 7.1X in Decaps than Classic McEliece460896.

\subsection{Benchmark Results with Unified NTT}

The C implementations of CTRU-768 and CNTR-768 with our unified NTT are also provided, whose benchmark results could be found in Table~\ref{tab-benchmark-of-schemes}. Although the overall performances of CTRU-768 and CNTR-768 with our unified NTT are inferior than those of CTRU-768 and CNTR-768 with mixed-radix NTT, we stress that the primary goal of the unified NTT is to provide a modular and convenient implementation, instead of a faster implementation. Note that an optimized AVX2 implementation of the unified NTT could be more precise to present benchmark results. But the AVX2 implementation is still a work in progress and we left as a future work.

%
%
%
%

\section*{Acknowledgments}
We would like to thank Haodong Jiang, Yang Yu, and Zhongxiang Zheng for their helpful feedbacks on this work.

%
%
%
%

\appendices

%
%
%
%

\section{On CCA Security Reduction of KEM in the ROM and the QROM}\label{app-sec-reduction}

Generic constructions of an efficient IND-CCA secure KEM are well studied in~\cite{fo-transform-Dent03,fo-transform-HHK17}, which are essentially various KEM variants of Fujisaki-Okamoto (FO) transformation~\cite{fo-transform-FO99} and GEM/REACT transformation~\cite{fo-transform-GEM02,fo-transform-REACT01}. The work~\cite{fo-transform-HHK17} gives a modular analysis of various FO transformations in the ROM and the QROM, and summarizes some practical FO transformations that are widely used to construct an IND-CCA secure KEM from a passive secure PKE (e.g., OW-CPA and IND-CPA), including the following transformations $\text{FO}^{\bot}$, $\text{FO}_{m}^{\bot}$, $\text{FO}^{\not\bot}$, $\text{FO}_{m}^{\not\bot}$, $\text{U}^{\not\bot}$ and $\text{U}_{m}^{\not\bot}$, etc, where $m$ (without $m$) means $K=H(m)$ ($K=H(m,c)$), $\not\perp$ ($\bot$) means implicit (explicit) rejection.

$\text{FO}^{\bot}$, $\text{FO}_{m}^{\bot}$, $\text{FO}^{\not\bot}$ and $\text{FO}_{m}^{\not\bot}$ are the most common transformations used in NIST PQC. According to~\cite{fo-transform-HHK17}, in the ROM, the reduction bound of these four transformations are all $\epsilon' \le \epsilon_{CPA} + q' \delta$ and $\epsilon' \le q' \epsilon_{OW} + q' \delta$, where $\epsilon'$ is the advantage of an adversary against IND-CCA security of KEM, $\epsilon_{CPA}$ ($\epsilon_{OW}$) is the the advantage of an adversary against IND-CPA (OW-CPA) security of the underlying PKE, $q'$ is the total number of hash queries, and $\delta$ is the error probability. Notice that in order to keep the comparison lucid, we ignore the small constant factors and additional inherent summands. The reduction is tight for IND-CPA secure PKE, but it has a loss factor $q'$ for OW-CPA secure PKE in the ROM. However, all of their reduction bounds in the QROM suffer from a quartic loss, i.e., $\epsilon' \le q' \sqrt{q'\sqrt{\epsilon_{OW}} + q'^2 \delta}$ with an additional hash in~\cite{fo-transform-HHK17}. Later, the bound of $\text{FO}^{\not\bot}$ is improved as follows: $\epsilon' \le q'\sqrt{\epsilon_{OW}} + q' \sqrt{\delta}$ without additional hash in~\cite{fo-transform-qrom-JZC18}, $\epsilon' \le \sqrt{q' \epsilon_{CPA}} + q' \sqrt{\delta}$ with semi-classical oracles~\cite{fo-transform-AHU19} in~\cite{fo-transform-JZM19b}, $\epsilon' \le \sqrt{q' \epsilon_{CPA}} + q'^2{\delta}$ with double-sided OW2H lemma in~\cite{fo-transform-BHH+19}, and $\epsilon' \le q'^2 \epsilon_{CPA} + q'^2{\delta}$ with measure-rewind-measure technique in~\cite{fo-transform-KSS+20}. The bound of $\text{FO}_{m}^{\not\bot}$  is improved as follows: $\epsilon' \le q'\sqrt{\epsilon_{OW}} + q' \sqrt{\delta}$ without additional hash in~\cite{fo-transform-qrom-JZC18}, $\epsilon' \le \sqrt{q' \epsilon_{CPA}} + q'^2{\delta}$ with disjoint simulatability in~\cite{fo-transform-HKSU20}, $\epsilon' \le \sqrt{q' \epsilon_{CPA}} + q'^2{\delta}$ with prefix hashing in~\cite{fo-transform-prefix-hash-DHK+21}. The bound of $\text{FO}_{m}^{\bot}$ is improved as follows: $\epsilon' \le q'\sqrt{\epsilon_{OW}} + q' \sqrt{\delta}$ and $\epsilon' \le \sqrt{q' \epsilon_{CPA}} + q'\sqrt{\delta}$ with extra hash in~\cite{fo-transform-JZM19a}, $\epsilon' \le q'\sqrt{\epsilon_{OW}} + q'^2 \sqrt{\delta}$ without extra hash in~\cite{fo-transform-online-extractability-DFMS21}.

There also exist some transformations with tight reduction for deterministic PKE (DPKE) with disjoint simulatability and perfect correctness, for example, a variant of $\text{U}_{m}^{\not\bot}$ proposed in~\cite{fo-transform-SXY18}. In the case that the underlying PKE is non-deterministic, all known bounds are of the form $O(\sqrt{q' \epsilon_{CPA}})$ and $O(q' \sqrt{\epsilon_{OW}})$ as we introduce above, with the exception of~\cite{fo-transform-KSS+20}. The work~\cite{fo-transform-JZM21} shows that the measurement-based reduction involving no rewinding will inevitably incur a quadratic loss of the security in the QROM. In another word, as for the underlying PKE, the IND-CPA secure PKE has a tighter reduction bound than the OW-CPA secure PKE. It also significantly leads us to construct an IND-CPA secure PKE for tighter reduction bound of the resulting IND-CCA secure KEM.

Some discussions are presented here for comparing the reduction bounds of CTRU and CNTR and other NTRU-based KEM schemes.
Most of the existing  NTRU-based encryption schemes can only achieve OW-CPA security. NTRU-HRSS and SNTRU Prime construct the KEM schemes from OW-CPA DPKEs via $\text{U}_{m}^{\not\bot}$ variants. Although they can reach tight CCA reductions with extra assumptions in the (Q)ROM~\cite{ntru-nist-round3,ntru-prime-nist-round3},  there is a disadvantage that  some  extra computation is needed to recover the randomness in the decryption algorithms.

Determinism is a much stricter condition, thus some  NTRU-based PKEs prefer to be non-deterministic (i.e., randomized). NTTRU applies $\text{FO}_{m}^{\bot}$ to build an IND-CCA KEM from an OW-CPA randomized PKE~\cite{nttru-LS19}. According to~\cite{fo-transform-HHK17,fo-transform-online-extractability-DFMS21}, its IND-CCA reduction bounds are not-tight in both the ROM ($O(q'\epsilon_{OW})$) and the QROM ( $O(q' \sqrt{\epsilon_{OW}})$).

NTRU-C is the general form of $\text{NTRU-C}_{3457}^{768}$. NTRU-C uses a slightly different way that it first constructs an IND-CPA PKE from an OW-CPA NTRU-based PKE via $\text{ACWC}_0$ transformation~\cite{ntru-variant-eprint-1352-DHK21}, and then transforms it into an IND-CCA KEM via $\text{FO}_{m}^{\bot}$. Note that $\text{ACWC}_0$ brings two terms of ciphertexts, where the extra term of ciphertexts costs 32 bytes. The IND-CPA security of the resulting after-$\text{ACWC}_0$ PKE can be tightly reduced to the OW-CPA security of the underlying before-$\text{ACWC}_0$ PKE in the ROM. However, there is a quadratic loss advantage in the QROM, i.e., $\epsilon_{CPA} \le q' \sqrt{\epsilon_{OW}}$. In the ROM, the advantage of the adversary against IND-CCA security of KEM is tightly reduced to that of the adversary against IND-CPA security of after-$\text{ACWC}_0$ PKE, and consequently   is tightly reduced to that of the adversary against OW-CPA security of before-$\text{ACWC}_0$ PKE. However, in the QROM, there is no known direct reduction proof about $\text{FO}_{m}^{\bot}$ from IND-CPA PKE to IND-CCA KEM without additional hash. The reduction bound of $\text{FO}_{m}^{\bot}$ in the QROM in~\cite{fo-transform-online-extractability-DFMS21} only aims at the underlying OW-CPA PKE. Since the IND-CPA security implies OW-CPA security~\cite{fo-transform-HHK17}, the reduction bound of IND-CCA KEM to before-$\text{ACWC}_0$ OW-CPA PKE will suffer from the  quartic advantage loss in the QROM. That is, if the adversary has $\epsilon_{OW}$ advantage against the before-$\text{ACWC}_0$ OW-CPA PKE, then it has $O(q'^{1.5} \sqrt[4]{\epsilon_{OW}})$ advantage against the resulting IND-CCA KEM in the QROM. On the other hand, with an additional hash, a better bound of $\text{FO}_{m}^{\bot}$ for after-$\text{ACWC}_0$ IND-CPA PKE can be achieved, i.e., $O(\sqrt{q' \epsilon_{CPA}})$ advantage against the resulting IND-CCA KEM in the QROM~\cite{fo-transform-JZM19a} at the cost of  some extra ciphertext burden. $\text{ACWC}_0$ also has an effect on the efficiency, since an extra transformation from OW-CPA PKE to IND-CPA PKE is also relatively time-consuming.

Our CTRU and CNTR seem to be more simple, compact, efficient and memory-saving than other NTRU-based KEM schemes, along with a tight bound in the ROM and a tighter bound in the QROM for IND-CCA security. When compared to NTRU-HRSS, SNTRU Prime and NTRU-C, an obvious efficiency improvement of our CTRU and CNTR is due to the fact that there is no extra requirement of recovering randomness in decryption algorithm or reinforced transformation to obtain IND-CPA security. CTRU/CNTR.PKE can achieve IND-CPA security in the case that its ciphertext can be only represented by a single polynomial, without any extra ciphertext term like NTRU-C. Starting from our IND-CPA PKE to construct KEM with $\text{FO}_{ID(pk),m}^{\not\bot}$, the reduction bound of IND-CCA security is tightly reduced to IND-CPA security in the ROM ($\epsilon' \le O(\epsilon_{CPA})$, restated), so it is tightly reduced to the underlying hardness assumptions. We also have the known best bound in the QROM ($\epsilon' \le O(\sqrt{ q' \epsilon_{CPA})}$, restated) according to~\cite{fo-transform-prefix-hash-DHK+21}, which is  better than those in NTTRU and NTRU-C.

\subsection{CCA Security in Multi-User Setting}

We remark that, the work~\cite{fo-transform-prefix-hash-DHK+21} originally gives the multi-user/challenge IND-CCA reduction bound of $\text{FO}_{ID(pk),m}^{\not\bot}$ in the ROM and the QROM. We adapt the results from Theorem 3.1 and Theorem 3.2 in~\cite{fo-transform-prefix-hash-DHK+21} into the single-user/challenge setting of CTRU and CNTR, which is only for ease of fair comparisons as other KEM schemes only utilize single-user/challenge FO transformations. As CTRU/CNTR.PKE is IND-CPA secure,  another advantage of using $\text{FO}_{ID(pk),m}^{\not\bot}$ is that CTRU and CNTR can be improved to enjoy the multi-user/challenge IND-CCA security as well.
To address this issue, some adjustments are needed as follows. Unlike the single-user/challenge setting, the adversary (against the $n'$-user/$q_C$-challenge IND-CPA security of the underlying PKE) is given the public keys of $n'$ users, and is allowed to make at most $q_C$ challenge queries w.r.t.  the same challenge plaintext $m_b$ chosen by the challenger.
According to~\cite{fo-transform-prefix-hash-DHK+21}, based on the single-user/challenge IND-CPA security of the underlying PKE, the formal multi-user/challenge IND-CCA security of the resulting KEM is given in Theorem~\ref{thm-ntru-e8-kem-ind-cca-security-multi-user}.

\begin{theorem}[$n'$-user/$q_C$-challenge IND-CCA security in the ROM and the QROM~\cite{fo-transform-prefix-hash-DHK+21}]~\label{thm-ntru-e8-kem-ind-cca-security-multi-user}
	Following~\cite{fo-transform-prefix-hash-DHK+21}, we will use (or recall) the following terms in the concrete security statements.
	\begin{itemize}
		\item $n'$-user error probability $\delta(n')$~\cite{fo-transform-prefix-hash-DHK+21}.
		\item Min-entropy $\ell$~\cite{fo-transform-FO99} of $ID(pk)$, i.e., $\ell = H_{\infty}(ID(pk))$, where $(pk,sk)\leftarrow CTRU/CNTR.PKE.KeyGen$.
		\item Bit-length $\iota$ of the secret seed $z \in \{0,1\}^{\iota}$.
		\item Maximal number of (Q)RO queries $q_H$.		
		\item Maximal number of decapsulation queries $q_D$.
		\item Maximal number of challenge queries $q_C$.		
	\end{itemize}	
	For any (quantum) adversary $\mathsf{A}$ against the $(n',q_C)$-IND-CCA security of CTRU/CNTR.KEM, there exits a (quantum) adversary $\mathsf{B}$ against the $(n',q_C)$-IND-CPA security of CTRU/CNTR.PKE with roughly the same running time of $\mathsf{A}$, such that:
	\begin{itemize}
		\item In the ROM, it holds that $\textbf{Adv}_{\text{CTRU/CNTR.KEM}}^{(n',q_C)\text{-IND-CCA}}(\mathsf{A}) \le$
		$$ 2\left( \textbf{Adv}_{\text{CTRU/CNTR.PKE}}^{(n',q_C)\text{-IND-CPA}}(\mathsf{B})   + \frac{(q_H+q_C)q_C}{|\mathcal{M}|}  \right) + \frac{q_H}{2^{\iota}}  + (q_H+q_D) \delta(n') + \frac{n'^2}{2^{\ell}};$$
		\item In the QROM, it holds that $\textbf{Adv}_{\text{CTRU/CNTR.KEM}}^{(n',q_C)\text{-IND-CCA}}(\mathsf{A}) \le$		
		$$ 2 \sqrt{q_{HD} \textbf{Adv}_{\text{CTRU/CNTR.PKE}}^{(n',q_C)\text{-IND-CPA}}(\mathsf{B}) }+ 4 q_{HD}\sqrt{\frac{q_C \cdot n'}{|\mathcal{M}|}}+ 4(q_H+1)\sqrt{\frac{n'}{2^{\iota}}}
		+ 16 q_{HD}^{2}\delta(n') + \frac{q_C^2}{|\mathcal{M}|} +\frac{n'^2}{2^{\ell}},$$
		where $q_{HD}:=q_H + q_D+1$.
	\end{itemize}		
\end{theorem}

%
%
%
%

\section{S-unit attack}\label{app-sec-s-unit-attack}

Here we refer to~\cite{BT21} to briefly introduce S-unit attack.

S-unit attack begins with a nonzero $v \in I$ and outputs $v/u$, but now $u$ is allowed to range over a larger subset of $K^{*}$ , specifically the group of S-units.

Here $S$ is a finite set of places, a subset of the set $V$ mentioned above. There are two types of places:

\begin{itemize}
	\item The ``infinite places'' are labeled $1, 3, 5,\ldots, n-1$, except that for $n = 1$ there is one infinite place labeled 1. The entry at place $j$ in $\log \alpha$ is defined as $2 \log | \sigma_j (\alpha ) |$, except that the factor 2 is omitted for $n = 1$. The set of all infinite places is denoted $\infty$, and is required to be a subset of $S$.
	
	\item For each nonzero prime ideal $P$ of $R$, there is a ``finite place'' which is labeled as $P$. The entry at place $P$ in $\log \alpha$ is defined as $-(ord_P \alpha) \log |(R/P)|$, where $ord_P \alpha$ is the exponent of $P$ in the factorization of $\alpha$ as the  product of powers of prime ideals. There are many choices of $S$ here. It focuses on the following form of $S$: choose a parameter $y$, and take $P \in S $ if and only if $|(R/P)| \le y$.
\end{itemize}

The group $U_S$ of S-units of $K$ is, by definition, the set of elements $u \in K^{*}$such that the vector $\log u$ is supported on $S$, i.e., it is 0 at every place outside $S$. The S-unit lattice is the lattice $\log U_S$, which has rank $|S - 1|$.

Short $v/u$ again corresponds to short $\log v - \log u$, but it is required to ensure that $v/u \in I$, i.e., $ord_P (v/u) \ge ord_P I$ for each finite place $P$. This was automatic for unit attacks but is not automatic for general S-unit attacks. One thus wants to find a vector $\log u$ in the S-unit lattice $\log U_S$ that is close to $\log v$ in the following sense: $\log u$ is close to $\log v $ at the infinite places, and $ord_P u$ is close to but no greater than $ord_Pv - ord_P I$.

As for closeness, as a preliminary step, if $ord_P v < ord_P I$ for some $P$, update $v$ by multiplying it by a generator of $P \hat{P}$ (or, if possible, of $P$) as explained above, and repeat this step. Then $v \in I$. Next, if some $u$ in the list has $ v/u$ shorter than $v$ and $v/u \in I,$ replace $v$ with $v/u$, and repeat this step. Output the final $v$.

As an extreme case, if $S = \infty$ (the smallest possible choice, not including any $ P$), then $U_S = R^{*}$ : the S-units of $K$ are the units of $ R$, the S-unit lattice is the unit lattice, and S-unit attacks are the same as unit attacks. Extending $S$ to include more and more prime ideals $P$ gives S-unit attacks the ability to modify more and more places in $\log v$.

\section{Generalization and More Variants}

Finally, to demonstrate the flexibility of our framework, we present and discuss some generalization approaches and more variants. The following approaches are applicable to both CTRU and CNTR.

\subsection{Compressing the Public Key $h$}

In general, let $q_1\leq q$ be an integer, we  set the public key to be $\hat{h}=\big\lfloor \frac{q_1}{q}  h \big\rceil \in \mathcal{R}_{q_1}$ in the KeyGen. This not only shortens the public key size, but also can strengthen the security of the NTRU assumption in general. Then, there are two approaches to deal with this change in PKE.Enc. 

\begin{itemize}
	\item $\sigma=\hat{h}r \in \mathcal{R}_{q_1}$, and now PolyEncode needs to work in $\mathcal{R}_{q_1}$ rather than $\mathcal{R}_q$ (i.e., the parameter $q$ is replaced with $q_1$). 
	
	\item $\sigma=\big\lfloor \frac{q}{q_1} \hat{h} \big\rceil r \in \mathcal{R}_{q}$. That is, we lift $\hat{h}$ from $\mathcal{R}_{q_1}$ to $\mathcal{R}_{q}$. With this approach, PolyEncode remains unchanged.
\end{itemize}

These approaches can reduce the size of public key, but at the cost of larger error probability or lower security (as we may need to narrow the space of secret polynomials for reducing error probability). With experiments, when $q_1=2^{11}$ (i.e., cutting off one bit from each dimension) we still can  achieve reasonable balance between security and performance.

\subsection{Masking the Public Key $h$}

Similarly, we would like also  to strengthen the NTRU assumption, by setting $h=g/f+x$, where $x$ is an $n$-dimension small noise polynomial with each coefficient typically taken from $B_1$ or $U_1$, i.e., the uniform distribution over $\{0,\pm 1\}$. In this case, $h$ is analogous to an RLWE sample, except that $f^{-1}$ is not publicly accessible. Intuitively, it makes the NTRU problem harder than its standard form. 
In this case, one extra error term $xr$ will be introduced. Thanks to the powerful error correction ability of the $\text{E}_8$ lattice code, our experiments show that we can still achieve good balance between security and performance, with $x$ taken from $B_1$ or $U_1$. Note that the above approach to compressing the public key can also strengthen the hardness of the NTRU problem.

\subsection{More Possibilities of PolyEncode and PolyDecode}

We choose the $\text{E}_8$ lattice code within our framework because: (1) the error correction ability of the $\text{E}_8$ lattice code is powerful and almost optimal; and (2) it is simple, very efficient, and  well fits our framework combining NTRU and RLWE/RLWR. However, in general, we can use other error correction codes (ECC) within our framework. Also, an extreme choice is to not use any extra ECC mechanism, i.e., letting $\text{PolyEncode}(m)=\frac{q}{2}m$. The corresponding decryption process is $\text{PolyDecode}( c f \bmod^{\pm} q_2 )=\lfloor \frac{2}{q_2} ( c f \bmod^{\pm} q_2 ) \rceil \bmod 2$.

\subsection{More Possibilities of the Underlying Rings}

The modulus $q$ is set to be a prime number that allows efficient NTT algorithms over $\mathbb{Z}_q[x]/(x^{n}-x^{n/2}+1)$ in this paper. One can also choose power-of-two $q$ for the flexibility of parameter selection. Upon setting power-of-two $q$ and $q_2$, the distribution $\chi$ in Theorem~\ref{lemma-ntru-e8-lwr-correctness-analysis} can be simplified as follows: Sample $u \xleftarrow{\$} [-\frac{q}{2q_2}, \frac{q}{2q_2} )\cap \mathbb{Z}$ and output $-\frac{q_2}{q} u$. However, in this case the computation of polynomial division is slightly slower, since NTT is invalid and we need to turn to other less efficient algorithms to compute polynomial divisions.

We can naturally generalize the underlying polynomial rings of the form $\mathbb{Z}_q[x]/(x^{n}-x^{n/2}+1)$ to power-of-two cyclotomic rings of the form $\mathbb{Z}_q[x]/(x^{n}+1)$. To consider the dimension $n$, one can pick $n=512,1024$ for NIST recommended security levels I and V. When the modulus $q$ is NTT-friendly (i.e., $2^e | (q-1)$ for some integer $e$ in this case), the efficient polynomial multiplications and divisions are possible via some NTT algorithms (or their variants). For the flexibility of ring selection, one can also consider the rings of the form $\mathbb{Z}_q[x]/(x^{n}-x-1)$ w.r.t. prime $n$ and $q$ like those in NTRU Prime~\cite{ntru-prime-BCLV17,ntru-prime-nist-round3}.

\subsection{Variants without RLWE or RLWR}

If we insist on a purely NTRU-based KEM, we can simply set $c:=\sigma +\big\lfloor\text{PolyEncode}(m) \big\rceil \bmod q$ where we can set $\sigma:=hr$. In this case, the resulting KEM scheme is OW-CPA secure based on the NTRU assumption. With this variant, at about the same error probabilities of CTRU and CNTR, we can choose much larger ranges for the secret key and the ephemeral secrecy $r$ leading to stronger NTRU hardness.

%
%
%
%

\bibliographystyle{IEEEtran}
\bibliography{CTRU-ref}

\vfill

\end{document}